\documentclass[runningheads]{llncs}

\usepackage{enumitem}
\usepackage{amssymb}
\usepackage{mathtools}
\usepackage{amsmath}

\usepackage{amsthm}
\usepackage{wasysym}
\usepackage{booktabs}
\usepackage{tikz}
\usetikzlibrary{positioning,fit,calc}
\usetikzlibrary{shapes,arrows}

\let\next\nextt

\makeatletter
\pgfdeclareshape{datastore}{
  \inheritsavedanchors[from=rectangle]
  \inheritanchorborder[from=rectangle]
  \inheritanchor[from=rectangle]{center}
  \inheritanchor[from=rectangle]{base}
  \inheritanchor[from=rectangle]{north}
  \inheritanchor[from=rectangle]{north east}
  \inheritanchor[from=rectangle]{east}
  \inheritanchor[from=rectangle]{south east}
  \inheritanchor[from=rectangle]{south}
  \inheritanchor[from=rectangle]{south west}
  \inheritanchor[from=rectangle]{west}
  \inheritanchor[from=rectangle]{north west}
  \backgroundpath{
    \southwest \pgf@xa=\pgf@x \pgf@ya=\pgf@y
    \northeast \pgf@xb=\pgf@x \pgf@yb=\pgf@y
    \pgfpathmoveto{\pgfpoint{\pgf@xa}{\pgf@ya}}
    \pgfpathlineto{\pgfpoint{\pgf@xb}{\pgf@ya}}
    \pgfpathmoveto{\pgfpoint{\pgf@xa}{\pgf@yb}}
    \pgfpathlineto{\pgfpoint{\pgf@xb}{\pgf@yb}}
 }
}
\makeatother

\newcommand{\defequals}{\hspace{0.3em}\mathrel{\overset{\makebox[0pt]{\mbox{\normalfont\tiny\sffamily def}}}{=}}\hspace{0.3em}}
\newcommand{\commentout}[1]{}
\newtheorem{algorithm}{{\bf Algorithm}}
\newtheorem{assumption}{{\bf Assumption}}
\newtheorem*{acknowledgments}{{\bf Acknowledgments}}

\DeclareMathAlphabet{\mathcal}{OMS}{cmsy}{m}{n}
\newcommand{\until}{{\mathcal U}}
\newcommand{\release}{{\mathcal R}}
\newcommand{\eventually}{\Diamond}
\newcommand{\always}{\Box}
\newcommand{\next}{\Circle}

\newcommand{\I}{{\mathbb I}}
\newcommand{\G}{{\mathbb G}}
\newcommand{\B}{{\mathbb B}}
\newcommand{\C}{{\mathbb C}}
\newcommand{\T}{{\mathbb T}}
\newcommand{\YT}{{\mathbb Y}{\mathbb T}}
\newcommand{\CT}{{\mathbb C}{\mathbb T}}
\newcommand{\DT}{{\mathbb D}{\mathbb T}}
\newcommand{\ST}{{\mathbb S}{\mathbb T}}
\newcommand{\WT}{{\mathbb W}{\mathbb T}}

\newcommand{\Cognition}{\Omega}
\newcommand{\Rationality}{\Gamma}
\newcommand{\Guard}{\Lambda}
\newcommand{\Strategy}{\Sigma}
\newcommand{\CStrategy}{\Pi}

\newcommand{\trans}[1]{{\longrightarrow}_{T}^{#1}}
\newcommand{\ctrans}[1]{{\longrightarrow}_C^{#1}}

\newcommand{\dist}[1]{{\mathcal D}(#1)}
\newcommand{\powerset}[1]{{\mathcal P}(#1)}

\newcommand{\cM}{\mathcal{M}}
\newcommand{\states}{S}
\newcommand{\sinit}{S_\mathrm{init}}
\newcommand{\initdist}{\mu_0}
\newcommand{\belief}{b}
\newcommand{\beliefsinit}{\mathcal{B}_\mathrm{init}}
\newcommand{\beliefinit}{b_\mathrm{init}}
\newcommand{\astrat}[1]{\sigma_{#1}}
\newcommand{\cstrat}[1]{\pi_{#1}}
\newcommand{\probm}[1]{\mathrm{Pr}_{#1}} 
\newcommand{\probs}[2]{\mathrm{Pr}_{#1}^{#2}}
\newcommand{\probf}[2]{Prob_{#1}^{#2}} 
\newcommand{\sspace}[1]{\Omega_{#1}} 
\newcommand{\measurable}[1]{\mathcal{F}_{#1}}
\newcommand{\obseq}[1]{\sim_{#1}^{o}}
\newcommand{\tp}[1]{tp_{#1}}

\newcommand{\SMG}{SMG}
\newcommand{\SMGC}{SMG$_\Cognition$}
\newcommand{\POSMG}{POSMG}
\newcommand{\AutoSMAS}{ASMAS}

\newcommand{\Ags}{Ags}
\newcommand{\Act}[1]{Act_{#1}}
\newcommand{\LAct}[1]{LAct_{#1}}

\newcommand{\Obs}[1]{O_{#1}}
\newcommand{\obs}[1]{obs_{#1}}
\newcommand{\pathobs}[1]{OPath_{#1}}
\newcommand{\dmap}[1]{\delta_{#1}}
\newcommand{\supp}{supp}
\newcommand{\Bel}[1]{\beta_{#1}}
\newcommand{\TBel}[1]{T_{#1}^{\beta}}
\newcommand{\Goal}[1]{Goal_{#1}}

\newcommand{\lgoal}[1]{\omega^g_{#1}}
\newcommand{\goal}[1]{\pi^g_{#1}}
\newcommand{\goalGd}[1]{\lambda_{#1}^{g}}
\newcommand{\Intn}[1]{Int_{#1}}
\newcommand{\lintn}[1]{\omega^i_{#1}}
\newcommand{\intn}[1]{\pi^i_{#1}}
\newcommand{\intnGd}[1]{\lambda_{#1}^{i}}
\newcommand{\pref}[1]{p_{#1}}
\newcommand{\gpref}[2]{gp_{{#1},{#2}}}
\newcommand{\ipref}[2]{ip_{{#1},{#2}}}
\newcommand{\Exp}{E}

\newcommand{\legal}{legal}
\newcommand{\possible}{possible}

\newcommand{\rimp}{\Rightarrow}

\newcommand{\apath}[1]{\mathrm{Path}_{#1}}
\newcommand{\fpath}[1]{\mathrm{FPath}_{#1}}
\newcommand{\infpath}[1]{\mathrm{IPath}_{#1}}
\newcommand{\last}{\mathrm{last}}

\makeatletter
\DeclareRobustCommand*\cal{\@fontswitch\relax\mathcal}
\makeatother
\newcommand{\nat}{{\tt N}}
\newcommand{\R}{{\cal R}}

\newcommand{\PRTLS}{PRTL$^*$}
\newcommand{\PCTLS}{PCTL$^*$}
\newcommand{\PCTL}{PCTL}

\newcommand{\CTLS}{CTL$^*$}

\newcommand{\BPRTLS}{BPRTL$^*$}
\newcommand{\PRTLSS}{PRTL$_1^*$}
\newcommand{\PQRTLSS}{PQRTL$_1^*$}

\newcommand{\PCTLSM}{PCTL$^*_\Cognition$}
\newcommand{\PCTLM}{PCTL$_\Cognition$}
\newcommand{\PRATLS}{PRATL$^*$}

\newcommand{\PRTLSf}{PRTL^*}

\newcommand{\prob}[2]{{\tt P}^{#1}#2}

\newcommand{\current}{{\tt c}}
\newcommand{\history}{{\tt h}}

\newcommand{\Lang}{{\cal L}}
\newcommand{\synth}[1]{{eval_{#1}}}
\newcommand{\gsynth}[1]{{eval_{#1}^g}}
\newcommand{\isynth}[1]{{eval_{#1}^i}}

\newcommand{\trustgame}{{\cal G}}
\newcommand{\staghunt}{{\cal H}}
\newcommand{\gs}{gs}
\newcommand{\is}{is}

\newcommand{\sep}{\discretionary{}{}{}}
\newcommand{\indStrat}{\astrat{ind}}

\newcommand{\cylinder}{\mathrm{Cyl}}
\newcommand{\literal}[1]{\mathsf{#1}}
\newcommand{\ap}{\mathit{AP}}

\newcommand{\con}{1}
\newcommand{\adv}{2}

\newcommand{\co}{{\tt co}}
\newcommand{\de}{{\tt de}}

\newcommand{\dde}{\mathrm{de}}
\newcommand{\pde}{\mathrm{pde}}
\newcommand{\pdeco}{\mathrm{deco}}
\newcommand{\pco}{\mathrm{pco}}
\newcommand{\dco}{\mathrm{co}}

\newcommand{\tftde}{\mathrm{tftde}}
\newcommand{\adep}{\mathrm{adep}}
\newcommand{\tft}{\mathrm{tft}}
\newcommand{\aco}{\mathrm{aco}}
\newcommand{\tftco}{\mathrm{tftco}}
\newcommand{\acop}{\mathrm{acop}}

\newcommand{\todo}[1]{{\color{red} #1}}

\begin{document}


\title{Reasoning about Cognitive Trust in Stochastic Multiagent Systems}

\author{Xiaowei Huang \and
Marta Kwiatkowska \and
Maciej Olejnik}

\authorrunning{X. Huang, M. Kwiatkowska, M. Olejnik}

\institute{University of Oxford, Oxford, UK \\
\email{xiaowei.huang@live.com \\ marta.kwiatkowska@cs.ox.ac.uk \\ maciej.olejnik@cs.ox.ac.uk}}

\maketitle

\begin{abstract}

We consider the setting of stochastic multiagent systems modelled as stochastic multiplayer games and formulate an automated verification framework for quantifying and reasoning about agents' trust. To capture human trust, we work with a cognitive notion of trust defined as a subjective evaluation that agent $A$ makes about agent $B$'s ability to complete a task, which in turn may lead to a decision by $A$ to rely on $B$. We propose a probabilistic rational temporal logic \PRTLS, which extends the probabilistic computation tree logic \PCTLS\ with reasoning about mental attitudes (beliefs, goals and intentions), and includes novel operators that can express concepts of social trust such as competence, disposition and dependence. The logic can express, for example, that ``agent $A$ will eventually trust agent $B$ with probability at least $p$ that B will behave in a way that ensures the successful completion of a given task''.  We study the complexity of the automated verification problem and, while the general problem is undecidable, we identify restrictions on the logic and the system that result in decidable, or even tractable, subproblems. 
\end{abstract}


\section{Introduction}
Mobile autonomous robots are rapidly entering the fabric of our society, to mention driverless cars and home assistive robots.
Since robots are expected to work with or alongside humans in our society, they need to form partnerships with humans, as well as other robots, understand the social context, and behave, and be seen to behave, according to the norms of that context. Human partnerships such as cooperation are based on \textit{trust}, which is influenced by a range of subjective factors that include subjective preferences and experience. As the degree of autonomy of mobile robots increases and the nature of partnerships becomes more complex, to mention shared autonomy, understanding and reasoning about social trust and the role it plays in decisions whether to rely on autonomous systems is of paramount importance. A pertinent example is the recent Tesla fatal car accident while on autopilot mode \cite{TeslaCrash1:BBC:2016}, which is a result of over-reliance (``overtrust'') by the driver, likely influenced through his personal motivation and preferences.

Trust is a complex notion, 
viewed as a belief, attitude, intention or behaviour, and is most generally understood as \textit{a subjective evaluation of a truster on a trustee about something in particular}, e.g., the completion of a task \cite{Hardin:Trust:2002}.  A classical definition from organisation theory \cite{Mayer:AMR:1995} defines trust as \textit{the willingness of a party to be vulnerable to the actions of another party based on the expectation that the other will perform a particular action important to the trustor, irrespective of the ability to monitor or control that party}. The importance of being able to correctly evaluate and calibrate trust to guide reliance on automation was recognised in \cite{Lee:HF:2004}.
Trust (and trustworthiness) have also been actively studied in many application contexts such as security~\cite{KFJ2001} and e-commerce~\cite{CTY2003}. However, in this paper we are interested in trust that governs social relationships between humans and autonomous systems, and to this end consider {\it cognitive trust} that captures the human notion of trust. By understanding how human trust in an autonomous system evolves, and being able to quantify it and reason about it, we can offer guidance for selecting an appropriate level of reliance on autonomy.

The goal of this paper is therefore to develop foundations for \emph{automated, quantitative reasoning} about cognitive trust between (human and robotic) agents, which can be employed to support decision making in dynamic, uncertain environments. The underlying model is that of \emph{multiagent systems}, where agents are autonomous and endowed with individual goals and preferences in the style of BDI logic~\cite{RG1991}.
To capture uncertainty, we work in the setting of \emph{stochastic} multiagent systems, represented concretely in terms of \emph{concurrent stochastic multiplayer games}, where stochasticity can be used to model, e.g., component failure or environmental uncertainty. %
This also allows us to represent agent \emph{beliefs} probabilistically.
Inspired by the concepts of social trust in~\cite{FC2001}, we formulate
a probabilistic rational temporal logic \PRTLS\ as an extension of the probabilistic temporal logic \PCTLS~\cite{HJ1994} with 
\emph{cognitive} aspects.
\PCTLS\ allows one to express \emph{temporal} properties pertaining to system execution, for example ``with probability $p$, agent A will eventually complete a given task''.
\PRTLS\ includes, in addition, mental attitude operators (belief, goal, intention and capability), together with
a collection of novel \emph{trust} operators (competence, disposition and dependence), in turn expressed using beliefs.
The logic \PRTLS\ is able to express properties such as ``agent A will eventually trust agent B with probability at least $p$ that B will behave in a way that ensures the successful completion of a given task'' that are informally defined in~\cite{FC2001}.
%
\PRTLS\ is interpreted over a stochastic multiagent system, where the cognitive reasoning processes for each agent can be modelled based on a \emph{cognitive mechanism} that describes his/her mental state (a set of goals and an intention, referred to as pro-attitudes) and subjective preferences.

Since we wish to model dynamic evolution of beliefs and trust, the mechanisms are history-dependent, and thus the underlying semantics is an infinite branching structure, resulting in undecidability of the general model checking problem for \PRTLS. In addition, there are two types of nondeterministic choices available to the agents, those made along the \emph{temporal} or the \emph{cognitive} dimension. By convention, the temporal nondeterminism is resolved using the concept of adversaries and quantifying over them to obtain a fully probabilistic system~\cite{HT1993}, as is usual for models combining probability and nondeterminism. We use a similar approach for the cognitive dimension, instead of the classical accessibility relation employed in logics for agency, and resolve cognitive nondeterminism by \emph{preference functions}, given as probability distributions 
that model subjective knowledge about other agents. 
Also, in contrast to defining beliefs in terms of
knowledge~\cite{FHMVbook} and
probabilistic knowledge~\cite{HT1993} operators, which are based solely on agents' (partial) observations, we
additionally allow
agents' cognitive changes and subjective
preferences to influence their belief.

This paper makes the following original contributions.
\begin{itemize}
\item We introduce autonomous stochastic multiagent systems as an extension of stochastic multiplayer games with a cognitive mechanism for each agent (a set of goals and intentions).
\item We provide a mechanism for reasoning about agent's cognitive states based on preference functions that enables a sound formulation of probabilistic beliefs.
\item We formalise a collection of trust operators (competence, disposition and dependence) informally introduced in~\cite{FC2001} in terms of probabilistic beliefs.
\item We formulate a novel probabilistic rational temporal logic \PRTLS\ that extends the logic \PCTLS~\cite{HJ1994} with mental attitude and  trust operators.
\item We study the complexity of the automated verification problem for \PRTLS\ and, while the general problem is undecidable, we identify restrictions on the logic and the system that result in decidable, or even tractable, subproblems.
\end{itemize}

The structure of the paper is as follows. Section~\ref{sec:related} gives an overview of related work, and in Section~\ref{sec:trustidea} we discuss the concept of cognitive trust. Section~\ref{sec:smas} presents stochastic multiplayer games and strategic reasoning on them. Section~\ref{sec:smgc} introduces autonomous stochastic multiagent systems, an extension of stochastic multiplayer games with a cognitive mechanism. Section~\ref{sec:preferencefuncs} introduces preference functions and derives the semantics of the subjective, probabilistic belief operator. Section~\ref{sec:beliefsupdate} defines how beliefs vary with respect to agents' observations. In Section~\ref{sec:pref} we define trust operators and the logic \PRTLS. We consider the interactions of beliefs and pro-attitudes in Section~\ref{sec:proupdate} via pro-attitude synthesis. Section~\ref{sec:overviewcomplexity} gives the undecidable complexity result for the general \PRTLS\ model checking problem and Section~\ref{sec:decidable} presents several decidable logic fragments. 
We conclude the paper in Section~\ref{sec:concl}. 

A preliminary version of this work appeared as~\cite{HK2017}.
This extended version includes detailed derivations of the concepts, illustrative examples and full proofs of the complexity results omitted from~\cite{HK2017}.

\section{Related Work}\label{sec:related}

The notion of trust has been widely studied in management, psychology, philosophy and economics
(see \cite{Lahijanian:AAAI:2016} for an overview).
Recently, the importance of trust in human-robot cooperation was highlighted in~\cite{Kuipers:2018:WTR:3190347.3173087}.
Trust in the context of human-technology relationships can be roughly classified into three categories: \textit{credentials-based}, \textit{experience-based}, and \textit{cognitive trust}.  Credentials-based trust is used mainly in security, where a user must supply credentials in order to gain access.
Experience-based trust, which includes reputation-based trust in peer-to-peer and e-commerce applications,
involves online evaluation of a trust value for an agent informed by experiences of interaction with that agent.
A formal foundation for quantitative reputation-based trust has been proposed in \cite{Krukow:PTRSL:2008}.
In contrast, we focus on (quantitative) cognitive trust, which captures the social (human) notion of trust and, in particular, trust-based decisions between humans and robots.  The cognitive theory of \cite{FC2001}, itself founded on organisational trust of \cite{Mayer:AMR:1995}, provides an intuitive definition of complex trust notions but lacks rigorous semantics. Several papers, e.g., \cite{Meyer:AI:1999,Josang:IJUFKS:2001,Herzig:IGPL:2010,Herzig:GC:2013}, have formalised the theory of  \cite{FC2001} using modal logic, but none are quantitative and automatic verification is not considered.
Of relevance are recent approaches \cite{Sweet:AIAA:2016,Setter:ACC:2016}  that model the evolution of trust in human-robot interactions as a dynamical system; instead, our formalism supports evolution of trust through events and agent interactions.

A number of logic frameworks have been proposed that develop the theory of human decisions~\cite{Bratman1987} for artificial agents, see \cite{MBH2014} for a recent overview. The main focus has been on studying the relationships between modalities with various axiomatic systems, but their amenability to automatic verification is arguable because of a complex underlying possible world semantics, to mention the sub-tree relation of BDI logic~\cite{RG1991}.
The only attempt at model checking such logics \cite{Schild2000} ignores the internal structure of the possible worlds to enable a reduction
to temporal logic model checking.

The distinctive aspects of our work is thus a \textit{quantitative} formalisation of \textit{cognitive} trust in terms of \textit{probabilistic} temporal logic, based on a probabilistic notion of \textit{belief}, together with algorithmic complexity of the corresponding \textit{model checking} problem.

\section{Cognitive Theory of Social Trust}\label{sec:trustidea}

In the context of automation, trust is understood as delegation of responsibility for actions to the autonomous system and willingness to accept risk (possible harm) and uncertainty. The decision to delegate is based on a \textit{subjective} evaluation of the system's capabilities for a particular task, informed by factors such as past experience, social norms and individual preferences. Moreover, trust is a dynamic concept which evolves over time, influenced by events and past experience. The cognitive processes underpinning trust are captured in the influential theory of social trust by \cite{FC2001}, which is particularly appropriate for human-robot relationships and serves as an inspiration for this work.

The theory of \cite{FC2001} views trust as a complex \textit{mental attitude} that is relative to a set of goals and expressed in terms of \emph{beliefs}, which in turn influence decisions about agent's future behaviour. They consider agent $A$'s trust in agent $B$ for a specific goal $\psi$ (goals may be divided into tasks), and distinguish the following core concepts: \textit{competence} trust, where $A$ believes that $B$ is able to perform $\psi$, and
\textit{disposition} trust, where $A$ believes that $B$ is willing to perform $\psi$. The decision to delegate or rely on $B$ involves a complex notion of trust called
\textit{dependence}: $A$ believes that $B$ needs, depends, or is at least better off to rely on $B$ to achieve $\psi$, which has two forms, \textit{strong} ($A$ needs or depends on $B$) and \textit{weak} (for $A$, it is better to rely than not to rely on $B$).  
\cite{FC2001} also identify \textit{fulfilment} belief arising in the truster's mental state, which we do not consider.

We therefore work in a stochastic setting (to represent uncertainty), aiming to \textit{quantify} belief probabilistically and express trust as a \textit{subjective, belief-weighted expectation}, informally understood as a degree of trust.

\section{Stochastic Multiagent Systems and Temporal Reasoning}\label{sec:smas}

A multiagent system $\cM$ comprises a set of agents (humans, robots, components, processes, etc.) running in an environment~\cite{FHMVbook}. To capture random events such as failure and environmental uncertainty that are known to influence trust, we work with \emph{stochastic} multiagent systems, concretely represented using \emph{concurrent stochastic multiplayer games}, where each agent corresponds to a player.
Taking such models as a starting point, in this section we gradually extend them to autonomous multiagent systems that support cognitive reasoning.

\subsection{Stochastic Multiplayer Games}\label{subsec:smg}
Given a finite set $S$, we denote by $\dist{S}$ the set of probability distributions on $S$ and by $\powerset{S}$ the power set of $S$. Given a probability distribution $\delta$ over a set $S$, we denote by $\supp(\delta)$ the support of $\delta$, i.e., the set of elements of $S$ which have positive probability. 
We call $\delta$ a Dirac distribution if $\delta(s) = 1$ for some $s \in S$.

We now introduce concurrent stochastic games, in which several players repeatedly make choices simultaneously to determine the next state of the game. 

\begin{definition}\label{def:smg}
A \emph{stochastic multiplayer game (\SMG)} 
is a tuple $\cM = (\Ags, \sep\states,\sep\sinit,\sep\{\Act{A}\}_{A\in\Ags}, \sep T, \sep L)$, where:
\begin{itemize}
\item
$\Ags=\{1,...,n\}$ is a finite set of players called agents, ranged over by $A$, $B$, ...
\item
$\states$ is a finite set of states,
\item
$\sinit\subseteq \states$ is a set of initial states,
\item
$\Act{A}$ is a finite set of actions for the agent $A$,
\item
$T:S\times \Act{} \rightarrow \dist{S}$ is a (partial) probabilistic transition function, where $\Act{}= \times_{A\in\Ags}\Act{A}$, 
\item
$ L:S\rightarrow \powerset{AP}$ is a labelling function mapping each state to a set of atomic propositions taken from a set $AP$.
\end{itemize}
\end{definition}
We assume that each (global) state $s \in \states$ of the system includes a local state of each agent and an (optional) environment state. In every state of the game, each player $A \in \Ags$ selects a local action $a_A\in\Act{A}$ independently and the next state of the game is chosen by the environment according to the probability distribution $T(s,a)$ where $a \in \Act{}$ 
is the joint action. In other words, the probability of transitioning from state $s$ to state $s'$ when action $a$ is taken is $T(s,a)(s')$.
In this paper we will usually omit the environment part of the state.

In concurrent games players perform their local actions simultaneously. \emph{Turn-based} games are a restricted class of SMGs whose states are partitioned into subsets, each of which is controlled by a single agent, meaning that only that agent can perform an action in any state in the partition. 
Turn-based games can be simulated by concurrent games by requiring that, in states controlled by agent $A$, $A$ performs an action $a_A \in \Act{A}$ and the other agents perform a distinguished silent action $\bot$.

Let $a_A$ denote agent $A$'s action in the joint action $a\in\Act{}$. We let $\Act{}(s)=\{ a \in \Act{} \mid T(s,a) \mathrm{\;is\;defined} \}$ be the set of \textit{valid} joint actions in state $s$, and $\Act{A}(s)=\{a_A~|~a\in\Act{}(s)\}$ be the set of valid actions in state $s$ for agent $A$. 
$T$ is called \emph{serial} (or total) if, for any state $s$ and joint action $a\in \Act{}$, $T(s,a)$ is defined. We often write $s \trans{a} s'$ for a transition from $s$ to $s'$ via action $a$, provided that $T(s,a)(s') > 0$.

\paragraph{\bf Paths}
A \textit{path} $\rho$ is a finite or infinite sequence of states $s_0s_1s_2...$ induced from the transition probability function $T$, i.e., satisfying $T(s_k,a)(s_{k+1}) > 0$ for some $a \in \Act{}$, for all $k \geq 0$.
Paths generated by $T$ are viewed as occurring in the \emph{temporal} dimension.
We denote the set of finite (resp. infinite) paths of $\cM$ starting in $s$ by $\fpath{}^\cM(s)$ (resp. $\infpath{}^\cM(s)$), and the set of paths starting from any state by $\fpath{}^\cM$ (resp. $\infpath{}^\cM$). We may omit $\cM$ if clear from the context. For any path $\rho$ we write $\rho(k)$ for its $(k+1)$-th state, $\rho[0..n]$ for the prefix $s_0...s_n$, and $\rho[n..\infty]$ for the suffix $s_ns_{n+1}...$ when $\rho$ is infinite. If $\rho$ is finite then we write $last(\rho)$ for its last state and $|\rho|$ for its \textit{length}, i.e., the number of states in $\rho$. 
Given two paths $\rho=s_0...s_n$ and $\rho'=s'_0...s'_m$, we write $\rho\cdot\rho'=s_0...s_ns'_1...s'_m$ when $s_n=s'_0$ for their concatenation  by an overlapping state, and $\rho\rho'=s_0...s_ns'_0s'_1...s'_m$ for the regular concatenation.



\paragraph{\bf Strategies}
A (history-dependent and stochastic) \textit{action strategy} $\astrat{A}$ of agent $A \in \Ags$ in an \SMG\ $\cM$ is a function $\astrat{A} : \fpath{}^\cM \rightarrow \dist{\Act{A}}$, such that for all $a_A \in \Act{A}$ and finite paths $\rho$ it holds that $\astrat{A}(\rho)(a_A) > 0$ only if $a_A \in \Act{A}(\last(\rho))$. We call a strategy $\astrat{}$ \emph{pure} if $\astrat{}(\rho)$ is a Dirac distribution for any $\rho \in \fpath{}^\cM$. A \emph{strategy profile} $\astrat{D}$ for a set $D$ of agents is a vector of action strategies $\times_{A \in D}\astrat{A}$, one for each agent $A\in D$. We let $\Strategy_A$ be the set of agent $A$'s strategies, $\Strategy_D$ be the set of strategy profiles for the set of agents $D$, and $\Strategy$ be the set of strategy profiles for all agents.


\paragraph{\bf Probability space}
In order to reason formally about a given \SMG\ $\cM$ we need to quantify the probabilities of different paths being taken. We therefore define a probability space over the set of infinite paths $\infpath{}^\cM(s_0)$ starting in a given state $s_0 \in S$, adapting the standard construction from~\cite{Kemeny1967}. Our probability measure is based on the function assigning probability to a given finite path $\rho = s_0...s_n$ under strategy $\astrat{} \in \Strategy$, defined as $\probs{\astrat{}}{}(\rho) = \prod_{i = 0}^{n-1}\sum_{a \in \Act{}} \astrat{}(\rho[0..k])(a) \cdot T(s_k,a)(s_{k+1})$. To define measurable sets, for a path $\rho$ we let $\cylinder_{\rho}$ be a basic cylinder, which is a set of all infinite paths starting with $\rho$. We then set $\measurable{s}^{\cM}$ to be the smallest $\sigma$-algebra generated by the basic cylinders $\{\cylinder_{\rho} \mid \rho \in \fpath{}^{\cM}(s)\}$ and $\probs{\astrat{}}{\cM}$ to be the unique measure on the set of infinite paths $\infpath{}^\cM(s)$ such that $\probs{\astrat{}}{\cM}(\cylinder_{\rho}) = \probs{\astrat{}}{}(\rho)$. It then follows that $(\infpath{}^{\cM}(s),\measurable{s}^{\cM},\probs{\astrat{}}{\cM})$ is a probability space~\cite{Kemeny1967}.


\begin{example}\label{example:trustgame} 
We consider a simple (one shot) trust game from \cite{Kuipers2016}, in which there are two agents, Alice and Bob. At the beginning, Alice has 10 dollars and Bob has 5 dollars. If Alice does nothing, then everyone keeps what they have. If Alice invests her money with Bob, then Bob can turn the 15 dollars into 40 dollars. After having the investment yield, Bob can decide whether to share the 40 dollars with Alice. If so, each will have 20 dollars. Otherwise, Alice will lose her money and Bob gets 40 dollars.

\begin{table}
\begin{center}
\caption{Payoff of a simple trust game}
\label{tab:trust}
\begin{tabular}{r|cc}
\multicolumn{1}{c}{} & \multicolumn{1}{c}{share}  & \multicolumn{1}{c}{keep} \\\cline{2-3}
invest & $(20,20)$ & $(0,40)$ \\ 
withhold & $(10,5)$ & $(10,5)$ \\ 
\end{tabular}
\end{center}
\end{table}

\commentout{
\begin{table}
\caption{Payoff of a simple trust game}
\label{tab:trust}
\begin{center}
\begin{tabular}{c|cc}
\toprule
 & share & keep \\
\midrule
invest & (20,20) & (0,40)\\
withhold & (10,5) & (10,5) \\
\bottomrule
\end{tabular}
\end{center}
\end{table}
}

For the simple trust game, the payoffs of the agents are shown in Table~\ref{tab:trust}. The game has a Nash equilibrium of Alice withholding her money and Bob keeping the investment yield. This equilibrium discourages collaboration between agents and has not been confirmed empirically under the standard economic assumptions of pure self-interest~\cite{Berg-trust}. 

To illustrate our methods, we construct a stochastic multiplayer game $\trustgame$ with $\Ags = \{Alice, Bob\}$, $S = \{s_0, s_1,..., s_4\}$ with $s_0$ being the initial state, $Act_{Alice} = \{invest, withhold, \bot\}$, $Act_{Bob} = \{share, keep, \bot\}$ and the transition function defined in an obvious way (see Figure~\ref{fig:trustgame}). 
Note that we do not represent payoffs explicitly in our modelling of the trust game, but rather capture them using atomic propositions. For example, $richer_{Alice,Bob}$ is true is state $s_1$, while $richer_{Bob,Alice}$ holds in $s_3$.
Note also that $Alice$ and $Bob$ proceed in turns, which is captured through joint actions where the other agent takes the silent action $\bot$. We represent states as pairs: 
\[
(a_{Alice},a_{Bob}),
\]
where $a_{Alice} \in Act_{Alice}$ is Alice's last action and $a_{Bob} \in Act_{Bob}$ is Bob's last action. For example, $s_0 = (\bot, \bot)$, $s_2 = (invest, \bot)$ and $s_4 = (\bot, share)$.

\setlength{\tabcolsep}{4pt}

\begin{table}
\caption{Strategies for Alice and Bob}
\label{tab:strategyTrustGame1}
\begin{center}
\begin{tabular}{ccccc}
\toprule
Strategy & $withhold$ & $invest$ & $keep$ & $share$ \\
\midrule
$\astrat{passive}$ & 0.7 & 0.3 & & \\
$\astrat{active}$ & 0.1 & 0.9 & & \\
$\astrat{share}$ & & & 0.0 & 1.0 \\
$\astrat{keep}$ & & & 1.0 & 0.0 \\
\bottomrule
\end{tabular}
\end{center}
\end{table}

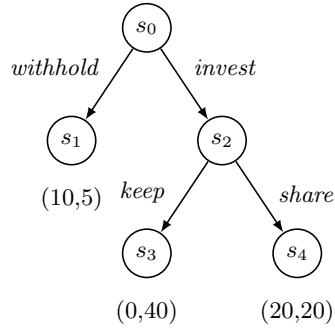
\begin{figure}
\begin{center}
\begin{tikzpicture}[-latex ,auto ,node distance =1.5 cm and 1 cm,on grid,semithick,state/.style ={ circle,draw,minimum width =0.5 cm}]
\node[state] (A) {$s_0$};
\node[state] (B) [below left=of A] {$s_1$};
\node[state] (C) [below right=of A] {$s_2$};
\node[state] (D) [below left=of C] {$s_3$};
\node[state] (E) [below right=of C] {$s_4$};
\path (A) edge [] node[swap] {\emph{withhold}} (B);
\path (A) edge [] node[] {\emph{invest}} (C);
\path (C) edge [] node[swap, pos=0.7] {\emph{keep}} (D);
\path (C) edge [] node[pos=0.7] {\emph{share}} (E);

\node (B1) [below=of B, yshift=2.2em] {(10,5)};
\node (D1) [below=of D, yshift=2.2em] {(0,40)};
\node (E1) [below=of E, yshift=2.2em] {(20,20)};
\end{tikzpicture}
\end{center}
\caption{Simple trust game}
\label{fig:trustgame}
\end{figure}

We now equip agents with strategies. For Alice, we define $\astrat{active}$ and $\astrat{passive}$, where the former corresponds to high likelihood of Alice investing her money and the latter allocates greater probability to withholding it. For Bob, we set two pure strategies $\astrat{keep}$ and $\astrat{share}$, corresponding to him keeping and sharing the money with Alice. The strategies are summarised in Table~\ref{tab:strategyTrustGame1}. \hfill $\Box$
\end{example}

\subsection{Temporal Reasoning about \SMG s}

We now recall the syntax of the Probabilistic Computation Tree Logic \PCTLS~\cite{Aziz1995,baierbook}\footnote{Note that we do not consider here the coalition operator, and therefore the logic rPATL* \cite{CFK+13b} that is commonly defined over stochastic game models.}
for reasoning about temporal properties in systems exhibiting nondeterministic and probabilistic choices. 
\PCTLS\ is based on \CTLS~\cite{clarkebook} for purely nondeterministic systems and retains its expressive power, additionally extending it with a probabilistic operator $\prob{\bowtie q}{\psi}$.

\begin{definition}
The syntax of the logic \PCTLS\ is as follows.
$$
\begin{array}{lcl}
\phi & ::= & p ~|~\neg \phi~|~\phi\lor \phi~|~\forall\psi~|~\prob{\bowtie q}{\psi}~\\
\psi & ::= & \phi~|~\neg \psi~|~\psi\lor \psi~|~\next \psi ~|~ \psi\until \psi
\end{array}
$$
where $p$ is an atomic proposition, $\bowtie\in \{<, \leq, >, \geq\}$, and $q\in [0,1]$.
\end{definition}

In the above, $\phi$ is a \PCTLS\ (state) formula and $\psi$ an LTL (path) formula. The operator $\forall$ is the (universal) path quantifier of CTL$^*$
and $\prob{\bowtie q}{\psi}$ is the probabilistic operator of PCTL~\cite{HJ1994}, which expresses that $\psi$ holds with probability in relation $\bowtie$ with $q$.
The remaining operators, $\next \psi$ (next) and $\psi\until \psi$ (until) 
follow their usual meaning from PCTL and CTL$^*$. The derived operators such as $\phi_1\land \phi_2$, $\eventually\psi$ (eventually), $\always \psi$ (globally), $\psi\release \psi$ (release) and $\exists\phi$ (existential path quantifier) can be obtained in the standard way.

Let $\cM$ be an SMG. Given a path $\rho s$ which has $s$ as its last state, a strategy $\astrat{}\in\Strategy$ of $\cM$, and a formula $\psi$, we write:
$$
\probf{\cM,\astrat{},\rho s}{}(\psi) \defequals  \probs{\astrat{}}{\cM}\{\delta\in \infpath{T}^{\cM}(s)~|~\cM,\rho s,\delta\models \psi\}
$$
for the probability of implementing $\psi$ on a path $\rho s$ when a strategy $\astrat{}$ applies. The relation $\cM,\rho,\delta\models \psi$ is defined below. Based on this, we define:
\begin{align*}
\probf{\cM,\rho}{min}(\psi) & \defequals \textstyle\inf_{\astrat{}\in \Strategy} \probf{\cM,\astrat{},\rho}{}(\psi), \\[0em]
\probf{\cM,\rho}{max}(\psi) & \defequals \textstyle\sup_{\astrat{}\in \Strategy} \probf{\cM,\astrat{},\rho}{}(\psi)
\end{align*}
\commentout{
$$
{
\setlength\arraycolsep{2pt}
\begin{array}{lcl}
\probs{\cM,\rho}{min}(\psi) & \defequals & \inf_{\astrat{}\in \Strategy} \probm{\cM,\astrat{},\rho}(\psi), \\
\probs{\cM,\rho}{max}(\psi) & \defequals & \sup_{\astrat{}\in \Strategy} \probm{\cM,\astrat{},\rho}(\psi) \\
\end{array}
}
$$
}
as the minimum and maximum probabilities of implementing $\psi$ on a path $\rho$ over all strategies in $\Strategy$.

We now give semantics of the logic \PCTLS\ for concurrent stochastic games. 

\begin{definition}\label{def:semantics1}
Let $\cM=(\Ags,\states,\sinit,\{\Act{A}\}_{A\in\Ags},T,L)$ be an \SMG\ and $\rho \in \fpath{T}^\cM$. 
The satisfaction relation $\models$ of \PCTLS\ is defined inductively by:

\begin{itemize}

\item $\cM,\rho\models  p$ if $p\in L(last(\rho))$,

\item $\cM,\rho\models \neg \phi$ if not $\cM,\rho\models  \phi$,

\item $\cM,\rho\models \phi_1\lor \phi_2$ if  $\cM,\rho\models \phi_1$ or $\cM,\rho\models \phi_2$,

\item $\cM,\rho\models \forall\psi$ if $\cM,\rho,\delta\models \psi$ for all $\delta\in \infpath{T}^{\cM}(last(\rho))$,

\commentout{
\item $\cM,\rho\models \prob{\bowtie q}{\psi}$ if
\[
\left \{
\begin{array}{ll}
\probs{\cM,\rho}{min}(\psi) \bowtie q & \text{ when } \bowtie \in \{\geq, >\} \\
\probs{\cM,\rho}{max}(\psi) \bowtie q & \text{ when } \bowtie \in \{\leq, <\} \\
\end{array}
\right.
\]
}

\item $\cM,\rho\models \prob{\bowtie q}{\psi}$ if $\probf{\cM,\rho}{opt(\bowtie)}(\psi) \bowtie q$, where 
\[
opt(\bowtie) = \left \{
\begin{array}{ll}
\text{min} & \text{ when } \bowtie \in \{\geq, >\} \\
\text{max} & \text{ when } \bowtie \in \{\leq, <\} \\
\end{array}
\right.
\]
\end{itemize}
and for any infinite continuation $\delta \in \infpath{T}^\cM$ of $\rho$ (i.e., $\delta(0) = last(\rho)$):
\begin{itemize}

\item $\cM,\rho,\delta\models \phi$ if $\cM,\rho\models \phi$,
\item $\cM,\rho,\delta\models \neg \psi$  if not $\cM,\rho,\delta\models \psi$,
\item $\cM,\rho,\delta\models  \psi_1\lor \psi_2$  if  $\cM,\rho,\delta\models \psi_1$ or $\cM,\rho,\delta\models \psi_2$,
\item $\cM,\rho, \delta\models \next \psi$ if $\cM, \rho\cdot \delta[0..1], \delta[1..\infty]\models \psi$,
\item $\cM,\rho, \delta\models \psi_1\until\psi_2$ if there exists $n\geq 0$ such that $\cM,\rho\cdot\delta[0..n], \delta[n..\infty]\models \psi_2$ and $\cM,\rho\cdot\delta[0..k], \delta[k..\infty] \models \psi_1$ for all $0\leq k < n$.
\end{itemize}
\end{definition}

We note that the semantics of state formulas is defined on finite paths (histories) rather than states, whereas the semantics of path formulas is defined on a finite path together with its infinite continuation (rather than a single infinite path). The reason for defining it in such a way is to be consistent with definitions of trust operators (which we introduce in Section~\ref{sec:pref}), whose semantics is dependent on execution history (understood as a sequence of past states of a system).

Below, we write $\overline{<}$ for $ >$, $\overline{>}$ for $<$, $\overline{\leq}$ for $\geq$ and $\overline{\geq}$ for $\leq$ (inverting the order), and write $\widehat{<}$ for $ \leq $, $\widehat{>}$ for $\geq$, $\widehat{\leq}$ for $<$, and $\widehat{\geq}$ for $>$ (strict/non-strict variants).
\begin{proposition}
Let $\cM$ be an \SMG\ and $\rho \in \fpath{}^{\cM}$. The following equivalences hold for any formula $\psi$:
\begin{enumerate}
\item $\cM,\rho\models \neg \prob{\hspace{1pt}\bowtie q}{\psi} $ ~iff~ $\cM,\rho\models  \prob{\hspace{1pt}\widehat{\bowtie} 1- q}{\neg \psi} $
\item $\cM,\rho\models \prob{\hspace{1pt}\bowtie q}{\psi} $ iff $\cM,\rho\models  \prob{\hspace{1pt}\overline{\bowtie} 1- q}{\neg \psi} $
\end{enumerate}
\end{proposition}

\begin{definition}
For a given \SMG\ $\cM$ and a formula $\phi$ of the language \PCTLS, the \emph{model checking problem}, written as $\cM\models \phi$, is to decide whether $\cM,s\models\phi$ for all initial states $s\in \sinit$.
\end{definition}

\begin{example}\label{example:trustgameformulas}
Here we give a few examples of \PCTLS\ formulas that we may wish to check on the trust game from Example~\ref{example:trustgame}. The formula
\[
\prob{\leq 0.9}{\next (a_{Alice} = invest)}
\]
expresses that the probability of Alice investing in the next step is no greater than 0.9. On the other hand, the formula
\[
\prob{\leq 1}{\eventually (a_{Bob} = keep)}
\]
states (the obvious fact) that the probability of Bob keeping the money in the future is no greater than 1. 
Finally, the formula
\[
\exists\eventually richer_{Alice,Bob},
\]
where $richer_{Alice,Bob}$ is an atomic proposition with obvious meaning, states that eventually a state can be reached
where Alice has more money than Bob. 

All the above formulas are true when evaluated at the state $s_0$, given in Fig~\ref{fig:trustgame}. \hfill $\Box$
\end{example}

\section{Stochastic Multiagent Systems with the Cognitive Dimension 
}\label{sec:smgc}

\newcommand{\be}{{\tt be}}

\newcommand{\reach}{reach}

In this section, we present a framework for reasoning about autonomous agents in multiagent systems. The key novelty of our model is the consideration of agents' mental attitudes to enable autonomous decision making, which we achieve by equipping agents with \emph{goals} and \emph{intentions} (also called \emph{pro-attitudes}). We enhance stochastic multiplayer games with a cognitive mechanism to represent reasoning about goals and intentions. The system then evolves along two interleaving dimensions: \emph{temporal}, representing actions of agents in the physical space, and \emph{cognitive}, which corresponds to mental reasoning processes, i.e., goal and intention changes, that determine the actions that agents take. 

\subsection{Cognitive Reasoning}\label{subsec:cognitiveReasoning}
We now motivate and explain the concepts of the \emph{cognitive state} of an agent and the \emph{cognitive mechanism}, and how they give rise to \emph{partial observability}. 

\paragraph{\bf Cognitive State}

We assume that each agent has a set of \emph{goals} and a set of \emph{intentions}, referred to as \emph{pro-attitudes} and viewed as high-level concepts.
We follow existing literature, 
see e.g., Bratman~\cite{Bratman1987} and Gollwitzer~\cite{Gollwitzer1993}, etc., and identify commitment as the distinguishing factor between goals and intentions, i.e., an intention is a committed plan of achieving some immediate goal. We therefore think of goals as abstract attitudes, for example selflessness or risk-taking, whereas intentions are more concrete and directly influence agents' behaviour. Goals are usually static and independent of external factors, whereas intentions are dynamic, influenced by agent's own goals and by its beliefs about other agents' pro-attitudes. For the purposes of our framework, we assume that agents use action strategies to \emph{implement their intentions} and therefore there exists a one-to-one association between intentions and action strategies. This way, agents' behaviour in the physical space is determined by their mental state.

For an agent $A$, we use $\Goal{A}$ to denote its set of goals and $\Intn{A}$ to denote its set of intentions. At any particular time, an agent may have several goals, but can only have a single intention. Goals are not required to satisfy constraints such as consistency. 

\begin{definition}
Let $A$ be an agent 
and let $\Goal{A}$ and $\Intn{A}$ be its set of goals and intentions, respectively. A \emph{cognitive state} of $A$ consists of a set of goals and an intention, which can be retrieved from global states of the system using the following functions: 
\begin{itemize}
\item $\gs_A:S\rightarrow \powerset{\Goal{A}}$, i.e., $\gs_A(s)$ is a set of agent $A$'s \emph{goals} in state $s$,
\item  $\is_A:S\rightarrow \Intn{A}$, i.e., $\is_A(s)$ is the agent $A$'s \emph{intention} in state $s$.
\end{itemize}
\end{definition}

To illustrate the concepts, we now extend the trust game given in the previous examples to include agents' goals and intentions.

\begin{table}
\caption{Payoff of a simple trust game with trust as a decision factor}
\label{tab:updatedtrust}
\begin{center}
\begin{tabular}{r|cc}
\multicolumn{1}{c}{} & \multicolumn{1}{c}{share}  & \multicolumn{1}{c}{keep} \\\cline{2-3}
invest & (20,20+5) & (0,40-20) \\ 
withhold & (10,5) & (10,5) \\ 
\end{tabular}
\end{center}
\end{table}

\begin{example}\label{example:trustgamecstate}
It is argued in \cite{Kuipers2016} that the single numerical value as the payoff of the trust game introduced in Example~\ref{example:trustgame} is an over-simplification. A more realistic utility should include both the payoff and other hypotheses, including trust. An example payoff table is given in Table~\ref{tab:updatedtrust}, in which Bob's payoff will increase by 5 to denote that he will gain Alice's trust if sharing the investment yield and decrease by 20 to denote that he will lose Alice's trust if keeping the investment yield without sharing. With the updated payoffs, the new Nash equilibrium is for Alice to invest her money and Bob to share the investment yield.

The main point for the new payoffs is for the agents to make decisions not only based on the original payoffs, but also based on the trust that the other agent has. This reflects some actual situations in which one agent may want to improve, or at least maintain, the trust of the other agent.
In our modelling of such a game, we show that this can be captured by adding the cognitive dimension and assuming that Bob makes decisions by considering additionally whether Alice's trust in him reaches a certain level.

For Alice, we let $\Goal{Alice}=\{passive,active\}$ be two goals which represent her attitude towards investment. Intuitively, \emph{passive} represents the goal of keeping the cash and \emph{active} represents the goal of investing. For simplicity, we assume that Alice's intention is determined by her goals and set $\Intn{Alice} = \{passive, active\}$. We also assume that Alice uses strategy $\astrat{passive}$ to implement her \emph{passive} intention, and $\astrat{active}$ to implement her \emph{active} intention, where the strategies are defined in Example~\ref{example:trustgame}.


Bob has a set of goals $\Goal{Bob}=\{investor, opportunist\}$, which represent the goals of being an investor pursuing long-term profits and being an opportunist after short-term profits, respectively. As for Alice, Bob's intentions are associated with action strategies, and we have already defined two such strategies: $\astrat{share}$, in which Bob shares the investment yield with Alice, and $\astrat{keep}$, in which Bob keeps all the money for himself. Hence $\Intn{Bob}=\{share, keep\}$, with the obvious association. We assume that Bob's intention will be \emph{share} when he is an investor and his belief in Alice being active is above a certain threshold, and \emph{keep} otherwise. Intuitively, when he is an \emph{investor}, Bob intends to build a good relationship with Alice (in other words, gain Alice's \emph{trust}), hoping that it will pay off in his future interactions with her.

We extend the trust game $\trustgame$ defined in Example~\ref{example:trustgame} by expanding the states to additionally include cognitive states. In particular, each state can now be represented as a tuple:
$$
(a_{Alice},a_{Bob},\gs_{Alice},\gs_{Bob},\is_{Alice},\is_{Bob}),
$$
such that 
$a_{Alice}$ and $a_{Bob}$ are as before and $\gs_{Alice}\subseteq \Goal{Alice}\cup \{\bot\}$, $\gs_{Bob}\subseteq \Goal{Bob}\cup\{\bot\}$, $\is_{Alice} \in \Intn{Alice} \cup \{\bot\}$, and $\is_{Bob}\in \Intn{Bob}\cup \{\bot\}$ denote the cognitive part of the state, namely Alice's and Bob's goals and intentions. \hfill $\Box$
\end{example}

\paragraph{\bf Partial Observation}

It is common that, in real-world systems, agents are {\it not} able to fully observe the system state at any given time. In a typical scenario, every agent runs a local protocol, maintains a local state, observes part of the system state by, e.g., sensing devices, and communicates with other agents. It is impractical, and in fact undesirable, from the system designer's point of view, to assume that agents can learn the local states of other agents or learn what the other agents observe. 
In the context of our work partial observability arises naturally through the cognitive state, which represents an \textit{internal} state of every agent that is, in general, not observable by other agents. We formalise this notion with the following definition.

\begin{definition}
A \emph{partially observable} stochastic multiplayer game (\POSMG) is a tuple $\cM = (G, \sep\{\Obs{A}\}_{A\in\Ags}, \sep\{\obs{A}\}_{A\in\Ags})$, where 
\begin{itemize}
    \item $G = (\Ags, \sep\states, \sep\sinit, \sep\{\Act{A}\}_{A\in\Ags}, \sep T, \sep L)$ is an \SMG, 
    \item $\Obs{A}$ is a finite set of observations for agent $A$, and 
    \item $\obs{A}: \states \longrightarrow \Obs{A}$ is a labelling of states with observations for agent $A$.
\end{itemize}
\end{definition}

\begin{remark}
We note that, unlike partially observable Markov processes (POMDPs), in  which observations are probability distributions, we follow the setting in~\cite{FH1997} and work with deterministic observations. In \cite{GM1998}, it is stated without proof that probabilistic observations do not increase the complexity of the problem. 
\end{remark}


\commentout{
\begin{definition}
Let $\cM$ be a multi-agent system with set of global states $S$ and set of agents $\Ags$. For agent $A \in \Ags$ we let $\Obs{A}$ be a set of their possible observations and $\obs{A}: S\rightarrow \Obs{A}$ be an observation function providing agent $A$ with one observation in every state\footnote{Unlike POMDPs, in  which observations are probabilistic distributions, we follow the setting in~\cite{FH1997} and define a deterministic observation. With a footnote, \cite{GM1998} states without proofs that probabilistic observation does not increase the complexity of the problem. In Section~\ref{sec:probobs}, we formally show that this is the case in our framework. }.
\end{definition}

We remark that the above definition applies not only to \SMG s, but to any multi-agent system we define in due course.
}

We lift the observations from states to paths in the obvious way. Formally, for a finite path $\rho=s_0...s_n$, we define $\obs{A}(\rho)=\obs{A}(s_0)...\obs{A}(s_n)$. 

\begin{remark}\label{rem:syncPerfRecall}
In our model, an agent remembers both its past observations and the number of states, known as synchronous perfect recall~\cite{FHMVbook}. This assumption is necessary for the definition of belief in Section~\ref{sec:beliefsupdate}. 
\end{remark} 
%




\begin{example}\label{example:trustgameobs}
For the trust game from Example~\ref{example:trustgamecstate}, the agents' observation functions for $A\in\{Alice, Bob\}$ are as follows:
$$\obs{A}((a_{Alice},a_{Bob},\gs_{Alice},\gs_{Bob},\is_{Alice},\is_{Bob}))=(a_{Alice},a_{Bob},\gs_{A},\is_{A}),$$ denoting that they cannot observe the opponent's cognitive state but can observe their last action. The set of observations $\Obs{A}$ can be easily inferred from this definition. \hfill $\Box$
\end{example}

We say $\cM$ is \emph{fully observable} if $s = s'$ iff $\obs{A}(s) = \obs{A}(s')$ for all $A \in \Ags$. In fully observable systems, agents make decisions based on the current state, whereas in partially observable systems decisions must be based on all past observations.  

\paragraph{\bf Cognitive Mechanism}
While agents interact with each other and with the environment by taking actions in the physical space, 
they make decisions through cognitive processes that determine their behaviour. Thus, in addition to the temporal dimension of transitions $s \trans{a} s'$, we also distinguish a \textit{cognitive} dimension of transitions $s \ctrans{} s'$, which corresponds to mental reasoning processes. The idea is that each temporal transition is preceded by a cognitive transition, which represents an agent's reasoning that led to its decision about which action to take. While transitions in the temporal dimension conform to the transition function $T$, cognitive changes adhere to the \emph{cognitive mechanism},
which determines for an agent its legal goals and intentions. Formally, we have the following definition of a stochastic game extended with a cognitive mechanism.

\newcommand{\smgcTuple}{(\Ags, \sep\states, \sep\sinit, \sep\{\Act{A}\}_{A\in\Ags}, \sep T, \sep L, \allowbreak \sep\{\Obs{A}\}_{A\in\Ags}, \sep\{\obs{A}\}_{A\in\Ags}, \sep\{\Cognition_A\}_{A \in \Ags}, \sep\{\cstrat{A}\}_{A \in \Ags})}

\begin{definition}
A \emph{stochastic multiplayer game with the cognitive dimension} (\SMGC) is
a tuple $\cM = (G, \sep\{\Cognition_A\}_{A \in \Ags}, \sep\{\cstrat{A}\}_{A \in \Ags} )$, where 
\begin{itemize}
\item $G = (\Ags, \sep\states, \sep\sinit, \sep\{\Act{A}\}_{A\in\Ags}, \sep T, \sep L, \sep\{\Obs{A}\}_{A\in\Ags}, \sep\{\obs{A}\}_{A\in\Ags})$ is a \POSMG,
\item $\Cognition_A = \langle \lgoal{A},\sep\lintn{A} \rangle$ is the \emph{cognitive mechanism} of agent $A$, consisting of a \legal\ goal function $\lgoal{A}:S\rightarrow \powerset{\powerset{\Goal{A}}}$ and a \legal\ intention function $\lintn{A}:S\rightarrow \powerset{\Intn{A}}$, and 
\item $\cstrat{A} = \langle \goal{A}, \sep\intn{A} \rangle$ is the \emph{cognitive strategy} of agent $A$, consisting of a goal strategy $\goal{A}: \fpath{}^\cM\rightarrow \dist{\powerset{\Goal{A}}}$ and an intention strategy $\intn{A}: \fpath{}^\cM\rightarrow \dist{\Intn{A}}$.
\end{itemize}
We refer to the \SMG\ $G$ from the above definition as the \emph{induced \SMG} from $\cM$.
\end{definition}

Thus, the \SMGC\ model generalises the usual notion of multiplayer games by extending states $S$ with agents' cognitive states, 
and adding for each agent $A$ the cognitive mechanism $\Cognition_A$ and $A$'s cognitive strategies to enable autonomous decision making.
We sometimes refer to the set $\Cognition = \{ \Cognition_A \}_{A \in \Ags}$ of cognitive mechanisms of all agents as the cognitive mechanism of the system $\cM$.

The \emph{\legal\ goal} (resp. \emph{intention}) function $\lgoal{A}$ (resp. $\lintn{A}$) specifies \legal\ goal (resp. intention) changes in a given state. Intuitively, those are goal (resp. intention) changes that an agent is allowed (but might not be willing) to make. One possible use of those functions is to enforce application-specific constraints that goals or intentions must satisfy (see Example~\ref{example:trustgamectrans}). 

The \emph{cognitive strategy} $\cstrat{A}$ determines how an agent's cognitive state evolves over time. Specifically, the goal (resp. intention) strategy $\goal{A}$ (resp. $\intn{A}$) specifies the incurred goal (resp. intention) changes (along with their probabilities), which are under agent $A$'s consideration for a given execution history. We sometimes call $\goal{A}$ (resp. $\intn{A}$) a \possible\ goal (resp. intention) function and require that \possible\ goal (resp. intention) changes are also \legal. Formally, for all $\rho s\in  \fpath{}^\cM$, we must have $supp(\goal{A}(\rho s))\subseteq \Goal{A}(s)$  and $supp(\intn{A}(\rho s))\subseteq \Intn{A}(s)$. 

\begin{remark}
We remark that cognitive strategies $\goal{A}$ and $\intn{A}$ are, in general, not computable. In Section~\ref{sec:proupdate}
we propose how to realise cognitive strategies so that they can be effectively computed.
\end{remark}

\commentout{
\begin{definition}
With a slight abuse of notation, we define the cognitive mechanism $\Cognition_A$ of an agent $A$ in a multi-agent system as a tuple $\langle \Goal{A}, \Intn{A} \rangle$ of functions which specify \legal\ cognitive transitions in a given global state:
\begin{itemize}
\item $\Goal{A}:S\rightarrow \powerset{\powerset{\Goal{A}}}$ gives agent $A$'s set of \legal\ goal changes in state $s$ and
\item $\Intn{A}:S\rightarrow \powerset{\Intn{A}}$ gives agent $A$'s set of \legal\ intention changes in state $s$.
\end{itemize}
\end{definition}
}

We note the following correspondence between the cognitive and temporal dimension: $\Goal{A}$ and $\Intn{A}$ (regarded as sets) specify all the goals and intentions of a given agent $A$ similarly to the way that $\Act{A}$ specifies all the actions of $A$, while the cognitive mechanism $\Cognition$ identifies possible cognitive transitions in a similar fashion to the transition function $T$ playing the same role for the temporal dimension. Finally, a cognitive strategy gives the probability for agents' goal and intention changes, analogously to an action strategy quantifying the likelihood of actions taken by agents in the temporal dimension. 

The following standard assumption ensures that agents' temporal and cognitive transitions, as well as their cognitive strategies, are consistent with their partial observability. 
\begin{assumption}\label{assump:uni} (Uniformity Assumption)
An \SMGC\ $\cM = \smgcTuple$ satisfies the \emph{Uniformity Assumption} if the following conditions hold.
\begin{itemize}
\item Agents can distinguish states with different sets of joint actions or \legal\ cognitive changes: for any two states $s_1$ and $s_2$ and an agent $A\in\Ags$, $\obs{A}(s_1)=\obs{A}(s_2)$ implies $\Act{B}(s_1)=\Act{B}(s_2)$, $\lgoal{B}(s_1)=\lgoal{B}(s_2)$, and $ \lintn{B}(s_1)= \lintn{B}(s_2)$, for all $B\in \Ags$.
\item Agents can distinguish execution histories which give rise to different cognitive strategies: for any two finite paths $\rho_1$, $\rho_2$, $\obs{A}(\rho_1)=\obs{A}(\rho_2)$ implies $\goal{A}(\rho_1)(x) = \goal{A}(\rho_2)(x)$ and  $\intn{A}(\rho_1)(x) = \intn{A}(\rho_2)(x)$ for all $x \subseteq \Goal{A}$.
\end{itemize}
\end{assumption}
Given a state $s$ and a set of agent A's goals $x\subseteq \Goal{A}$, we write $A.g(s,x)$ for the state obtained from $s$ by substituting agent $A$'s goals with $x$.
Similar notation $A.i(s,x)$ is used for the intention change when $x\in \Intn{A}$.  Alternatively, we may write $s \ctrans{A.g.x} s'$ if $s'=A.g(s,x)$ contains the goal set $x$ for $A$ and $s \ctrans{A.i.x} s'$ if $s'=A.i(s,x)$ contains the intention $x$ for $A$.

\begin{example}\label{example:trustgamectrans}
Having extended states of the trust game $\trustgame$ with goals and intentions in Example~\ref{example:trustgamecstate}, we now make $\trustgame$ into an \SMGC\ by introducing cognitive transitions, defining the cognitive mechanism and specifying cognitive strategies of agents.

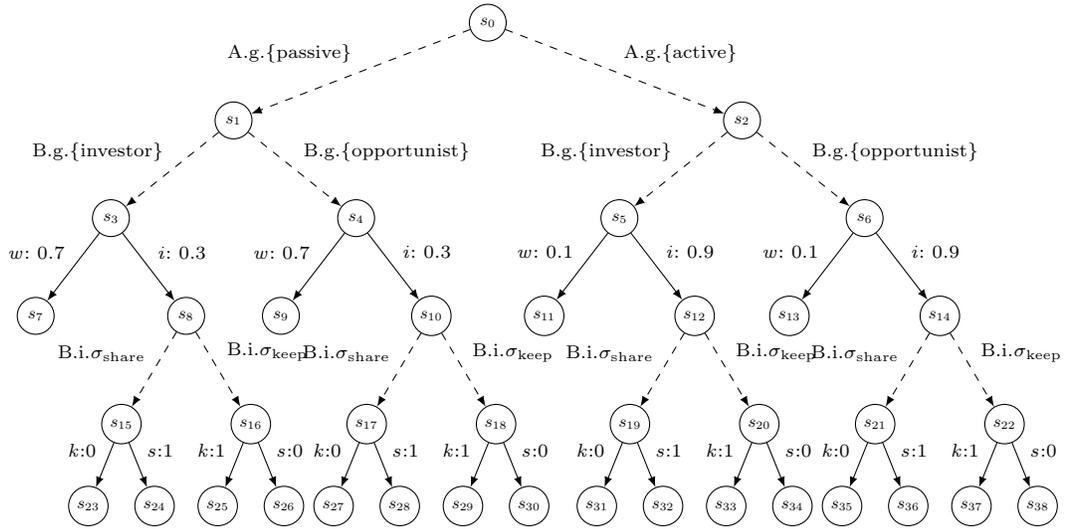
\begin{figure}
\begin{center}
\begin{tikzpicture}[-latex ,auto ,node distance =1.3 cm and 1 cm,on grid,state/.style ={ circle,draw,minimum size =0.7 cm,scale=0.7}]
\node[state] (0A) {$s_0$};
\newcommand{\oneShiftX}{3.4cm}
\node[state] (1A) [below left=of 0A, xshift=-\oneShiftX] {$s_1$};
\node[state] (1B) [below right=of 0A, xshift=\oneShiftX] {$s_2$};

\newcommand{\twoShiftX}{0.9cm}
\node[state] (2A) [below left=of 1A, xshift=-\twoShiftX] {$s_3$};
\node[state] (2B) [below right=of 1A, xshift=\twoShiftX] {$s_4$};
\node[state] (2C) [below left=of 1B, xshift=-\twoShiftX] {$s_5$};
\node[state] (2D) [below right=of 1B, xshift=\twoShiftX] {$s_6$};

\newcommand{\threeShiftX}{0cm}
\node[state] (3A) [below left=of 2A, xshift=-\threeShiftX, yshift=-0cm] {$s_7$};
\node[state] (3B) [below right=of 2A, xshift=\threeShiftX, yshift=-0cm] {$s_8$};
\node[state] (3C) [below left=of 2B, xshift=-\threeShiftX, yshift=-0cm] {$s_9$};
\node[state] (3D) [below right=of 2B, xshift=\threeShiftX, yshift=-0cm] {$s_{10}$};
\node[state] (3E) [below left=of 2C, xshift=-\threeShiftX, yshift=-0cm] {$s_{11}$};
\node[state] (3F) [below right=of 2C, xshift=\threeShiftX, yshift=-0cm] {$s_{12}$};
\node[state] (3G) [below left=of 2D, xshift=-\threeShiftX, yshift=-0cm] {$s_{13}$};
\node[state] (3H) [below right=of 2D, xshift=\threeShiftX, yshift=-0cm] {$s_{14}$};

\newcommand{\fourShiftX}{0.2cm}
\newcommand{\fourShiftY}{-0.2cm}
\node[state] (4A) [below left=of 3B, xshift=\fourShiftX, yshift=\fourShiftY] {$s_{15}$};
\node[state] (4B) [below right=of 3B, xshift=-\fourShiftX, yshift=\fourShiftY] {$s_{16}$};
\node[state] (4C) [below left=of 3D, xshift=\fourShiftX, yshift=\fourShiftY] {$s_{17}$};
\node[state] (4D) [below right=of 3D, xshift=-\fourShiftX, yshift=\fourShiftY] {$s_{18}$};
\node[state] (4E) [below left=of 3F, xshift=\fourShiftX, yshift=\fourShiftY] {$s_{19}$};
\node[state] (4F) [below right=of 3F, xshift=-\fourShiftX, yshift=\fourShiftY] {$s_{20}$};
\node[state] (4G) [below left=of 3H, xshift=\fourShiftX, yshift=\fourShiftY] {$s_{21}$};
\node[state] (4H) [below right=of 3H, xshift=-\fourShiftX, yshift=\fourShiftY] {$s_{22}$};

\newcommand{\fiveShiftX}{0.8cm}
\newcommand{\fiveShiftY}{0.3cm}
\node[state] (5A) [below left=of 4A, xshift=\fiveShiftX, yshift=\fiveShiftY] {$s_{23}$};
\node[state] (5B) [below right=of 4A, xshift=-\fiveShiftX, yshift=\fiveShiftY] {$s_{24}$};
\node[state] (5C) [below left=of 4B, xshift=\fiveShiftX, yshift=\fiveShiftY] {$s_{25}$};
\node[state] (5D) [below right=of 4B, xshift=-\fiveShiftX, yshift=\fiveShiftY] {$s_{26}$};
\node[state] (5E) [below left=of 4C, xshift=\fiveShiftX, yshift=\fiveShiftY] {$s_{27}$};
\node[state] (5F) [below right=of 4C, xshift=-\fiveShiftX, yshift=\fiveShiftY] {$s_{28}$};
\node[state] (5G) [below left=of 4D, xshift=\fiveShiftX, yshift=\fiveShiftY] {$s_{29}$};
\node[state] (5H) [below right=of 4D, xshift=-\fiveShiftX, yshift=\fiveShiftY] {$s_{30}$};
\node[state] (5I) [below left=of 4E, xshift=\fiveShiftX, yshift=\fiveShiftY] {$s_{31}$};
\node[state] (5J) [below right=of 4E, xshift=-\fiveShiftX, yshift=\fiveShiftY] {$s_{32}$};
\node[state] (5K) [below left=of 4F, xshift=\fiveShiftX, yshift=\fiveShiftY] {$s_{33}$};
\node[state] (5L) [below right=of 4F, xshift=-\fiveShiftX, yshift=\fiveShiftY] {$s_{34}$};
\node[state] (5M) [below left=of 4G, xshift=\fiveShiftX, yshift=\fiveShiftY] {$s_{35}$};
\node[state] (5N) [below right=of 4G, xshift=-\fiveShiftX, yshift=\fiveShiftY] {$s_{36}$};
\node[state] (5O) [below left=of 4H, xshift=\fiveShiftX, yshift=\fiveShiftY] {$s_{37}$};
\node[state] (5P) [below right=of 4H, xshift=-\fiveShiftX, yshift=\fiveShiftY] {$s_{38}$};

\tikzset{every edge/.append style={font=\scriptsize}}
\path (0A) edge [dashed] node[swap] {A.g.\{passive\}} (1A);
\path (0A) edge [dashed] node[] {A.g.\{active\}} (1B);

\path (1A) edge [dashed] node[swap] {B.g.\{investor\}} (2A);
\path (1A) edge [dashed] node[] {B.g.\{opportunist\}} (2B);
\path (1B) edge [dashed] node[swap] {B.g.\{investor\}} (2C);
\path (1B) edge [dashed] node[] {B.g.\{opportunist\}} (2D);

\path (2A) edge [] node[swap] {\emph{w}: $0.7$} (3A);
\path (2A) edge [] node[] {\emph{i}: $0.3$} (3B);
\path (2B) edge [] node[swap] {\emph{w}: $0.7$} (3C);
\path (2B) edge [] node[] {\emph{i}: $0.3$} (3D);

\path (2C) edge [] node[swap] {\emph{w}: $0.1$} (3E);
\path (2C) edge [] node[xshift=-0cm, yshift=-0cm] {\emph{i}: $0.9$} (3F);
\path (2D) edge [] node[swap] {\emph{w}: $0.1$} (3G);
\path (2D) edge [] node[xshift=-0cm, yshift=-0cm] {\emph{i}: $0.9$} (3H);

\path (3B) edge [dashed] node[swap] {B.i.$\mathrm{\astrat{share}}$} (4A);
\path (3B) edge [dashed] node[] {B.i.$\mathrm{\astrat{keep}}$} (4B);
\path (3D) edge [dashed] node[swap] {B.i.$\mathrm{\astrat{share}}$} (4C);
\path (3D) edge [dashed] node[] {B.i.$\mathrm{\astrat{keep}}$} (4D);
\path (3F) edge [dashed] node[swap] {B.i.$\mathrm{\astrat{share}}$} (4E);
\path (3F) edge [dashed] node[] {B.i.$\mathrm{\astrat{keep}}$} (4F);
\path (3H) edge [dashed] node[swap] {B.i.$\mathrm{\astrat{share}}$} (4G);
\path (3H) edge [dashed] node[] {B.i.$\mathrm{\astrat{keep}}$} (4H);

\path (4A) edge [] node[swap] {\emph{k}:$0$} (5A);
\path (4A) edge [] node[] {\emph{s}:$1$} (5B);
\path (4B) edge [] node[swap] {\emph{k}:$1$} (5C);
\path (4B) edge [] node[] {\emph{s}:$0$} (5D);
\path (4C) edge [] node[swap] {\emph{k}:$0$} (5E);
\path (4C) edge [] node[] {\emph{s}:$1$} (5F);
\path (4D) edge [] node[swap] {\emph{k}:$1$} (5G);
\path (4D) edge [] node[] {\emph{s}:$0$} (5H);
\path (4E) edge [] node[swap] {\emph{k}:$0$} (5I);
\path (4E) edge [] node[] {\emph{s}:$1$} (5J);
\path (4F) edge [] node[swap] {\emph{k}:$1$} (5K);
\path (4F) edge [] node[] {\emph{s}:$0$} (5L);
\path (4G) edge [] node[swap] {\emph{k}:$0$} (5M);
\path (4G) edge [] node[] {\emph{s}:$1$} (5N);
\path (4H) edge [] node[swap] {\emph{k}:$1$} (5O);
\path (4H) edge [] node[] {\emph{s}:$0$} (5P);

\end{tikzpicture}
\end{center}
\caption{Trust game with cognitive dimension}
\label{fig:trustgamecogn}
\end{figure}

We give a graphical representation of $\trustgame$ in Figure~\ref{fig:trustgamecogn}. Note that \emph{w} and \emph{i} stand for Alice's actions \emph{withhold} and \emph{invest} respectively, whereas \emph{s} and \emph{k} denote Bob's actions \emph{share} and \emph{keep}. 
Cognitive transitions are represented with dashed lines. Temporal actions are annotated with probabilities which reflect the intention (i.e., an action strategy) that an agent has in a given state. 

Below we explain how we arrived at such a system. The execution of the game starts by agents choosing their goals. While it may seem unnatural (a more realistic approach would probably involve multiple initial states corresponding to agents having different goals), such a way of modelling plays well with our formalism and does not restrict the generality of our approach. 
Formally, we specify those cognitive transitions using \legal\ goal function for Alice and Bob as follows:
\begin{alignat*}{2}
\lgoal{Alice}(s_0) &= \{\{active\}, \{passive\}\}, \\
\lgoal{Bob}(s_k) &= \{\{investor\}, \{opportunist\}\},
\end{alignat*}
where $k \in \{1,2\}$.

Once their goals are set, agents begin interacting with one another in the physical space by taking actions. However, each action is preceded by an agent determining its intention, which represents the mental reasoning that results in action selection. 
Note that we do not depict Alice's intention change in Figure~\ref{fig:trustgamecogn}, the reason being our assumption that Alice's intention is fully determined by her goals. 

As mentioned above, Alice's actions are annotated with probabilities, according to our assumption that cognitive state determines agents' behaviour. For example, in states $s_3$ and $s_4$, Alice's intention is \emph{passive}, and so the probabilities of withholding and investing the money are given by Table~\ref{tab:strategyTrustGame1}.
If Alice withholds her money, the game ends. Otherwise, Bob determines his intention (his choice depends on his goals and his belief about Alice's goals) and performs his action of keeping or sharing his profit. Bob's \legal\ intentions are given by his \legal\ intention function, defined as:
\[
\lintn{Bob}(s_k) = \{ share, keep \},
\]
where $k \in \{8,10,12,14\}$.

Finally, we define Bob's intention strategy. As mentioned above, Bob takes on intention \emph{share} only if he is an \emph{investor} and he believes that Alice is \emph{active}. The latter depends on Alice's actions in the physical space (since Bob can't observe Alice's cognitive state). In this case, Alice investing her money with Bob increases his belief that she is \emph{active}. We therefore set:
\begin{alignat*}{2}
\intn{Bob}(s_0s_1s_3s_8) &= \langle share \mapsto 1, keep \mapsto 0 \rangle \\
\intn{Bob}(s_0s_1s_4s_{10}) &= \langle share \mapsto 0, keep \mapsto 1 \rangle \\
\intn{Bob}(s_0s_2s_5s_{12}) &= \langle share \mapsto 1, keep \mapsto 0 \rangle \\
\intn{Bob}(s_0s_2s_6s_{14}) &= \langle share \mapsto 0, keep \mapsto 1 \rangle
\end{alignat*}

Note that the above strategy is pure. While the framework does not enforce it, we believe defining it in such a way more accurately resembles human cognitive processes. Also, as we will see later, pure intention strategies are more compatible with our trust formulations. 

Finally, note that we don't consider goal strategies here. The reason for that is our representation of goals in this example as static mental attitudes which agents possess or not, rather than choose dynamically. We therefore treat the first two cognitive transitions as nondeterministic for now. In Section~\ref{sec:preferencefuncs}, we introduce a mechanism which motivates choosing such representation. 
\hfill $\Box$
\end{example}

\begin{remark}
We remark that, while defining the intention strategy as in the above example is easy for simple systems, 
for more complex games this approach does not scale. In particular, we could consider repeated version of the trust game, for which constructing Bob's intention strategy manually is impractical. However, in Section~\ref{sec:beliefsupdate}, we formalise the notion of agent's belief and, in Section~\ref{sec:proupdate}, we propose an intuitive method of constructing the intention strategy efficiently.
\end{remark}

\paragraph{\bf Induced \SMG}
For a concrete system, it is conceptually simpler to assume a certain, predefined interleaving of cognitive and temporal transitions, as we did in Example~\ref{example:trustgamectrans}. However, in general, such interleaving might be arbitrary since agents may change their mental state at any time. It is therefore often useful to think of a \SMGC\ as a collection of induced \SMG s, each corresponding to one configuration of mental states of agents. Those induced \SMG s do not differ as far as their components are concerned (i.e., states, actions and transitions are the same), but different mental states give rise to different behaviour. Execution of such an \SMGC\ can then be viewed as starting in one of the induced \SMG s, remaining there as long as agents perform temporal actions and moving to a different induced \SMG\ as soon as one agent changes its mental state. In the long run, execution alternates between different standard \SMG s, where, at any point, current \SMG\ reflects current mental states of agents, and each temporal transition preserves the current \SMG, while each cognitive transition switches to a different \SMG. The benefit of such an approach is that each induced \SMG\ can then be reasoned about using standard techniques. We note also that, for the purposes of the logic operators introduced later in the paper, we assume that both temporal and cognitive transitions are available to agents in any state.

\paragraph{\bf Paths}
In contrast to the conventional multiplayer games, where each path includes a sequence of temporal transitions, which are consistent with the transition function $T$, in an \SMGC, in view of the active participation of agents' pro-attitudes in determining their behaviour, a path can be constructed by interleaving of temporal \emph{and} cognitive transitions. Each cognitive transition represents a change of an agent's goals or intention. We now extend the definition of a path to allow cognitive transitions. 

\begin{definition}\label{def:path}
Given a stochastic multiplayer game with the cognitive dimension (\SMGC) $\cM$, we define a finite, respectively infinite, \emph{path} $\rho$ as a sequence of states $s_0s_1s_2...$ such that, for all $k \geq 0$, one of the following conditions is satisfied:
\begin{enumerate}
\item $s_k \trans{a} s_{k+1}$ for some joint action $a$ such that $a\in \Act{}(s_k)$,
\item $s_k \ctrans{A.g.x} s_{k+1}$ for some $A\in\Ags$ and $x\subseteq \Goal{A}(s_k)$,
\item $s_k\ctrans{A.i.x} s_{k+1}$ for some $A\in\Ags$ and $x\in \Intn{A}(s_k)$.
\end{enumerate}
We reuse the notation introduced in Section~\ref{subsec:smg} for regular paths and denote the set of finite (resp. infinite) paths by $\fpath{}^\cM$ (resp. $\infpath{}^\cM$), and the set of finite (resp. infinite) paths starting in state $s$ by $\fpath{}^\cM(s)$  (resp. $\infpath{}^\cM(s)$).
\end{definition}

The first condition represents the standard temporal transition, while the other two stand for cognitive transitions -- the former is a goal change, whereas the latter is an intention change.

We often require that both the \legal\ goal function and the \legal\ intention function are serial, i.e., for any state $s$ and any subset $x$ of $\Goal{A}$, there exists a state $s'\in\states$ such that $s\ctrans{A.g.x}s'$, and for any state $s$ and any intention $x\in \Intn{A}$, there exists a state $s'\in\states$ such that $s\ctrans{A.i.x} s'$. This is a non-trivial requirement, since such states $s'$ can be unreachable via temporal transitions. 

\begin{remark}
We remark that the seriality requirement for the probabilistic transition function is usually imposed for model checking, and by no means reduces the generality of the problem, as we can introduce absorbing states to accommodate undefined temporal or cognitive transitions in an obvious way.
\end{remark}

\paragraph{\bf Deterministic Behaviour Assumption}
%


Recall that intentions of agents are associated with action strategies, thereby determining agents' behaviour in a physical space. Below, we formalise that idea, in addition requiring that the associated strategies are pure, which simplifies results.

\begin{assumption} \label{assump:detbeh}
(Deterministic Behaviour Assumption)
An \SMGC\ $\cM$ satisfies the \emph{Deterministic Behaviour Assumption} if each agent's cognitive state deterministically decides its behaviour in terms of selecting its next local action. 
In other words, agent's cognitive state induces a pure action strategy that it follows. 
\end{assumption}
Hence, since global states encode agents' cognitive states, under Assumption~\ref{assump:detbeh} the transition function $T$ becomes deterministic, i.e., for every state $s$, there exists a unique joint action $a$ such that next state $s'$ is chosen with probability $T(s,a)(s')$. 
Therefore, each \SMG\ induced from $\cM$, with fixed pro-attitudes, can be regarded as a Markov chain.

\begin{remark}
We note that a more general version of Assumption~\ref{assump:detbeh} is possible, where the action strategy is not assumed to be pure, and our results can be easily adapted to that variant. 
In fact, action strategies introduced in Example~\ref{example:trustgame} are not pure and so the trust game does not satisfy the strict version of the Deterministic Behaviour Assumption as stated above. Example~\ref{example:trustgamepref} illustrates how the calculations can be adapted to handle that.
\end{remark}

\subsection{Cognitive Reasoning}\label{sec:cognReas}

We can now extend the logic \PCTLS\ with operators for reasoning about agent's cognitive states, resulting in the logic called Probabilistic Computation Tree Logic with Cognitive Operators, \PCTLSM.
\begin{definition}
The syntax of the logic \PCTLSM\ is:
$$
\begin{array}{l}
\phi  ::=  p ~|~\neg \phi~|~\phi\lor \phi~|~\forall\psi~|~\prob{\bowtie q}{\psi}~|~\G_A\phi~|~\I_A\phi~|~\C_A\phi \\
\psi  ::=  \phi~|~\neg\psi~|~\psi \lor \psi~|~\next \psi ~|~ \psi\until \psi
\end{array}
$$
where $p\in AP$, $A\in\Ags$, $\bowtie\in \{<,\leq,>,\geq\}$, and $q\in [0,1]$.
\end{definition}

Intuitively, the newly introduced \textit{cognitive} operators $\G_A\phi$ (goal), $\I_A\phi$ (intention) and $\C_A\phi$ (capability) consider the task expressed as $\phi$ and respectively quantify, in the cognitive dimension, over \emph{\possible\ changes of goals}, \emph{\possible\ intentions} and \emph{\legal\ intentions}.

The semantics for \PCTLSM\ is as follows.

\begin{definition}\label{def:semantics3}Let $\cM = \smgcTuple$ be a \SMGC, $\rho$ a finite path in $\cM$ and $s \in S$ such that $\rho s \in \fpath{}^\cM$. The semantics of previously introduced operators of \PCTLS\ remains the same in \PCTLSM. For the newly introduced cognitive operators, the satisfaction relation $\models$ is defined as follows:
\begin{itemize}

\item $\cM,\rho s \models\G_A\phi$ if 
$\forall x\in \supp(\goal{A}(\rho s)) \, \exists s' \in \states$ : $s \ctrans{A.g.x} s'$ and $\cM,\rho s s' \models \phi$,

\item $\cM,\rho s\models \I_A\phi$ if 
$\forall x\in \supp(\intn{A}(\rho s)) \, \exists s' \in \states$ : $s \ctrans{A.i.x} s'$ and $\cM,\rho s s' \models \phi$,

\item $\cM,\rho s\models \C_A\phi$ if 
$\exists x\in \lintn{A}(s) \, \exists s' \in \states$ : $s \ctrans{A.i.x} s'$ and $\cM,\rho s s' \models \phi$.

\end{itemize}
\end{definition}

Thus, $\G_A\phi$ expresses that $\phi$ holds in future regardless of agent $A$ changing its goals.
Similarly, $\I_A\phi$ states that $\phi$ holds regardless of $A$ changing its intention, whereas $\C_A\phi$ quantifies over the \legal\ intentions, and thus expresses that agent $A$ could change its intention to achieve $\phi$ (however, such a change might not be among agent's \possible\ intention changes).

\begin{remark}\label{rem:pctlsEval}
We note that, when evaluating \PCTLS\ operators, we assume that agents keep their current mental attitudes, i.e., that the future path is purely temporal. Formally, for a \SMGC\ $\cM$, \PCTLS\ operators should be interpreted over the \SMG\ induced from $\cM$. Furthermore, when evaluating \PCTLSM\ formulas, we assume agents can change their goals and intentions at any time, in line with the `induced \SMG s' interpretation presented in Section~\ref{subsec:cognitiveReasoning}. That ensures that the cognitive operators can be applied at any point of execution, as well as meaningfully chained, nested or manipulated in any other way.
\end{remark}

\begin{remark}
We comment here about our definition of the semantics of $\G_A\phi$, $\I_A\phi$ and $\C_A\phi$, whereby the changes of goals and intentions do not incur the changes of other components of the state. This can be counter-intuitive for some cases, e.g., it is reasonable to expect that the intention of the agent may change when its goals are changed. 
We believe that such dependencies are better handled at the modelling stage. For example, a simultaneous change of goals and intention can be modelled as two consecutive cognitive transitions -- a goal change followed by an intention change.
\end{remark}

Below, we write $\overline{X}\phi$ with $X\in \{\G_A,\I_A,\C_A\}$ for $\neg X \neg \phi$. For instance, formula $\overline{\I_A}\phi$ expresses that it is possible to achieve $\phi$ by changing agent $A$'s intention. Note that it is not equivalent to $\C_A \phi$, which quantifies over \legal, rather than \possible, intentions.

\begin{example}\label{example:trustgamecognformulas}
Here we give examples of formulas that we may wish to check on the trust game from Example~\ref{example:trustgamectrans}. The formula:
\[
\G_{Alice} \prob{\leq 0.9}{\eventually (a_{Alice} = invest)}
\]
expresses that, regardless of Alice changing her goals, the probability of her investing in the future is no greater than 90\%. On the other hand, the formula: 
\[
\C_{Bob} \prob{\leq 0}{\next (a_{Bob} = keep)}
\]
states that Bob has a \legal\ intention which ensures that he will not keep the money as his next action. 
Also, the formula:
\[
\overline{\I_{Alice}} \exists\eventually richer_{Alice,Bob},
\]
where $richer_{Alice,Bob}$ is an atomic proposition with obvious meaning, states that Alice can find an intention such that eventually a state can be reached 
where Alice has more money than Bob. Finally, the formula:
\[
\overline{\I_{Alice}} \exists\eventually \G_{Bob} \forall\eventually \neg richer_{Alice,Bob}
\]
expresses that Alice can find an intention such that it is possible to reach a state such that, for all possible Bob's goals, the game will always reach a state in which Bob is no poorer than Alice. \hfill $\Box$
\end{example}

In this paper we study the model checking problem defined as follows.

\begin{definition}
For a given \SMGC\ $\cM$ and a formula $\phi$ of the language \PCTLSM\, the \emph{model checking problem}, written as $\cM\models \phi$, is to decide whether $\cM,s\models\phi$ for all initial states $s\in \sinit$. 
\end{definition}

\begin{table}
\caption{Complexity of Markov chain model checking}
\label{tab:mcmc}
\begin{center}
\begin{tabular}{c|c}
\textbf{\PCTL} & PTIME-complete\\
\hline
\textbf{\PCTLS} & PSPACE-complete \\
\end{tabular}
\end{center}
\end{table}

The model checking problem amounts to checking whether a given finite model satisfies a formula. For logics \PCTL\ and \PCTLS\ the model checking problem over Markov chains is known to be decidable, with the complexity results summarised in Table~\ref{tab:mcmc}. These logics thus often feature in probabilistic model checkers, e.g., PRISM \cite{KNP11} and Storm \cite{Katoen-Storm}. In contrast, the satisfiability problem, i.e.,\ the whether there exists a model that satisfies a given formula, for these logics is an open problem \cite{chakraborty-katoen}. 
Therefore. the satisfiability problem for the logic introduced here is likely to be challenging.

Consider a \SMGC\ $\cM$ in which agents never change their mental attitudes. Then, thanks to Assumption~\ref{assump:detbeh}, all transitions in the system are effectively deterministic and $\cM$ can be viewed as a Markov chain. Using complexity results summarised in Table~\ref{tab:mcmc} 
we formalise the above observation with the following theorem.

\begin{theorem}
If the cognitive strategies of all agents in a system $\cM$ are constant, then the complexity of model checking \PCTLSM\ over \SMGC\ is PSPACE-complete, and model checking \PCTLM\ over \SMGC\ is PTIME-complete.
\end{theorem}

However, it is often unrealistic to assume that agents' cognitive strategies are constant.
In Section~\ref{sec:proupdate}, we suggest a variation of our model based on agents reasoning about their beliefs and trust, all of whose components are finite, which makes it amenable for model checking.

\section{Preference Functions and Probability Spaces}\label{sec:preferencefuncs}

In this section, we develop the foundations for reasoning with \emph{probabilistic beliefs}, which we define in Section \ref{sec:beliefsupdate}.
In order to support \emph{subjective} beliefs, we utilise the concept of \textit{preference functions}, which resolve the nondeterminism arising from agents' cognitive transitions in a similar way to how action strategies resolve nondeterminism in the temporal dimension. 
This enables the definition of probability spaces to support reasoning about beliefs, which will in turn provide the basis for reasoning about trust. 
The central model of this paper, autonomous stochastic multiagent systems, is then introduced.

\paragraph{\bf Preference Functions}

To define a belief function, the usual approach is to employ for every agent a preference ordering, which is a measurement over the worlds. This measurement is commonly argued to naturally exist with the problem, see \cite{FH1997} for an example.
However, the actual definition of such a measurement 
can be non-trivial, because the number of worlds can be very large or even infinite, e.g.\ in \cite{FH1997} and in this paper, and thus enumerating all worlds can be infeasible.

In this paper, instead of an (infinite definition of) preference ordering, we resolve the nondeterminism in the system by introducing \emph{preference functions}. The key idea for a preference function is to estimate, for an agent $A$, the possible changes of goals or intentions of another agent $B$ in a given state using a probability distribution. In other words, preference functions model \emph{probabilistic prior} knowledge of agent $A$ about pro-attitudes of another agent $B$. That knowledge may be derived from prior experience (through observations), personal preferences, social norms, etc., and will in general vary between agents. A uniformly distributed preference function can be assumed if no prior information is available, as is typical in Bayesian settings.

Resolving the nondeterminism (for temporal dimension -- by Assumption~\ref{assump:detbeh}, for cognitive dimension -- by preference functions) allows us to define a family of probability spaces for each agent. Since preference functions vary among agents, probability spaces are also different for every agent. Moreover, each agent has multiple probability spaces, corresponding to its own various cognitive states. Intuitively, every probability space represents agent's \emph{subjective} view on the relative likelihood of different infinite paths being taken in the system.

We are now ready to define the central model of this paper.

\newcommand{\asmasTuple}{\cM = (\Ags, \sep\states, \sep\sinit, \sep\{\Act{A}\}_{A\in\Ags}, \sep T, \sep L, \sep\{\Obs{A}\}_{A\in\Ags}, \sep\{\obs{A}\}_{A\in\Ags}, \sep\{\Cognition_A\}_{A \in \Ags}, \sep \{\cstrat{A}\}_{A \in \Ags}, \sep \{\pref{A}\}_{A \in \Ags})}

\begin{definition}\label{def:autosmas}
An \emph{autonomous stochastic multi-agent system} (ASMAS) is a tuple 
$\cM = (G,\{\pref{A}\}_{A\in\Ags})$, where 
\begin{itemize}
\item $G = \smgcTuple$ is an \SMGC\ and 
\item $\pref{A}$ is a family of preference functions of agent $A \in \Ags$, defined as
\[
\pref{A} \defequals \{\gpref{A}{B}, \ipref{A}{B}~|~B\in\Ags \text{ and } B \neq A\},
\]
where:
\begin{itemize}
\item $\gpref{A}{B} : \states \rightarrow \dist{\powerset{\Goal{B}}}$ is a goal preference function of $A$ over $B$ such that, for any state $s$ and $x \in \powerset{\Goal{B}}$, we have $\gpref{A}{B}(s)(x)> 0$  only if  $x\in \lgoal{B}(s)$, and

\item $\ipref{A}{B} : \states \rightarrow \dist{\Intn{B}}$ is an intention preference function of $A$ over $B$ such that, for any state $s$ and $x \in \Intn{B}$, we have $\ipref{A}{B}(s)(x)> 0$  only if $x\in \lintn{B}(s)$.
\end{itemize}
\end{itemize} 
\end{definition}

\commentout{
\begin{definition}
Let $\cM$ be an \AutoSMAS\ with set of agents $\Ags$. For every agent $A\in\Ags$, we define a set of preference functions
\[
\pref{A} \defequals \{\gpref{A}{B}, \ipref{A}{B}~|~B\in\Ags \text{ and } B \neq A\},
\]
where
\begin{itemize}
\item $\gpref{A}{B} : \states \rightarrow \dist{\powerset{\Goal{B}}}$ is such that, for any state $s$ and $x \in \powerset{\Goal{B}}$, we have $\gpref{A}{B}(s)(x)> 0$  only if  $x\in \Goal{B}(s)$, and

\item $\ipref{A}{B} : \states \rightarrow \dist{\Intn{B}}$ is such that, for any state $s$ and $x \in \Intn{B}$, we have $\ipref{A}{B}(s)(x)> 0$  only if $x\in \Intn{B}(s)$.
\end{itemize}
\end{definition}
}

\begin{remark}
During system execution, preference functions may be updated as agents learn new information through interactions; we will discuss in Section~\ref{sec:proupdate} how this can be implemented in our framework.
\end{remark}

Intuitively, a preference function provides agent $A$ with a probability distribution over another agent $B$'s changes of pro-attitudes. 
Naturally, we expect preference functions to be consistent with partial observability. We therefore extend the \textit{Uniformity Assumption} in the following way.

\begin{assumption}\label{assump:uniII} (Uniformity Assumption II)
Let $\asmasTuple$ be an \AutoSMAS.
For an agent $A \in \Ags$ and any two states $s_1, s_2 \in \states$, we assume that $\obs{A}(s_1)=\obs{A}(s_2)$ implies $\gpref{B}{A}(s_1)(x) = \gpref{B}{A}(s_2)(x)$ and $\ipref{B}{A}(s_1)(x) = \ipref{B}{A}(s_2)(x)$ for any $B \in \Ags$ such that $B \neq A$ and $x \subseteq \Goal{A}$. That is, agents' preferences over a given agent are the same on all paths which a given agent cannot distinguish.
\end{assumption}

We also mention that Assumption~\ref{assump:detbeh} (Deterministic Behaviour Assumption) extends to \AutoSMAS\ in a straightforward manner.

\paragraph{\bf Transition Type}

\newcommand{\type}{{tp}}
\newcommand{\ttrans}{{\tt tt}}
In general, in any state of the system, an agent may choose a temporal or a cognitive transition. However, it is often desirable, e.g., when constructing a probability space, to restrict the \textit{type} of transition available to an agent.
\begin{definition}
Given a path $s_0s_1s_2...$ in an \AutoSMAS\ $\cM$ satisfying Assumption~\ref{assump:detbeh} and an agent $A \in \Ags$, we use $\tp{A}(s_k,s_{k+1})$ to denote the type, as seen by agent $A$, of the transition that is taken to move from state $s_k$ to $s_{k+1}$. More specifically, we distinguish five different \emph{transition types}: 
\begin{itemize}
\item $\tp{A}(s_k,s_{k+1})=a$ if $s_k \trans{a} s_{k+1}$ for some $a\in\Act{}$,
\item $\tp{A}(s_k,s_{k+1})=A.g.x$ if $s_k\ctrans{A.g.x} s_{k+1}$ for some $x \subseteq \lgoal{A}(s_k)$, 
\item $\tp{A}(s_k,s_{k+1})=A.i.x$ if $s_k\ctrans{A.i.x} s_{k+1}$ for some $x\in \lintn{A}(s_k)$,
\item $\tp{A}(s_k,s_{k+1})=B.g$ if $s_k\ctrans{B.g.x} s_{k+1}$ for another agent $B \in \Ags$ and $x \subseteq \lgoal{B}(s_k)$, 
\item $\tp{A}(s_k,s_{k+1})=B.i$ if $s_k\ctrans{B.i.x} s_{k+1}$ for another agent $B \in \Ags$ and $x\in \lintn{B}(s_k)$.
\end{itemize}
\end{definition}

\begin{remark}
Assumption~\ref{assump:detbeh} guarantees that, if the transition form $s_k$ to $s_{k+1}$ is temporal, then the action $a$ is uniquely determined.
\end{remark}

We write $\tp{A}(\rho)=\tp{A}(\rho(0),\rho(1)) \cdot \tp{A}(\rho(1),\rho(2)) \cdot ...$ for the type of a path $\rho$. When $\rho$ is a finite path and $t$ is a type of an infinite path, we say that $\rho$ is \textit{consistent} with $t$ if there exists an infinite extension $\rho' = \rho \delta$ of $\rho$ such that $\tp{A}(\rho') = t$.

We note that the type of agent's own cognitive transitions is defined differently than the type of another agent's cognitive transitions. This stems from our implicit assumption that agents can observe their own cognitive changes, but they cannot, in general, observe another agent's cognitive transitions. Formally, we make the following assumption.
\begin{assumption}\label{assump:typeObs} (Transition Type Distinguishability Assumption)
Let $\asmasTuple$ be an \AutoSMAS. For an agent $A \in \Ags$ and any two finite paths $\rho_1$, $\rho_2$, we assume that $\obs{A}(\rho_1) = \obs{A}(\rho_2)$ implies $\tp{A}(\rho_1) = \tp{A}(\rho_2)$. That is, agents can distinguish paths of different types.
\end{assumption}

This assumption is essential to ensure that the belief function is defined over finite paths in the same probability space.

\paragraph{\bf Initial Distribution}

We replace the set of initial states with a more accurate notion of \textit{initial distribution} $\initdist \in \dist{\sinit}$, as a common prior assumption of the agents about the system, which is needed to define a belief function.

\paragraph{\bf Probability Spaces}
We now define a family of probability spaces for an arbitrary \AutoSMAS. Recall that our ultimate goal is to define the belief function. To motivate the following construction, we mention that the belief function provides a probability distribution over paths which are indistinguishable (i.e., have the same observation) to a given agent. There are two important consequences of this: (i) probability spaces will vary among agents, reflecting the difference in their partial observation, and (ii) rather than defining a single probability space, spanning all the possible paths, for a given agent, we may define many of them, as long as every pair of indistinguishable paths lies in the same probability space. 
Hence, using the fact that agents can observe the type of transitions (by Assumption~\ref{assump:typeObs}), we parameterise probability spaces by type, so that each contains paths of unique type. 

We begin by introducing an auxiliary transition function, specific to each agent, which will be used to define the probability measure.

\begin{definition}
Let $\asmasTuple$ be an \AutoSMAS. Based on temporal transition function $T$ and preference functions $\{\pref{A}\}_{A \in \Ags}$, we define an auxiliary \emph{transition function} $T_A$ for agent $A \in \Ags$ as follows for $s,s'\in S$:
\[
\displaystyle
T_A(s,s') =
\left \{
\begin{array}{ll}
T(s,a)(s')   & \text{if } \tp{A}(s,s')=a \\[-0.1em]
\gpref{A}{B}(s)(x) & \text{if } \tp{A}(s,s')=B.g \text{ and } s \ctrans{B.g.x} s' \\[0.3em]
\ipref{A}{B}(s)(x) & \text{if } \tp{A}(s,s')=B.i \text{ and } s \ctrans{B.i.x} s' \\
1 & \text{if } \tp{A}(s,s')=A.g.x \text{ for some } x \in \lgoal{A}(s) \\
 & \text{or } \tp{A}(s,s')=A.i.x \text{ for some } x \in \lintn{A}(s)
\end{array}
\right.
\]
\end{definition}

The application of the function $T_A$ resolves the nondeterminism in both the temporal and cognitive dimensions and the resulting system is a family of probability spaces, each containing a set of paths of the same type.
While the probabilities of temporal transitions are given by the transition function $T$ and computed by all agents in the same way, each agent resolves the nondeterminism in the cognitive dimension differently. For an agent $A$, all possible cognitive transitions of another agent $B$ have the same type and hence lie in the same probability space; $A$ uses its preference functions to associate probabilities to them. On the other hand, each of $A$'s own possible pro-attitude changes has a different type and lies in a different probability space; $A$ therefore treats it as a deterministic transition and assigns probability 1 to it. 


We now construct a probability space for an arbitrary \AutoSMAS\ $\cM$, agent $A$ and a type $t$ in a similar way as in Section~\ref{sec:smas}. 
The sample space consists of infinite paths of type $t$ starting in one of the initial states, i.e., $\sspace{\cM,t}=\bigcup_{s\in\states, \initdist(s)>0}\{\delta\in \infpath{}^{\cM}(s)~|~\tp{A}(\delta) = t\}$. We associate a probability to each finite path $\rho = s_0...s_n$ consistent with $t$ via function $\probm{A}(\rho) = \initdist(\rho(0))\cdot  \prod_{0\leq i\leq |\rho|-2} T_A(\rho(i),\rho(i+1))$. 
We then set $\measurable{\cM}$ to be the smallest $\sigma$-algebra generated by cylinders $\{\cylinder_{\rho} \cap \sspace{\cM,t} \mid \rho \in \fpath{}^{\cM}\}$ and $\probs{A}{\cM}$ to be the unique measure such that $\probs{A}{\cM}(\cylinder_{\rho}) = \probm{A}(\rho)$. It then follows that $(\sspace{\cM,t},\measurable{\cM},\probs{A}{\cM})$ is a probability space.

Note that, for agents $A$ and $B$ such that $A\neq B$ in a system $\cM$, agents' probability spaces will in general differ in probability measures $\probs{A}{\cM}$ and $\probs{B}{\cM}$, because their preference functions may be different. 

\begin{example}\label{example:trustgamepref} 
We now define preference functions for Bob and Alice of Example \ref{example:trustgame}. For example, setting:
\[
\gpref{Bob}{Alice}(s_0) = \langle passive \mapsto 1/3, active \mapsto 2/3 \rangle
\]
indicates that Bob believes Alice is more likely to be \emph{active} than \emph{passive}. 
Similarly, we give Alice's preference functions.
We first note that 
$\obs{Bob}(s_1) = \obs{Bob}(s_2)$. Therefore, by Assumption~\ref{assump:uniII}, $\gpref{Alice}{Bob}(s_1) = \gpref{Alice}{Bob}(s_2)$. Setting:
\[
\gpref{Alice}{Bob}(s_k) = \langle investor \mapsto 1/2, opportunist \mapsto 1/2 \rangle,
\]
for $k \in \{1,2\}$, represents that Alice has no prior knowledge regarding Bob's mental attitudes.
Finally, we define Alice's intention preference over Bob. Since $\obs{Bob}(s_8) = \obs{Bob}(s_{12})$ and $\obs{Bob}(s_{10}) = \obs{Bob}(s_{14})$, Assumption~\ref{assump:uniII} implies that $\ipref{Alice}{Bob}(s_8) = \ipref{Alice}{Bob}(s_{12})$ and $\ipref{Alice}{Bob}(s_{10}) = \ipref{Alice}{Bob}(s_{14})$. We may set:
\begin{alignat*}{2}
\ipref{Alice}{Bob}(s_k) & = \langle share \mapsto 3/4, keep \mapsto 1/4 \rangle && \quad \text{ for } k \in \{8,12\}, \\
\ipref{Alice}{Bob}(s_k) & = \langle share \mapsto 0, keep \mapsto 1 \rangle && \quad \text{ for } k \in \{10,14\} 
\end{alignat*}
to indicate that Alice knows that Bob will keep the money when he is an \emph{opportunist}, but she thinks it is quite likely that he will share his profit when he is an \emph{investor}.

We now compute the probability that Alice and Bob will have the same amount of money at the end of the game. In other words, we want to find the probability that Alice invests her money with Bob and Bob shares his profit with her. As noted above, probability spaces differ between agents. We first perform the computation from Alice's point of view and consider two cases: (i) Alice being \emph{passive} and (ii) Alice being \emph{active}. For (i), letting $\rho_1 = s_0s_1s_3s_8s_{15}s_{24}$ and $\rho_2 = s_0s_1s_4s_{10}s_{17}s_{28}$, we compute:
\begin{equation*}
\begin{split}
\probm{Alice}(\rho_1) 
  & =  \gpref{Alice}{Bob}(s_1)(investor) \cdot (\astrat{passive}(s_0s_1s_3)(invest) \cdot T(s_3,invest)(s_8)) \\
  & ~\cdot~ \ipref{Alice}{Bob}(s_8)(share) 
   ~\cdot~ (\astrat{share}(s_0s_1s_3s_8s_{15})(share) \cdot T(s_{15},share)(s_{24})) \\
  & =     \frac{1}{2} \cdot (\frac{3}{10} \cdot 1) \cdot \frac{3}{4} \cdot (1 \cdot 1) = \frac{9}{80}, \\
\probm{Alice}(\rho_2) 
  & =  \gpref{Alice}{Bob}(s_1)(opportunist) \cdot (\astrat{passive}(s_0s_1s_4)(invest) \cdot T(s_4,invest)(s_{10})) \\
  & ~\cdot~ \ipref{Alice}{Bob}(s_{10})(share) 
   ~\cdot~ (\astrat{share}(s_0s_1s_4s_{10}s_{17})(share) \cdot T(s_{17},share)(s_{28})) \\
  & =     \frac{1}{2} \cdot (\frac{3}{10} \cdot 1) \cdot 0 \cdot (1 \cdot 1) = 0.
\end{split}
\end{equation*}

Similarly, in case (ii), letting $\rho_3 = s_0s_2s_5s_{12}s_{19}s_{32}$ and $\rho_3 = s_0s_2s_6s_{14}s_{21}s_{36}$, we compute:
\begin{equation*}
\begin{split}
\probm{Alice}(\rho_3) 
  & =  \gpref{Alice}{Bob}(s_2)(investor) \cdot (\astrat{active}(s_0s_2s_5)(invest) \cdot T(s_5,invest)(s_{12})) \\
  & ~\cdot~ \ipref{Alice}{Bob}(s_{12})(share) 
   ~\cdot~ (\astrat{share}(s_0s_2s_5s_{12}s_{19})(share) \cdot T(s_{19},share)(s_{32})) \\
  & =     \frac{1}{2} \cdot (\frac{9}{10} \cdot 1) \cdot \frac{3}{4} \cdot (1 \cdot 1) = \frac{27}{80}, \\
\probm{Alice}(\rho_4) 
  & =  \gpref{Alice}{Bob}(s_2)(opportunist) \cdot (\astrat{active}(s_0s_2s_6)(invest) \cdot T(s_6,invest)(s_{14})) \\
  & ~\cdot~ \ipref{Alice}{Bob}(s_{14})(share) 
   ~\cdot~ (\astrat{share}(s_0s_2s_6s_{14}s_{21})(share) \cdot T(s_{21},share)(s_{36})) \\
  & =     \frac{1}{2} \cdot (\frac{9}{10} \cdot 1) \cdot 0 \cdot (1 \cdot 1) = 0.
\end{split}
\end{equation*}

Hence, from Alice's point of view, the probability that both will have 20 dollars at the end of the game is three times higher when she is \emph{active}, and roughly equal to 1/3. This is consistent with our expectations, since the only difference between the two scenarios is the likelihood of her investing, which is three times greater when she is \emph{active}.

We now perform a similar computation for Bob. Again, we consider (i) Bob being an \emph{investor} and (ii) Bob being an \emph{opportunist}. We start with case (i), but we have to be a little more careful now. In order for the result to be meaningful, all paths in which Bob is an \emph{investor} must be in the same probability space. However, currently, this is not the case for paths corresponding to different intention changes for Bob. To fix that, we use Bob's intention strategy $\intn{Bob}$ to quantify the likelihoods of paths in different probability spaces. We may then compute:
\begin{equation*}
\begin{split}
\probm{Bob}(\rho_1) 
  & =  \gpref{Bob}{Alice}(s_0)(passive) \cdot (\astrat{passive}(s_0s_1s_3)(invest) \cdot T(s_3,invest)(s_8)) \\
  & ~\cdot~ \intn{Bob}(s_0s_1s_3s_8)(share) \\
  & ~\cdot~ (\astrat{share}(s_0s_1s_3s_8s_{15})(share) \cdot T(s_{15},share)(s_{24})) \\
  & =     \frac{1}{3} \cdot (\frac{3}{10} \cdot 1) \cdot 1 \cdot (1 \cdot 1) = \frac{1}{10}, \\
\probm{Bob}(\rho_3) 
  & =  \gpref{Bob}{Alice}(s_0)(active) \cdot (\astrat{active}(s_0s_2s_5)(invest) \cdot T(s_5,invest)(s_{12})) \\
  & ~\cdot~ \intn{Bob}(s_0s_2s_5s_{12})(share) \\
  & ~\cdot~ (\astrat{share}(s_0s_2s_5s_{12}s_{19})(share) \cdot T(s_{19},share)(s_{32})) \\
  & =     \frac{2}{3} \cdot (\frac{9}{10} \cdot 1) \cdot 1 \cdot (1 \cdot 1) = \frac{3}{5}.
\end{split}
\end{equation*}

In case (ii), Bob chooses intention $\astrat{keep}$ with probability 1 and so the probability of him sharing the profit is 0. \hfill $\Box$

\end{example}

\section{Reasoning with Probabilistic Beliefs}\label{sec:beliefsupdate}

In presence of partial observation and, resulting from it, imperfect information that agents have about system state, one resorts to using beliefs to represents agents' knowledge.
In this section, we define a \emph{probabilistic belief} function $\be_A$ of an agent $A$ in an autonomous stochastic multiagent system (\AutoSMAS), which models agents' uncertainty about the system state and its execution history. 

In \cite{FH1997}, a qualitative belief notion is defined based on a plausibility measure. We consider probability measure instead, and define a \emph{quantitative} belief function expressing that agent $A$ believes $\phi$ with probability in relation $\bowtie$ to $q$ if it \emph{knows} $\phi$ with probability in relation $\bowtie$ to $q$. This can  be seen as a quantitative variant of the qualitative definition in~\cite{Pearl1994}, according to which agent $A$ believes $\phi$ if the probability of $\phi$ is close to 1. It is generally accepted that trust describes agents' subjective probability~\cite{Gambetta1990} and therefore needs to be measured quantitatively.

\subsection{Belief Function}\label{sec:beliefdefn}
Before defining the belief function, we establish an equivalence relation $\obseq{A}$ on the set of finite paths of an autonomous system $\cM$ for an agent $A$. We say that paths $\rho, \rho' \in \fpath{}^\cM$ are equivalent (from the point of view of $A$), denoted $\rho \obseq{A} \rho'$, if and only if $\obs{A}(\rho) = \obs{A}(\rho')$. Intuitively, each equivalence class contains paths which a given agent cannot differentiate from each other. Note that, by Assumption~\ref{assump:typeObs}, each equivalence class comprises paths of the same type.

\begin{example}\label{example:equivClasses}
Recall the trust game $\trustgame$ from Example~\ref{example:trustgamectrans} and the observation function defined on $\trustgame$ in Example~\ref{example:trustgameobs}. Using that definition, we obtain $\obs{Bob}(s_0s_1) = \obs{Bob}(s_0s_2)$, which is intuitively correct, since Bob cannot observe Alice's goals. These two paths form an equivalence class. Similarly, $\obs{Bob}(s_0s_1s_3s_8) = \obs{Bob}(s_0s_2s_5s_{12})$ and $\obs{Bob}(s_0s_1s_4s_{10}) = \obs{Bob}(s_0s_2s_6s_{14})$, and both pairs of paths form an equivalence class. \hfill $\Box$
\end{example}

\newcommand{\eqclass}{class}

The belief function quantifies agent's belief about the execution history of the system at any given time.
Intuitively, at any point during system execution, agent's observation history uniquely identifies an equivalence class (with respect to $\obseq{A}$), consisting of all paths with that observation. The belief function gives a probability distribution on that equivalence class. We define $\pathobs{A} = \{ \obs{A}(\rho) \mid \rho \in \fpath{}^\cM \}$ to be a set of all finite observation histories (path observations), and for a given path observation $o \in \pathobs{A}$, $\eqclass(o)$ denotes the equivalence class associated with $o$. Recall that, for any finite path $\rho$, $\cylinder_{\rho}$ denotes a basic cylinder with prefix $\rho$.

\begin{definition}
Let $\cM$ be an \AutoSMAS\ and $A$ an agent in $\cM$. The \emph{belief function} $\be_A : \pathobs{A} \rightarrow \dist{\fpath{}^\cM}$ of an agent $A$ is given by:
\[
\be_A(o)(\rho) = \probs{A}{\cM}(\cylinder_{\rho}~|~\bigcup_{\rho' \in \eqclass(o)}\cylinder_{\rho'}).
\]
\end{definition}
Hence, given that agent $A$'s path observation is $o$ at some point, 
their belief that the execution history at that time is $\rho$ is expressed as the \emph{conditional probability} of the execution history being $\rho$, given that the execution history belongs to the equivalence class $\eqclass(o)$. We note that sometimes, e.g., when the observation is clear from the context, we might omit it, and simply write $\be_A(\rho)$ for agent $A$'s belief that the execution history at some point is $\rho$.

\begin{example}\label{example:trustgamebelief}
We return again to Example~\ref{example:trustgamectrans} and let $\rho_1 = s_0s_1$ and $\rho_2 = s_0s_2$. Recall from Example~\ref{example:equivClasses} that $\obs{Bob}(\rho_1) = \obs{Bob}(\rho_2)$, and let $o_1$ denote that common observation. Hence, $o_1$ is a path observation associated with the equivalence class $\{\rho_1, \rho_2\}$. We compute $\be_{Bob}(o_1,\rho_1)$ and $\be_{Bob}(o_1,\rho_2)$ below. 
\begin{equation*}
\begin{split}
\be_{Bob}(o_1,\rho_1) 
   & = \probs{Bob}{\trustgame}(\cylinder_{\rho_1}~|~\bigcup_{\rho \in class(o)}\cylinder_{\rho}) \\
   & = \frac{ \probs{Bob}{\trustgame}(\cylinder_{\rho_1})}{\probs{Bob}{\trustgame}(\cylinder_{\rho_1}) + \probs{Bob}{\trustgame}(\cylinder_{\rho_2})} \\
   & = \frac{\gpref{Bob}{Alice}(s_0)(passive)}{\gpref{Bob}{Alice}(s_0)(passive) + \gpref{Bob}{Alice}(s_0)(active)} \\
   & = \frac{1}{3}.
\end{split}
\end{equation*}

Similar computation shows that:
\[
\be_{Bob}(o_1,\rho_2) = \frac{2}{3}.
\]

We also let $\rho_3 = s_0s_1s_3s_8$, $\rho_4 = s_0s_2s_5s_{12}$ and recall from Example~\ref{example:equivClasses} that $\obs{Bob}(\rho_3) = \obs{Bob}(\rho_4)$. We let $o_2$ denote that common observation and, performing similar computation as above, we obtain:
\begin{align*}
\be_{Bob}(o_2,\rho_3) &= \frac{1}{7}, \\
\be_{Bob}(o_2,\rho_4) &= \frac{6}{7}.
\end{align*}

Analogously, letting $\rho_5 = s_0s_1s_4s_{10}$, $\rho_6 = s_0s_2s_6s_{14}$ and $o_3 = \obs{Bob}(\rho_5) = \obs{Bob}(\rho_6)$, we compute:
\begin{align*}
\be_{Bob}(o_3,\rho_5) &= \frac{1}{7}, \\
\be_{Bob}(o_3,\rho_6) &= \frac{6}{7}.
\end{align*} \hfill $\Box$

\end{example}

The following theorem gives a recursive definition of the belief function. Intuitively, the beliefs of agent $A$ are represented as a distribution over a set of possible states which are reachable by paths consistent with the current observation history. This belief distribution will be updated in the Bayesian way whenever a new observation is taken.
\begin{theorem}
The beliefs $\be_A(\rho)$ can be computed recursively over the length of path $\rho$:
\begin{itemize}
\item For an initial state $s$,
$
\be_A(s)=\displaystyle\frac{\initdist(s)}{\sum_{s'\in S \land \obs{A}(s') = \obs{A}(s)}\initdist(s')}.
$

\item For $\rho$ a path such that $|\rho| = k+2$ and $k\geq 0$,
\[
\be_A(\rho) = \frac{\be_A(\rho[0..k])\times T_A(\rho(k),\rho(k+1))}{\sum_{\rho'\in \fpath{}^\cM \land \obs{A}(\rho') = \obs{A}(\rho[0..k])}\sum_{s\in \states \land \obs{A}(s) = \obs{A}(\rho(k+1))} \be_A(\rho')\times T_A(\rho'(k),s)}.
\]
\end{itemize}
\end{theorem}


\subsection{Belief \AutoSMAS}\label{sec:beliefasmas}
A conceptually simpler construction, defined below, can be effectively used to reason about \emph{a single agent's} understanding about the system, when there is no need for nested reasoning on agents' beliefs. 
We employ a well-known construction \cite{DBLP:journals/jcss/ChatterjeeCT16} which, for a system $\cM$ and a fixed agent $A$, induces an equivalent \textit{belief \AutoSMAS} $\Bel{A}(\cM)$ whose states are \textit{distributions} over states of $\cM$ called \textit{belief states}. Intuitively, belief states quantify agent's uncertainty about the current state of a system by specifying their belief, expressed as a probability distribution: although we may not know which of several observationally-equivalent states we
are currently in, we can determine the likelihood of being in each one. The belief \AutoSMAS\ is fully observable but its state space is possibly infinite.

Formally, a belief \AutoSMAS\ for an \AutoSMAS\ $\cM$ and an agent $A$ is a tuple
$\Bel{A}(\cM) = (\Ags, \sep \dist{\states}, \sep \beliefsinit, \sep \{\Act{A}\}_{A\in\Ags}, 
 \sep \TBel{A}, \sep L)$, where $\beliefsinit = \{\beliefinit^\mathrm{o}~|~o\in\Obs{A}\}$ is a set of initial belief states such that:
$$
\beliefinit^\mathrm{o}(s) = \left \{
\begin{array}{ll}
\displaystyle \frac{\initdist(s)}{\sum_{s'\in \states \& \obs{A}(s')=\obs{A}(s)}\initdist(s')} & \text{ if } \obs{A}(s) = o\\
0 & \text{ otherwise. }
\end{array}
\right.
$$


We now give the transition function $\TBel{A}$. Intuitively, from a belief state $\belief$, different belief states $\belief'$ are possible, each corresponding to a unique combination of $A$'s next observation and the type of transition taken.  Transition probabilities are then computed based on the temporal transition function $T$ or $A$'s preferences. The definition is split into three parts, reflecting the different types of transitions present in an \AutoSMAS. 

First, for $\belief, \belief' \in \dist{\states}$ and $a \in \Act{}$, we have

\[
\TBel{A}(\belief,a)(\belief') = \sum_{s\in S} \belief(s) \cdot (\sum_{o\in \Obs{A} \& \belief^{a,o} = \belief'} \sum_{s'\in S \& \belief'(s') >0} T(s,a)(s')),
\]
where $\belief^{a,o}$ is a belief state reached from $\belief$ by performing $a$ and observing $o$, i.e.,
\[
\displaystyle
\belief^{a,o}(s')=\left \{
\begin{array}{ll}
\displaystyle
\frac{\sum_{s\in S}\belief(s) \cdot T(s,a)(s')}{\sum_{s\in S}\belief(s) \sum_{s''\in S \& \obs{A}(s'') = o} T(s,a)(s'')} & \text{ if } \obs{A}(s') = o \\
0 & \text{ otherwise.}
\end{array}
\right.
\]
Intuitively, the belief state reached by agent $A$ from $\belief$ by performing $a$ and observing $o$ is a probability distribution over states $s$ such that $\obs{A}(s) = o$, where the probability of each state is weighted according to the probabilities of transitions from states in the support of $\belief$ to that state.


Second, for $x \subseteq \Goal{A}$,
\[
\TBel{A}(\belief,A.g.x)(\belief') = \sum_{s\in S} \belief(s) \cdot (\sum_{o\in \Obs{A} \& \belief^{\mathrm{A.g.x},\mathrm{o}} = \belief'} \sum_{s'\in S \& \belief'(s') >0} 
(s\ctrans{A.g.x} s')),
\]
where $\belief^{A.g.x,o}$ is belief reached from $\belief$ by performing $A$'s goal change into $x$ and observing $o$, i.e.,
\[
\displaystyle
\belief^{A.g.x,o}(s')=\left \{
\begin{array}{ll}
\displaystyle
\frac{\sum_{s\in S}\belief(s) \cdot 
(s\ctrans{A.g.x} s')}{\sum_{s\in S}\belief(s) \sum_{s''\in S \& \obs{A}(s'') = o} 
(s\ctrans{A.g.x} s'')} & \text{ if } \obs{A}(s') = o \\
0 & \text{ otherwise.}
\end{array}
\right.
\]
Note the we interpret $s \ctrans{A.g.x} s'$ as having value 0 or 1, depending on whether a transition of a specified type exists or not.
The case corresponding to an intention change of agent $A$ is defined analogously. 

Third, for $B \in \Ags$ such that $B \neq A$,

\[
\TBel{A}(\belief,B.g)(\belief') = \sum_{s\in S} \belief(s) \cdot (\sum_{o\in \Obs{A} \& \belief^{B.g,o} = \belief'} \sum_{s'\in S \& \belief'(s') >0} \sum_{x \in \lgoal{A}(s)} \gpref{A}{B}(s)(x) \cdot (s\ctrans{B.g.x} s')),
\]
where $\belief^{B.g,o}$ is belief reached from $\belief$ by performing $B$'s goal change into $x$ and observing $o$, i.e.,
\[
\displaystyle
\belief^{B.g,o}(s')=\left \{
\begin{array}{ll}
\displaystyle
\frac{\sum_{s\in S}\belief(s) \sum_{x} \gpref{A}{B}(s)(x) \cdot (s\ctrans{B.g.x} s')}{\sum_{s\in S}\belief(s) \sum_{s''\in S \& \obs{A}(s'') = o} \sum_{x} \gpref{A}{B}(s)(x) \cdot (s\ctrans{B.g.x} s'')} & \text{ if } \obs{A}(s') = o \\
0 & \text{ otherwise.}
\end{array}
\right.
\]
The case of intention change can be defined in a similar manner.


\begin{example}
We illustrate the above construction on the trust game $\trustgame$ from Example~\ref{example:trustgamectrans} by considering the initial few states of the belief \AutoSMAS\ $\Bel{Bob}(\trustgame)$ corresponding to $\trustgame$.
Since the game has a single initial state, the belief \AutoSMAS\ also has a single initial belief state, namely
$\belief_0 = \beliefinit^{o_1} = \langle s_0 \mapsto 1 \rangle$. 

We now compute possible belief states after the first transition, i.e. Alice's goal change. Recall that Bob cannot distinguish states $s_1$ and $s_2$; we let $o_1 = \obs{Bob}(s_1) = \obs{Bob}(s_2)$. Then the new belief state is $\belief_1 = \belief_0^{Alice.g,o_1}$, with:
\begin{equation*}
\begin{split}
\belief_1(s_1) = \belief_0^{Alice.g,o_1}(s_1)
 & = \frac{\belief_0(s_0) \cdot \gpref{Bob}{Alice}(s_0)(\{passive\})}{\belief_0(s_0) \cdot (\gpref{Bob}{Alice}(s_0)(\{passive\}) + \gpref{Bob}{Alice}(s_0)(\{active\}))} \\
 & = \frac{\frac{2}{3}}{\frac{2}{3} + \frac{1}{3}} = \frac{2}{3},
\end{split}
\end{equation*}
where we used Bob's preference function $\gpref{Bob}{Alice}$ defined in Example~\ref{example:trustgamepref}. Similar computation shows that:
\[
\belief_1(s_2) = \frac{1}{3}.
\]

Since Bob has only one possible observation at this point of the game, we expect that $\belief_1$ is his only possible belief state after first transition, i.e. $\TBel{Bob}(\belief_0,Alice.g)(\belief_1) = 1$. Indeed,
\begin{equation*}
\begin{split}
\TBel{Bob}(\belief_0,Alice.g)(\belief_1) & = \belief_0(s_0) \cdot (\gpref{Bob}{Alice}(s_0)(\{passive\}) + \gpref{Bob}{Alice}(s_0)(\{active\})) \\
 & = 1.
\end{split}
\end{equation*}

Next, Bob sets his goal. We let $o_2 = \obs{Bob}(s_3) = \obs{Bob}(s_5)$ and $o_3 = \obs{Bob}(s_4) = \obs{Bob}(s_6)$. If Bob observes $o_2$, then his belief state is $\belief_2 = \belief_1^{Bob.g,o_2}$, such that:
\begin{equation*}
\begin{split}
\belief_2 = \langle s_3 \mapsto 2/3, s_5 \mapsto 1/3 \rangle.
\end{split}
\end{equation*}

If he observes $o_3$ then his belief state is $\belief_3= \belief_1^{Bob.g,o_3}$, with:
\begin{equation*}
\begin{split}
\belief_3 = \langle s_4 \mapsto 2/3, s_6 \mapsto 1/3 \rangle.
\end{split}
\end{equation*}

The transition function is as follows.
\begin{equation*}
\begin{split}
\TBel{Bob}(\belief_1,Bob.g.\{investor\})(\belief_2) & = 1, \\
\TBel{Bob}(\belief_1,Bob.g.\{opportunist\})(\belief_3) & = 1.
\end{split}
\end{equation*}

Finally, we consider how belief state changes after Alice's temporal transition. There are four possible observations for Bob at this stage, namely $o_4 = \obs{Bob}(s_7) = \obs{Bob}(s_{11})$, $o_5 = \obs{Bob}(s_8) = \obs{Bob}(s_{12})$, $o_6 = \obs{Bob}(s_9) = \obs{Bob}(s_{13})$, $o_7 = \obs{Bob}(s_{10}) = \obs{Bob}(s_{14})$. Hence there are four possible belief states, $\belief_4= \belief_2^{Bob.g,o_4}$, $\belief_5= \belief_2^{Bob.g,o_5}$, $\belief_6= \belief_3^{Bob.g,o_6}$, $\belief_7= \belief_3^{Bob.g,o_7}$, such that:
\begin{align*}
\belief_4 &= \langle s_7 \mapsto 2/3, s_{11} \mapsto 1/3 \rangle, \\
\belief_5 &= \langle s_8 \mapsto 2/3, s_{12} \mapsto 1/3 \rangle, \\
\belief_6 &= \langle s_9 \mapsto 2/3, s_{13} \mapsto 1/3 \rangle, \\
\belief_7 &= \langle s_{10} \mapsto 2/3, s_{14} \mapsto 1/3 \rangle.
\end{align*}

Figure~\ref{fig:trustgamebelief} depicts belief states of $\Bel{Bob}(\trustgame)$ which we have computed.

\begin{figure}
\begin{center}
\begin{tikzpicture}[-latex ,auto ,node distance =1.5 cm and 1.5 cm,on grid,semithick,state/.style ={ circle,draw,minimum size =0.4 cm,scale=0.9}]
\node[state] (0A) {$\belief_0$};

\node[state] (1A) [below=of 0A] {$\belief_1$};

\node[state] (2A) [below left=of 1A] {$\belief_2$};
\node[state] (2B) [below right=of 1A] {$\belief_3$};

\newcommand{\threeShiftX}{0.7cm}
\node[state] (3A) [below left=of 2A,xshift=\threeShiftX] {$\belief_4$};
\node[state] (3B) [below right=of 2A,xshift=-\threeShiftX] {$\belief_5$};
\node[state] (3C) [below left=of 2B, xshift=\threeShiftX] {$\belief_6$};
\node[state] (3D) [below right=of 2B,xshift=-\threeShiftX] {$\belief_7$};

\tikzset{every edge/.append style={font=\footnotesize}}
\path (0A) edge [] node[] {A.g} (1A);
\path (1A) edge [] node[swap] {B.g.\{investor\}} (2A);
\path (1A) edge [] node[] {B.g.\{opportunist\}} (2B);

\path (2A) edge [] node[swap] {invest} (3A);
\path (2A) edge [] node[] {withhold} (3B);
\path (2B) edge [] node[swap] {invest} (3C);
\path (2B) edge [] node[] {withhold} (3D);

\node[draw,align=left,style={font=\footnotesize}] at (6,-2) 
{ $\belief_0 = \langle s_0 \mapsto 1 \rangle$ \\ 
 $\belief_1 = \langle s_1 \mapsto \frac{2}{3}, s_2 \mapsto \frac{1}{3} \rangle$ \\ 
 $\belief_2 = \langle s_3 \mapsto \frac{2}{3}, s_5 \mapsto \frac{1}{3} \rangle$ \\
 $\belief_3 = \langle s_4 \mapsto \frac{2}{3}, s_6 \mapsto \frac{1}{3} \rangle$ \\
 $\belief_4 = \langle s_7 \mapsto \frac{2}{3}, s_{11} \mapsto \frac{1}{3} \rangle$ \\
 $\belief_5 = \langle s_8 \mapsto \frac{2}{3}, s_{12} \mapsto \frac{1}{3} \rangle$ \\
 $\belief_6 = \langle s_9 \mapsto \frac{2}{3}, s_{13} \mapsto \frac{1}{3} \rangle$ \\
 $\belief_7 = \langle s_{10} \mapsto \frac{2}{3}, s_{14} \mapsto \frac{1}{3} \rangle$ };
\end{tikzpicture}
\end{center}
\caption{Belief \AutoSMAS\ $\Bel{Bob}(\trustgame)$}
\label{fig:trustgamebelief}
\end{figure}
\hfill $\Box$
\end{example}

As in our construction of the probability measure $\probm{s,\astrat{}}$ for \SMG s in Section~\ref{sec:smas}, a strategy profile $\astrat{}$ induces a probability measure $\probm{\belief,\astrat{}}$ on any belief state $\belief$ over infinite paths of \AutoSMAS\ $\Bel{A}(\cM)$. We note that, in general, $\Bel{A}(\cM)$ has a (continuous) infinite state space.


\section{Reasoning about Beliefs and Trust}\label{sec:pref}

In this section, we formalise the notions of trust, which are inspired by the social trust theory of~\cite{FC2001}. We define trust in terms of the belief function and further extend \PCTLSM to capture reasoning about beliefs and trust. 
%
%
%


In Section~\ref{sec:beliefsupdate}, a Bayesian-style definition of the belief function is presented. It uses agents' partial observation function and probability measure, which is in turn defined in terms of preference functions introduced in Section~\ref{sec:preferencefuncs}. 
In this section, we assume the existence of such a belief setting, i.e., one belief function for every agent, and define operators to reason about agents' beliefs and trusts.


\paragraph{\bf Expressing Trust}

Our trust operators follow the intuition of social trust concepts from \cite{FC2001}, but the technical definitions are different and more rigorous.
We distinguish between two types of trust, \emph{competence}, meaning that agent $A$ believes that $B$ is capable of producing the expected result, and \emph{disposition}, which means that agent $A$ believes that agent $B$ is willing to do what $A$ needs. 
Moreover, 
 we express trust with the usual \emph{probabilistic} quantification of certainty that we sometimes use in our daily life, e.g., ``I am 99\% certain that the autonomous taxi service is trustworthy'', or ``I trust the autonomous taxi service 99\%''.  Our formalisation then captures 
 how the value of 99\% can be computed based on the agent's past experience and (social, economic) preferences. Indeed, \cite{FC2001} also provides a justification for quantifying trust with probabilities, e.g. in Section 5, where it is stated that the degree of trust is a function of the subjective certainty of the pertinent beliefs.
We do not consider fulfilment belief also discussed in~\cite{FC2001}.

Beliefs and trust are intimately connected. As agents continuously interact with the environment and 
with each other, their understanding about the system and the other agents may increase,  
leading to an update of their beliefs.
This change in an agent's beliefs may also lead to an update of trust, which is reflected in the definition of the semantics of our trust operators. 

Changes of beliefs and trust will affect goals and intentions, which will be captured via so called pro-attitude synthesis, which is introduced in Section~\ref{sec:proupdate}. Based on its updated beliefs, an agent may modify its strategy to implement its intention.

\subsection{Probabilistic Rational Temporal Logic}

We now introduce Probabilistic Rational Temporal Logic (\PRTLS) that can express mental attitudes of agents in an \AutoSMAS, as well as beliefs and trust.
\PRTLS\ extends the logic \PCTLSM\ with operators for reasoning about agent's beliefs and cognitive trust.

\begin{definition}
The syntax of the logic \PRTLS\ is as follows:
$$
\begin{array}{l}
\phi  ::=  p ~|~\neg \phi~|~\phi\lor \phi~|~\forall\psi~|~\prob{\bowtie q}{\psi}~
|~\G_A\phi~|~\I_A\phi~|~\C_A\phi ~|~\B_A^{\bowtie q}\psi~|~
\CT_{A,B}^{\bowtie q}\psi~|~\DT_{A,B}^{\bowtie q}\psi~\\
\psi  ::=  \phi~|~\neg\psi~|~\psi \lor \psi~|~\next \psi ~|~ \psi\until \psi 
\end{array}
$$
where $p\in AP$, $A,B\in\Ags$, $\bowtie\in \{<,\leq,>,\geq\}$, and $q\in [0,1]$.
\end{definition}

Intuitively, $\B_A^{\bowtie q}\psi$ is the {\em belief} operator, expressing that agent $A$ believes $\psi$ with probability in relation $\bowtie$ with $q$. $\CT_{A,B}^{\bowtie q}\psi$ is the {\em competence trust} operator, expressing that agent $A$ trusts agent $B$ with probability in relation $\bowtie$ with $q$ on its \emph{capability} of completing the task $\psi$, where capability is understood to be the possibility of taking one of its \legal\ intentions defined with function $\lintn{B}$. $\DT_{A,B}^{\bowtie q}\psi$ is the {\em disposition trust} operator,  expressing that agent $A$ trusts agent $B$ with probability in relation $\bowtie$ with $q$ on its \emph{willingness} to do the task $\psi$, where the state of willingness is interpreted as unavoidably taking an intention defined with function $\intn{B}$. Recall that we use function $\lintn{B}$ for the \legal\ intentions and function $\intn{B}$ for the \possible\ intentions, i.e., intentions available assuming agent's willingness. We use $\T$ to range over the two trust operators $\CT$ and $\DT$.


Before we define the semantics, we require additional notation. We write:
\begin{alignat*}{2}
\probs{\cM,A,\rho}{max,min}(\psi) \enskip & \defequals && \enskip \textstyle\sup_{\astrat{A}\in\Strategy_A} \textstyle\inf_{\astrat{\Ags\setminus \{A\}}\in \Strategy_{\Ags\setminus \{A\}}} \probm{\cM,\astrat{},\rho}(\psi), \\
\probs{\cM,A,\rho}{min,max}(\psi) \enskip & \defequals && \enskip \textstyle\inf_{\astrat{A}\in\Strategy_A} \textstyle\sup_{\astrat{\Ags\setminus \{A\}}\in \Strategy_{\Ags\setminus \{A\}}} \probm{\cM,\astrat{},\rho}(\psi)
\end{alignat*}
to denote the \emph{strategic ability} of agent $A$ in implementing $\psi$ on a finite path $\rho$. Intuitively, $\probs{\cM,A,\rho}{max,min}(\psi)$ gives a lower bound on agent $A$'s ability to maximise the probability of $\psi$, while $\probs{\cM,A,\rho}{min,max}(\psi)$ gives an upper bound on agent $A$'s ability to minimise the probability of $\psi$. 

Moreover, we extend the expression $B.i(s,x)$ to work with finite paths:  $B.i(\rho,x)$ for the path $\rho'=\rho s$ and $last(\rho)\ctrans{B.i.x}s$. Intuitively, it is the concatenation of path $\rho$ with a cognitive transition $B.i.x$. 

For a measurable function $f: \fpath{}^\cM\rightarrow [0,1]$, we denote by $\Exp_{\be_A}[f]$ the \emph{belief-weighted expectation} of $f$,  i.e., 
\[
\Exp_{\be_A}[f]=\sum_{\rho\in \fpath{}^\cM}\be_A(\rho)\cdot f(\rho).
\]

With this in mind, we aim to define the new operators so that
$\B_A^{\bowtie q}\psi$ normalises the probability measure of the formula $\prob{\bowtie q}{\psi}$ with agent's probabilistic belief, whereas $\T_{A,B}^{\bowtie q}\psi$ operators consider moreover the \possible\ or \legal\ intention changes of another agent $B$. The intention changes are conditioned over agent's strategic ability of implementing formula $\psi$, expressed by $\probs{\cM,A,\rho}{max,min}(\psi)$ and $\probs{\cM,A,\phi}{min,max}(\psi)$.

\begin{definition}\label{def:semantics4}
Let $\asmasTuple$ be an \AutoSMAS\ and $\rho$ a finite path in $\cM$.
The semantics of previously introduced operators of \PCTLS\ and \PCTLSM\ remains the same in \PRTLS (see Definitions~\ref{def:semantics1} and~\ref{def:semantics3}). For the newly introduced belief and trust operators, the satisfaction relation $\models$ is defined as follows:
\begin{itemize}
\commentout{
\item $\cM,\rho\models \B_A^{\bowtie q}\phi$ if
\[
\Exp_{\be_A}[\mathrm{Sat}_{\cM,\phi}^{\bowtie}]\bowtie q,
\]
where 
the function $\mathrm{Sat}_{\cM,\phi}^{\bowtie}:\fpath{}^\cM\rightarrow [0,1]$ is such that
\[
\mathrm{Sat}_{\cM,\phi}^{\bowtie}(\rho') =
\left \{
\begin{array}{ll}
1 & \text{ if } \cM,\rho'\models \phi \\
0 & \text{ otherwise } \\
\end{array}
\right.
\]
}

\item $\cM,\rho\models \B_A^{\bowtie q}\psi$ if
\[
\Exp_{\be_A}[V_{\B,\cM,\psi}^{\bowtie}]\bowtie q,
\]
where the function $V_{\B,\cM,\psi}^{\bowtie}:\fpath{}^\cM\rightarrow [0,1]$ is  such that 
$$
V_{\B,\cM,\psi}^{\bowtie}(\rho') = 
\left \{
\begin{array}{ll}
\probs{\cM,A,\rho'}{max,min}(\psi) & \text{ if } \bowtie\in \{\geq, >\}\\
\probs{\cM,A,\rho'}{min,max}(\psi) & \text{ if }  \bowtie\in \{<, \leq\}\\
\end{array}
\right.
$$

\item $\cM,\rho\models \CT_{A,B}^{\bowtie q}\psi$ if
\[
E_{\be_A}[V_{\CT,\cM,B,\psi}^{\bowtie}]\bowtie q,
\]
where the function $V_{\CT,\cM,B,\psi}^{\bowtie}:\fpath{}^\cM\rightarrow [0,1]$ is  such that 
$$
V_{\CT,\cM,B,\psi}^{\bowtie}(\rho') = 
\left \{
\begin{array}{ll}
\displaystyle \sup_{x\in\lintn{B}(last(\rho'))} \probs{\cM,A,B.i(\rho',x)}{max,min}(\psi) & \text{ if } \bowtie\in \{\geq, >\}\\
\displaystyle\inf_{x\in\lintn{B}(last(\rho'))} \probs{\cM,A,B.i(\rho',x)}{min,max}(\psi) & \text{ if }  \bowtie\in \{<, \leq\}\\
\end{array}
\right.
$$

\item $\cM,\rho\models \DT_{A,B}^{\bowtie q}\psi$ if
\[
E_{\be_A}[V_{\DT,\cM,B,\psi}^{\bowtie}] \bowtie q,
\]
where the function $V_{\DT,\cM,B,\psi}^{\bowtie}:\fpath{}^\cM\rightarrow [0,1]$ is  such that 
$$
V_{\DT,\cM,B,\psi}^{\bowtie}(\rho') = 
\left \{
\begin{array}{ll}
\displaystyle\inf_{x\in\supp(\intn{B}(\rho'))} \probs{\cM,A,B.i(\rho',x)}{max,min}(\psi) & \text{ if }  \bowtie\in \{\geq, >\}\\
\displaystyle\sup_{x\in\supp(\intn{B}(\rho'))} \probs{\cM,A,B.i(\rho',x)}{min,max}(\psi) & \text{ if } \bowtie\in \{<, \leq\}\\
\end{array}
\right.
$$

\end{itemize}
\end{definition}

We interpret formulas $\phi$ in \AutoSMAS\ $\cM$ in a state reached after executing a path $\rho$, in history-dependent fashion. Recall that this path may have interleaved cognitive and temporal transitions. However, when evaluating a given belief or trust formula, we assume that agents do not change their mental attitudes (i.e., the future path is purely temporal, contained within a single induced \SMG), which motivates using the probability measure defined in Section~\ref{sec:smas} in the above definition. The belief formula corresponds to the probability of satisfying $\phi$ in future in the original \AutoSMAS\ $\cM$ weighted by the belief distribution; in other words, it is a belief-weighted expectation of future satisfaction of $\phi$, which is subjective, as it is influenced by $A$'s partial observation and its prior knowledge about $B$ encoded in the preference function. The competence trust operator reduces to the computation of optimal probability of satisfying $\psi$ in $\cM$ over \legal\ changes of agent's intention, which is again weighted by the belief distribution and compared to the probability bound $q$. Dispositional trust, on the other hand, computes the optimal probability of satisfying $\psi$ in $\cM$ over \possible\ changes of agent's intention, which is weighted by the belief distribution and compared to the probability bound $q$.

For a relational symbol $\bowtie\in \{\geq, >\}$, the expression $V_{\CT,\cM,B,\psi}^{\bowtie}(\rho)$ computes the maximum probability of completing the task $\psi$ on path $\rho$, among all \legal\ changes of agent $B$'s intention. Therefore, when interpreting formula $\CT_{A,B}^{\bowtie q}\psi$, we assume the optimal capability of agent $B$. Note that this capability is not necessarily within $B$'s \possible\ intention changes. On the other hand, for $\bowtie\in \{\geq, >\}$, the expression $V_{\DT,\cM,B,\psi}^{\bowtie}(\rho)$ computes the minimum probability of completing the task $\psi$ on path $\rho$ that all agent $B$'s \possible\ intentional changes can achieve. Therefore, when interpreting formula $\DT_{A,B}^{\bowtie q}\psi$, we consider all possible states of agent $B$'s willingness.
%
%
It should be noted that the $\T_{A,B}^{\bowtie q}$ operators cannot be derived from the other operators.

\begin{remark}
We remark that, while defining competence trust in terms of \legal\ intentions should be rather uncontroversial, the disposition trust operator may be expressed in an alternative way. Given that $\DT_{A,B}^{\bowtie q}$ computes $A$'s trust towards $B$ and that intentional attitudes of $B$ are in general not observable to $A$, one could argue that it is not realistic to assume that $A$ uses $B$'s intention strategy to compute its trust towards $B$. It would perhaps be more natural to use $A$'s intention preference function $\ipref{A}{B}$ instead, since it precisely represents $A$'s expectations about $B$'s mental attitudes. Indeed, it is possible to make such a modification, which we discuss below. 

First of all, note that, in the definition of $\DT_{A,B}^{\bowtie q}$ above, we restrict $A$'s knowledge to the support of $\intn{B}$, i.e., to $B$'s possible intentions, rather than allowing $A$ to know $B$'s full intention strategy. As a side note, the consequence of this is that $A$'s disposition trust towards $B$ differs from competence trust only when the set of \possible\ intentions of $B$ is a strict subset of $B$'s \legal\ intentions. In particular, it is the case when $B$'s intention strategy is pure, such as in Example~\ref{example:trustgamectrans}. Hence, an important practical consideration to bear in mind when constructing instances of \AutoSMAS\ is to keep cognitive strategies pure, or close to pure (with few intentions in the support). In fact, Section~\ref{sec:proupdate} provides an easy way to achieve this. 

Going back to the alternative way of defining $\DT_{A,B}^{\bowtie q}$, if we replace $B$'s intention strategy $\intn{B}$ by $A$'s intention preferences over $B$, $\ipref{A}{B}$, it makes sense to include in our calculations of trust the full probability distribution that $\ipref{A}{B}$ provides, rather than just its support, as we do for $\intn{B}$. That's because $\ipref{A}{B}$ represents knowledge that $A$ realistically possesses, so there's no reason to introduce any restrictions on it. We then end up with $\DT_{A,B}^{\bowtie q}$ defined as a \emph{belief-preference-weighted expectation} of agent's strategic ability of implementing a given formula $\psi$. Note that, with such formulation of $\DT_{A,B}^{\bowtie q}$, the last equivalence of Theorem~\ref{thm:surebeliefs} does not hold.

Finally, note that disposition trust operator of one agent may be nested inside a belief operator of another agent, which calls for nested preferences (i.e., preferences over preferences) in an idealised version of $\DT_{A,B}^{\bowtie q}$. We aim to avoid such complexities at this early stages of the development of the framework, and therefore settle for a simplified notion of disposition trust involving intention strategy. 
\end{remark}

\begin{example}\label{example:trustgametrustformulas}
Here we give examples of formulas that we may wish to verify on the trust game defined in Example~\ref{example:trustgamectrans}.
The formula
\[
\B_{Bob}^{\geq 0.6} \next (a_{Alice} = invest)
\]
states that Bob believes that, with probability at least 0.6, Alice will invest the money with him in the next step. On the other hand, the formula
\[
\DT_{Alice,Bob}^{\geq 0.9} \eventually (a_{Bob} = keep)
\]
states that Alice can trust Bob with probability no less than $0.9$ that he will keep the money for himself. The formula
\[
\forall\always(richer_{Bob,Alice} \to \prob{\geq 0.9}\eventually \CT_{Bob,Alice}^{\geq 1.0} richer_{Alice,Bob})
\]
states that, at any point of the game, if Bob is richer than Alice,
then with probability at least $0.9$, in future, he can almost surely, i.e., with probability 1, trust Alice on her ability to become richer.

It is also possible, and often desired, to have a trust operator within the scope of a belief operator. For instance, 
\[
\B_{Bob}^{\geq 0.7}\DT_{Alice,Bob}^{\geq 0.5} \eventually (a_{Bob} = share)
\]
expresses that Bob believes that, with probability at least 0.7, Alice's trust that he will share the profits with her in the future is at least 0.5. Such construct might represent Bob's reasoning leading to his decision of whether to share his profit or not, where he considers Alice's trust towards him as a determining factor.
A similar statement might be expressed by the following formula:
\[
\B_{Bob}^{\geq 0.7}\B_{Alice}^{\geq 0.5} \I_{Bob} \exists\eventually (a_{Bob} = share).
\]
The above states that Bob believes that, with probability at least 0.7, Alice believes that, with probability at least 0.5, Bob has a possible intention which ensures that it is possible that he shares the profit. The first formula may be seen as a quantitative version of the second, as it additionally (implicitly, via $\DT$ operator) considers the probability of Bob sharing the profit. However, both formulas are equivalent for the trust game, assuming the action strategies defined in Example~\ref{example:trustgamecstate} and intention strategy from Example~\ref{example:trustgamectrans}. To see that, first note that Bob's intention strategy is pure, hence he has only one possible intention. Second, since the action strategy implementing that intention (whatever it is) is pure, the probability of him sharing the profit is either 1 or 0. Therefore, $\inf_{x\in\supp(\intn{Bob}(\rho))} \probm{\trustgame,Bob.i(\rho,x)}(\eventually (a_{Bob} = share))$ takes only two values, 1 and 0, and is equal to 1 whenever $\probm{\trustgame,Alice,\rho}(\I_{Bob} \exists\eventually (a_{Bob} = share)) = 1$ for any path $\rho = \rho' s_k$ for $k \in \{8,10,12,14\}$. Hence,
\[
\DT_{Alice,Bob}^{\geq 0.5} \eventually (a_{Bob} = share) \equiv \B_{Alice}^{\geq 0.5} \I_{Bob} \exists\eventually (a_{Bob} = share).
\]

\commentout{
Because the strategies $\astrat{share}$ and $\astrat{keep}$, defined in Example~\ref{example:trustgame}, are pure, both formulas are in fact equivalent for the trust game. To see that, note that the existence of a possible intention guaranteeing the possibility of Bob sharing the profit (in other words, the satisfaction of formula $\C_{Bob} \exists\eventually (a_{Bob} = share)$) is equivalent to the probability of Bob sharing the profit being 1. Moreover, Bob's intention strategy is pure, which implies that he has only one possible intention at any time. Therefore, $\probm{\trustgame,Alice,\rho}(\C_{Bob} \exists\eventually (a_{Bob} = share)) = 1$ whenever $\inf_{x\in\supp(\intn{Bob}(\rho))} \probm{\trustgame,Bob.i(\rho,x)}(\eventually (a_{Bob} = share)) = 1$ for any path $\rho = \rho' s_k$ for $k \in \{8,10,12,14\}$. Hence,
\[
\DT_{Alice,Bob}^{\geq 0.5} \eventually (a_{Bob} = share) \equiv \B_{Alice}^{\geq 0.5} \I_{Bob} \exists\eventually (a_{Bob} = share).
\]
}
In general, however, those two formulas differ in their semantics. 
\hfill $\Box$
\end{example}

\commentout{
\todo{
We remark that a belief formula $\B_A^{\bowtie q}\psi$ or a trust formula $\T_{A,B}^{\bowtie q}\psi$ cannot be in the scope of an operator $\C_A$ for $A\in\Ags$. By the definition of preference functions, an agent's beliefs deal with those paths that are within its mental consideration.
However, the operator $\C_A$ tests the existence of a \legal\ intention which may be outside of the agent's intentional attitude (recall that $supp(\intn{A}(s_0...s_k))\subseteq \lintn{A}(s_k)$). On the other hand, the operator $\C_A$ can be in the scope of a belief formula.
This constraint is needed for our definition of belief updates.
}
}

\begin{definition}
For a given autonomous stochastic multi-agent system $\cM$ and a formula $\phi$ of the language \PRTLS, the \emph{model checking problem}, written as $\cM\models \phi$, is to decide whether $\cM,s\models\phi$ for all initial states $s\subseteq supp(\initdist)$.
\end{definition}

\subsection{Restricting to Deterministic Behaviour Assumption} 

The above semantics is complex because of the quantification over agents' strategies. However, under Assumption~\ref{assump:detbeh}, this quantification is not necessary, since the system conforms to a single action strategy, induced from agents' cognitive states, which we call the induced strategy $\indStrat$. As a result, some operators have simpler semantics. We begin with an easy consequence of restricting ourselves to a single action strategy.

\begin{proposition}\label{propn:simplified}
With Assumption~\ref{assump:detbeh} on \AutoSMAS\ $\cM$, we have the following equivalence for 
$\rho$ being a finite path, and $\psi$ being a \PRTLS\ formula.
\[
\probs{\cM,A,\rho}{max,min}(\psi) = \probs{\cM,A,\rho}{min,max}(\psi) = \probs{\cM,\rho}{max,min}(\psi) = \probs{\cM,\rho}{min,max}(\psi) = \probm{\cM,\indStrat,\rho}(\psi).
\]
Moreover, for systems $\cM$ satisfying Assumption~\ref{assump:detbeh}, we will often omit the $\indStrat$ component in the last expression above and write $\probm{\cM,\rho}(\psi)$ instead. Intuitively, it represents the probability of $\psi$ being satisfied on a future path in a system $\cM$ given that the current path is $\rho$. 
\end{proposition}


\begin{definition}\label{def:semantics5}
Let $\cM$ be an \AutoSMAS\ satisfying Assumption~\ref{assump:detbeh} and $\rho$ a finite path in $\cM$.
The following are several operators whose semantics can be simplified according to Proposition~\ref{propn:simplified}.


%
\begin{itemize}
\item $\cM,\rho\models \prob{\bowtie q}{\psi}$ if
\[
\probm{\cM,\rho}(\psi)\bowtie q.
\]

\item $\cM,\rho\models \B_A^{\bowtie q}\psi$ if
\[
\Exp_{\be_A}[V_{\B,\cM,\psi}]\bowtie q,
\]
where the function $V_{\B,\cM,\psi}:\fpath{}^\cM\rightarrow [0,1]$ is  such that 
\[
V_{\B,\cM,\psi}(\rho') = \probm{\cM,\rho'}(\psi) 
\]

\item $\cM,\rho\models \CT_{A,B}^{\bowtie q}\psi$ if 
\[
E_{\be_A}[V_{\CT,\cM,B,\psi}^{\bowtie}]\bowtie q,
\]
where the function $V_{\CT,\cM,B,\psi}^{\bowtie}:\fpath{}^\cM\rightarrow [0,1]$ is such that
$$
V_{\CT,\cM,B,\psi}^{\bowtie}(\rho') = 
\left \{
\begin{array}{ll}
\displaystyle \sup_{x\in\lintn{B}(last(\rho'))} \probm{\cM,B.i(\rho',x)}(\psi) & \text{ if } \bowtie\in \{\geq, >\}\\
\displaystyle\inf_{x\in\lintn{B}(last(\rho'))} \probm{\cM,B.i(\rho',x)}(\psi) & \text{ if }  \bowtie\in \{<, \leq\}.\\
\end{array}
\right.
$$

\item $\cM,\rho\models \DT_{A,B}^{\bowtie q}\psi$ if 
$$E_{\be_A}[V_{\DT,\cM,B,\psi}^{\bowtie}]\bowtie q,$$ 
where the function $V_{\DT,\cM,B,\psi}^{\bowtie}:\fpath{}^\cM\rightarrow [0,1]$ is such that
$$
V_{\DT,\cM,B,\psi}^{\bowtie}(\rho') = 
\left \{
\begin{array}{ll}
\displaystyle \inf_{x\in\supp(\intn{B}(\rho'))} \probm{\cM,B.i(\rho',x)}(\psi) & \text{ if } \bowtie\in \{\geq, >\}\\
\displaystyle\sup_{x\in\supp(\intn{B}(\rho'))} \probm{\cM,B.i(\rho',x)}(\psi) & \text{ if }  \bowtie\in \{<, \leq\}.\\
\end{array}
\right.
$$

\end{itemize}
\end{definition}

Intuitively, given that the action strategy in the system is fixed as $\indStrat$, there is no need to consider agents' strategic abilities, and hence the probability of a given formula $\psi$ being satisfied can be uniquely determined, which simplifies the definitions above.

\subsection{A Single-Agent Special Case}

The semantics can be further simplified if we work with the single-agent case, in which case we can employ the belief \AutoSMAS\ construct $\Bel{A}(\cM)$ as given in Section~\ref{sec:beliefasmas}. Recall that this approach explores belief states (distributions over states of $\cM$) and that $\dmap{A}$ is a mapping from paths in $\cM$ to paths in belief \AutoSMAS\ $\Bel{A}(\cM)$. Therefore, to evaluate belief and trust we can directly employ current belief $\last(\dmap{A}(\rho))$, instead of working with belief distribution $\be_A$ defined on finite paths $\rho$. 

Here we only present the semantics of the belief and trust operators. 

\begin{definition}\label{def:semantics6}
Let $\cM$ be an \AutoSMAS\ satisfying Assumption~\ref{assump:detbeh} and $\rho$ a finite path in $\cM$.
The semantics of the belief and trust operators can be defined on $\Bel{A}(\cM)$
as follows.

\begin{itemize}
\item $\cM,\rho\models \B_A^{\bowtie q}\psi$ if
\[
\Exp_{\last(\dmap{A}(\rho))}[ \mathrm{Sat}_{\psi} ] \bowtie q,
\] 
where the function $\mathrm{Sat}_{\psi}:S\rightarrow [0,1]$ is such that $\mathrm{Sat}_{\psi}(s) = 1 \; \mathrm{ if } \;  \cM,s \models \psi \; \mathrm{ and } \; 0 \; \mathrm{otherwise}.
$


\item $\cM,\rho\models \CT_{A,B}^{\bowtie q}\psi$ for $\bowtie \in \{\geq,>\}$ if
$$\Exp_{\last(\dmap{A}(\rho))}[ V_{\CT,\cM,B,\psi} ] \bowtie q
$$
where the function $V_{\CT,\cM,B,\psi}:S\rightarrow [0,1]$ is such that 
$$V_{\CT,\cM,B,\psi}(s)=\sup_{x\in\lintn{B}(s)}  \probm{\cM,B.i(s,x)}(\psi)
$$
and if $\bowtie$ is $\leq$ or $<$ we replace $\sup$ with $\inf$ in the above.


\item $\cM,\rho\models \DT_{A,B}^{\bowtie q}\psi$ if for $\bowtie \in \{\geq,>\}$ if
$$\Exp_{\last(\dmap{A}(\rho))}[ V_{\DT,\cM,B,\psi}^{\bowtie} ] \bowtie q
$$
where the function $V_{\DT,\cM,B,\psi}^{\bowtie}:S\rightarrow [0,1]$ is such that 
$$V_{\DT,\cM,B,\psi}^{\bowtie}(s)=\inf_{x\in\intn{B}(\rho')}  \probm{\cM,B.i(s,x)}(\psi)
$$
for $\rho'$ such that $\dmap{A}(\rho') = \dmap{A}(\rho)$ and $\last(\rho') = s$; 
if $\bowtie$ is $\leq$ or $<$ we replace $\inf$ with $\sup$ in the above. 
\end{itemize}

\end{definition}

Note that $A$'s belief in $\psi$ is evaluated as an expectation of satisfaction of $\psi$ computed in the current belief state. $A$'s competence trust in $B$'s ability to perform $\psi$ reduces to the expectation of the optimal probability of completing $\psi$, over all legal changes of $B$'s intention. Similarly, to evaluate $A$'s disposition trust in $B$'s willingness to perform $\psi$ we compute the expectation of the optimal probability of completing $\psi$, over all possible changes of $B$'s intention $ \intn{B}$. 

\begin{remark}
The sets of possible intentions $ \intn{B}(\cdot)$ of agent $B$ can be determined for a given belief state 
with pro-attitude synthesis defined in Section~\ref{sec:proupdate}. 
\end{remark}

\subsection{Other Trust Notions}

In~\cite{FC2001}, another concept related to trust called \emph{dependence} is introduced, whose interpretation is that agent $A$ needs agent $B$ (strong dependence), or is better off relying on $B$ (weak dependence). We now propose how those two concepts can be formalised in an \AutoSMAS\ under Assumption~\ref{assump:detbeh}.

\begin{definition}\label{def:dependence}
Let $\cM$ be an \AutoSMAS\ satisfying Assumption~\ref{assump:detbeh} and $\rho$ a finite path in $\cM$.
We introduce operators $\ST_{A,B}^{\bowtie q}$ and $\WT_{A,B}^{\bowtie}$ to express strong and weak dependence, with the semantics given by:
\begin{itemize}
\item $\cM,\rho\models \ST_{A,B}^{\bowtie q}\psi$ if 
$\cM,\rho\models (\neg\B_A^{\bowtie q}\psi \land \CT_{A,B}^{\bowtie q}\psi) \lor (\B_A^{\bowtie q}\psi \land \neg\CT_{A,B}^{\bowtie q}\psi)$ 

\item $\cM,\rho\models \WT_{A,B}^{\bowtie}\psi$ if 
\[
E_{\be_A}[V_{\cM,A,B,\psi}] \bowtie E_{\be_A}[V_{\CT,\cM,A,\psi}^{\bowtie}],
\]
where the function $V_{\cM,A,B,\psi}:\fpath{}^\cM\rightarrow [0,1]$ is such that 
\[
V_{\cM,A,B,\psi}(\rho') = \sum_{x\in \lintn{B}(last(\rho'))}\ipref{A}{B}(last(\rho'))(x) \cdot \probm{\cM,B.i(\rho',x)}(\psi).
\]

\end{itemize}
\end{definition}
Intuitively, $\ST_{A,B}^{\bowtie q}\psi$ states that, according to agent $A$'s beliefs, either (i) the probability of $\psi$ will not be as required unless agent $B$ takes on an appropriate intention, or (ii) the probability of $\psi$ will be as required unless $B$ takes on a certain intention. 
In other words, strong dependence of $A$ on $B$ over $\psi$ means that $B$ can control the probability of future satisfaction of $\psi$.
On the other hand, $\WT_{A,B}^{\bowtie}\psi$ states that, according to $A$'s beliefs and preferences, intentional changes of agent $B$ can bring about better results than any of the available intentional changes of agent $A$. 



We note that alternative definitions of the above notions are possible. For example, we could introduce a non-probabilistic version of strong dependence, based solely on the $\C_A$ operator.
We could define 
\[
\cM,\rho\models \ST_{A,B}\phi \text{ if } \cM,\rho\models \C_B \overline{\C_A} \neg\phi \land \C_B \C_A \phi.
\]

Intuitively, here $\ST_{A,B}\phi$ means that $B$ can take on an intention such that, regardless of what intention $A$ takes on, $\phi$ will not be satisfied, and $B$ can take on another intention, such that $A$ can take on an intention which ensures that $\phi$ holds. In other words, $B$ can either make it impossible for $A$ to achieve $\phi$, or it can allow $A$ to make $\phi$ true. 
We can view the above definition as a more objective, absolute interpretation of strong dependence, whereas the probabilistic version represents a subjective notion, influenced by agent's understanding of the system execution. 
Depending on the use case, either variant of the strong dependence operator might be employed.

Weak dependence can also be interpreted in a different way, especially in systems equipped with a reward structure. The intuitive meaning of the expression \emph{better off} used in the informal definition of weak dependence relates to some implicit notion of agent's well being. Definition~\ref{def:dependence} takes the probability of satisfying a given formula as agent's reward. However, in presence of a reward structure, we could allow $\psi$ to include a reward operator and adapt the definitions of the $V$-functions to compute expected reward.

\begin{example}
In the trust game, Bob's financial situation is dependent on Alice's actions. We can formally express that in \PRTLS\ with the following formula, where $\mathit{profit}_{Bob}$ is an atomic proposition, true in states in which Bob has more money than in the initial state:
\[
\ST_{Bob,Alice} \, \forall \eventually \mathit{profit}_{Bob}.
\]
It expresses that Bob needs Alice's cooperation to ensure that, at some point in the future, he will have more money than when he started. In fact, Alice's hopes of making a profit depend on Bob's cooperation, which we could express by a formula very similar to the one above.

We remark also that we could equip the trust game with a reward structure that reflects agents' payoffs. We could then express that Bob is weakly dependent on Alice to maximise his payoff with the following formula:
\[
\WT_{Bob,Alice} \, R^{=?} [\eventually end],
\]
where $R^{=?} [\eventually end]$ returns the expected reward accumulated before the game ends. 
\end{example}


\subsection{Trust for Systems with Sure Beliefs}\label{sec:surebelief}

We now consider a special class of  systems where agent $A$ has sure beliefs, i.e., for all finite paths $\rho$, we have either $\be_A(\rho)=1$ or $\be_A(\rho)=0$. Intuitively, agent $A$ is sure about the current system state and the execution history.
The following theorem shows that, in such systems, agent $A$'s beliefs and trusts can be expressed with other operators.
\begin{theorem}\label{thm:surebeliefs}
For \AutoSMAS\ $\cM$ in which agent $A$ has sure beliefs, the following equivalences hold for any finite path $\rho$, formula $\psi$, and agent $B\neq A$:
\begin{itemize}
\item $\cM,\rho\models \B_A^{\bowtie q}\psi$ if and only if $\cM,\rho\models \prob{\bowtie q}{\psi}$,
\item $\cM,\rho\models \CT_{A,B}^{\bowtie q}\psi$ if and only if $\cM,\rho\models \C_B\prob{\bowtie q}{\psi}$,
\item $\cM,\rho\models \DT_{A,B}^{\bowtie q}\psi$ if and only if $\cM,\rho\models \I_B\prob{\bowtie q}{\psi}$.
\end{itemize}
\end{theorem}
Intuitively, the first equivalence shows that agent $A$'s belief of $\psi$ reduces to computing the probability of satisfying $\psi$. 
The other two equivalences concern agent $A$ acting as a trustor and agent $B$ as a trustee. The former indicates that agent $A$'s trust in $B$'s competence in achieving $\psi$ is equivalent to the existence of $B$'s \legal\ intention which ensures satisfaction of $\psi$, while the latter means that $A$'s trust in $B$'s disposition towards achieving $\psi$ is equivalent to $B$ ensuring satisfaction of $\psi$ under all \possible\ intentions.

%


The following proposition shows that the sure-belief system has an equivalent definition as fully observable system. Therefore, Theorem~\ref{thm:surebeliefs} also holds for fully observable systems.
\begin{proposition}
Let $\cM$ be an \AutoSMAS\ and $A$ an agent in $\cM$. The following two statements are equivalent:
\begin{itemize}
\item Agent $A$ has sure beliefs, i.e., for all finite paths $\rho$, we have either $\be_A(\rho)=1$ or $\be_A(\rho)=0$.
\item Agent $A$ has full observation, i.e., $\obs{A}(s)=s$ for all $s\in S$.
\end{itemize}
\end{proposition}

\section{Pro-attitude Synthesis}\label{sec:proupdate}

Recall that the cognitive strategies $\goal{A}$ and $\intn{A}$ are defined on the (possibly infinite) set of paths and may lack a finite representation necessary for model checking.
We now formalise an idea, which is informally argued in~\cite{vdHJW2007,CP2007}, that the changes of pro-attitudes are determined by the changes of beliefs, but not vice versa.

This will enable us to \emph{synthesize} the cognitive strategies $\goal{A}$ and $\intn{A}$ by considering agent $A$'s beliefs and trust expressed as formulas of the logic \PRTLS\ and used to constrain pro-attitude changes.
In particular, we will associate with each set of goals a condition, expressed as a \PRTLS\ formula, which will \textit{guard} that set of goals. The intuitive interpretation of such an association is that an agent takes a set of goals when the associated condition is satisfied.
To guard intentions, we additionally use the agent's own goals, since agents' intentions are influenced by their goals. We note that such a construction closely resembles human reasoning -- for example, recall that in Example~\ref{example:trustgamecstate} we stated that Bob takes intention \emph{share} when he is an \emph{investor} (i.e., his goal is \emph{investor}) and his belief in Alice being \emph{active} is sufficient. In Example~\ref{example:trustgameguardfn}, we show how to express this statement formally using pro-attitude synthesis, which consists of evaluating Bob's belief in Alice being \emph{active} as a guarding mechanism for his intention to \emph{share} and setting a minimum belief threshold.

Let $\Lang_A(\PRTLSf)$ be the set of formulas of the logic \PRTLS\ that are Boolean combinations of atomic propositions and formulas of the form $\B_A^{\bowtie q}\psi$, $\T_{A,B}^{\bowtie q}\psi$, $\B_A^{\bowtie ?}\psi$ or $\T_{A,B}^{\bowtie ?}\psi$, such that $\psi$ does not contain temporal operators. The formulas $\B_A^{\bowtie ?}\psi$ and $\T_{A,B}^{\bowtie ?}\psi$ denote the quantitative variants of $\B_A^{\bowtie q}\psi$ and $\T_{A,B}^{\bowtie q}\psi$ that yield the actual value of the probability/expectation. 

The language $\Lang_A(\PRTLSf)$ allows nested beliefs\footnote{Because the semantics of trust is based on the belief function $\be_A$, a trust formula is also called a belief formula in this paper.}, for example $\B_A^{\geq 0.9}\B_B^{>0.7}\psi$, under the condition that the outermost belief operator is for agent $A$. The nesting is useful since one may want to reason about, e.g., an agent's belief over the other agent's trust in himself.
Moreover, all modal logic formulas in $\Lang_A(\PRTLSf)$ must be in the scope of a belief operator of agent $A$. This is to ensure that agent $A$ is able to invoke the synthesis with the limited information it has.

Below, we work with an arbitrary \AutoSMAS\ $\cM$ and a set of agents $\Ags$ in $\cM$.

\begin{definition}\label{def:guardfns}
For every agent $A\in\Ags$, we define:
\begin{itemize}
\item a \emph{goal guard} function $\goalGd{A}:  \powerset{\Goal{A}} \rightarrow \Lang_A(\PRTLSf)$ and
\item an \emph{intention guard} function $\intnGd{A}: \Intn{A} \times \powerset{\Goal{A}} \rightarrow \Lang_A(\PRTLSf)$.
\end{itemize}
\end{definition}

The guard functions are partial functions associating pro-attitudes with conditions expressed using belief and trust formulas that can be evaluated based on a finite execution history. Note how $\intnGd{A}$ differs from $\goalGd{A}$ in that it guards intentions by agent's own goals. We call $\goalGd{A}$ \textit{trivial} if $\goalGd{A}(x) = \mathsf{true}$ for all $x \in \powerset{\Goal{A}}$. Similarly for $\intnGd{A}$.

We recall that $\Cognition = \{ \langle \lgoal{A},\lintn{A} \rangle \}_{A\in\Ags}$ is the cognitive mechanism of $\cM$. Additionally, we set $\Guard = \{ \langle \goalGd{A},\intnGd{A} \rangle \}_{A\in\Ags}$ and call it the \textit{guarding mechanism}.
In the following, we present an approach to obtain the infinite structure $\CStrategy=\{\goal{A},\intn{A}\}_{A\in\Ags}$, i.e., cognitive strategies of all agents, from finite structures $\Cognition$ and $\Guard$.

First, for agent $A \in \Ags$, we define the \textit{goal evaluation function} $\gsynth{A}$, such that for $x \subseteq \powerset{\Goal{A}}$,  $\gsynth{A}(x) : \fpath{}^\cM \rightarrow [0,1]$ is given by:
\[
\gsynth{A}(x)(\rho) = \left \{
\begin{array}{ll}
\cM,\rho\models \goalGd{A}(x) & \text{ if } x\in \lgoal{A}(last(\rho)) \\
0 & \text{ otherwise } \\
\end{array}
\right.
\]

\commentout{
First, for agent $A \in \Ags$, we define the \textit{goal evaluation function} $\gsynth{A} : \powerset{\Goal{A}} \times \fpath{}^\cM \rightarrow [0,1]$ as follows.
\[
\gsynth{A}(x,\rho) = \left \{
\begin{array}{ll}
\cM,\rho\models \goalGd{A}(x) & \text{ if } x\in \Goal{A}(last(\rho)), \\
0 & \text{ otherwise. } \\
\end{array}
\right.
\]
}

Intuitively, $\gsynth{A}(x)$ evaluates the (probabilistic) satisfiability of agent $A$'s beliefs or trust expressed as $\Lang_A(\PRTLSf)$ formulas given by $\goalGd{A}(x)$. Note that the expressions $\B_A^{\bowtie ?}\psi$ and $\T_{A,B}^{\bowtie ?}\psi$ return their corresponding probabilistic values, and $\cM,\rho\models \B_A^{\bowtie q}\psi$ and $\cM,\rho\models \T_{A,B}^{\bowtie q}\psi$ return the value $0$ or $1$ depending on the verification result.

Similarly, we define the \textit{intention evaluation function} $\isynth{A}$, such that for any $x \in \Intn{A}$, $y \subseteq \powerset{\Goal{A}}$, $\isynth{A}(x,y) : \fpath{}^\cM \rightarrow [0,1]$ is given by:
\[
\isynth{A}(x,y)(\rho) = \left \{
\begin{array}{ll}
\cM,\rho\models \intnGd{A}(x,y) & \text{ if } x\in \lintn{A}(last(\rho)), \\
0 & \text{ otherwise. } \\
\end{array}
\right.
\]

\commentout{
Similarly, we define the \textit{intention evaluation function} $\isynth{A} : \Intn{A} \times \powerset{\Goal{A}} \times \fpath{}^\cM \rightarrow [0,1]$ as follows.
\[
\isynth{A}(x,y,\rho) = \left \{
\begin{array}{ll}
\cM,\rho\models \intnGd{A}(x,y) & \text{ if } x\in \Intn{A}(last(\rho)) \\
0 & \text{ otherwise } \\
\end{array}
\right.
\]
}

We now use pro-attitude evaluation functions to provide a finite definition of a cognitive strategy. 

\begin{definition}\label{def:cognitivefunction}
Let $\cM$ be an \AutoSMAS\ with a cognitive mechanism $\Cognition = \{ \lgoal{A}, \lintn{A} \}_{A \in \Ags}$ and let $\Guard = \{ \langle \goalGd{A},\intnGd{A} \rangle \}_{A\in\Ags}$ be a guarding mechanism for $\cM$.
In a finite-memory setting, the \emph{goal strategy} $\goal{A}:\fpath{}^\cM\rightarrow \dist{\powerset{\Goal{A}}}$ can be instantiated as follows: for any $x\subseteq \Goal{A}$ and  $\rho\in \fpath{}^\cM$, we set
\[
\goal{A}(\rho)(x) =\frac{\gsynth{A}(x)(\rho)}{\sum_{x\in \lgoal{A}(\last(\rho))}\gsynth{A}(x)(\rho)}.
\]

Likewise, the \emph{intention strategy} $\intn{A}:\fpath{}^\cM\rightarrow \dist{\Intn{A}}$ can be instantiated as follows: for any $x\in \Intn{A}$ and $\rho\in \fpath{}^\cM$, we set
\[
\intn{A}(\rho)(x) = \frac{\isynth{A}(x,y)(\rho)}{\sum_{x\in \lintn{A}(\last(\rho))}\isynth{A}(x,y)(\rho)},
\]
where $y = \gs_A(last(\rho))$.
\end{definition}

This definition allows us to transform the model $\cM$ into $\synth{}(\cM)$, which, besides the cognitive mechanism $\Cognition$, contains the cognitive strategy of each agent. Once $\synth{}(\cM)$ has been precomputed, model checking a given specification can be carried out on $\synth{}(\cM)$ instead of $\cM$. We refer to this approach as \textit{pro-attitude synthesis before model checking}. Besides formalising the intuition that agent's pro-attitudes are affected by its beliefs, this approach makes the interaction of agents' beliefs and trust possible, without resorting to nested beliefs.

\begin{example}\label{example:trustgameguardfn}
We show how a cognitive strategy can be constructed from the guarding mechanism for Bob in the trust game $\trustgame$. Since we assume that goals of agents are static throughout the execution of the game, we only concern ourselves with the intention guard function. We recall our informal assumption that Bob's intention will be \emph{share} when he is an investor and his belief in Alice being active is sufficient, and \emph{keep} otherwise. We formalise it as follows:
\begin{alignat*}{2}
\intnGd{Bob}(share, \{ investor \}) &= \B_{Bob}^{>0.7} active_{Alice}, \\
\intnGd{Bob}(keep, \{ investor \}) &= \neg \B_{Bob}^{>0.7} active_{Alice}, \\
\intnGd{Bob}(share, \{ opportunist \}) &= \bot, \\
\intnGd{Bob}(keep, \{ opportunist \}) &= \top,
\end{alignat*}
where $active_{Alice}$ holds in states in which Alice's goal is \emph{active} and we used a value $0.7$ to represent Bob's belief threshold.

We now synthesize Bob's intention strategy. We let $\rho_1 = s_0s_1s_3s_8$ and $\rho_2 = s_0s_2s_5s_{12}$.
We recall from Example~\ref{example:equivClasses} that $\obs{Bob}(\rho_1) = \obs{Bob}(\rho_2)$
and we let $o_1$ denote that common observation. By Example~\ref{example:trustgamebelief}:
\begin{align*}
\be_{Bob}(o_1,\rho_1) &= 1/7, \\
\be_{Bob}(o_1,\rho_2) &= 6/7.
\end{align*}
Therefore, since $\trustgame, \rho_1 \models \neg active_{Alice}$ and $\trustgame, \rho_2 \models active_{Alice}$ (below and in what follows, $k \in \{1,2\}$):
\[
\trustgame, \rho_j \models \B_{Bob}^{=6/7} active_{Alice}.
\]
Hence
\begin{align*}
\isynth{Bob}(share, \{investor\})(\rho_k) &= 1, \\
\isynth{Bob}(keep, \{investor\})(\rho_k) &= 0,
\end{align*}
and so:
\begin{align*}
\intn{Bob}(\rho_k)(share) &= 1, \\
\intn{Bob}(\rho_k)(keep) &= 0.
\end{align*}
Likewise, letting $\rho_3 = s_0s_1s_4s_{10}$, $\rho_4 = s_0s_2s_6s_{14}$ and $k \in \{3,4\}$ one can show that:
\begin{align*}
\intn{Bob}(\rho_j)(share) &= 0, \\
\intn{Bob}(\rho_j)(keep) &= 1.
\end{align*}
As expected, if Bob is an \emph{opportunist}, he keeps the money for himself. However, when he is an \emph{investor} and Alice invests the money with him, his belief that Alice is \emph{active} is high enough for him to share the profits, hoping to gain Alice's trust.
\hfill $\Box$
\end{example}

Finally, we remark that, while it may not be immediately clear how powerful pro-attitude synthesis is after considering a simple example such as this, full appreciation of it can be gained by applying it to more complex systems. One such system is provided by an iterated trust game, in which Alice and Bob interact with each other repeatedly, in which beliefs and trust evolve dynamically. Rather than specifying cognitive strategies by hand for each iteration, pro-attitude synthesis enables one to encode them in a natural way via guard functions and generate them dynamically. Such an approach scales well to larger systems and closely resembles human reasoning.

\paragraph{\bf Preference Functions Update}\label{sec:prefupdate}
During the execution of the system, agents learn new information about other agents by observing their actions. As a result, their understanding of others' behaviour, motivations and cognitive state increases. This new knowledge gained by agents is reflected in their belief function, which could be used to update agents' preference functions. To do that, we extend guarding mechanism of each agent $A$ by introducing guard functions of $A$ over $B$ for every other agent $B$.
In particular, for agents $A$ and $B$, we define a goal (resp. intention) guard function $\goalGd{A,B}$ (resp. $\intnGd{A,B}$) of $A$ over $B$ in a similar way as in Definition~\ref{def:guardfns}, and then the goal (resp. intention) evaluation function $\gsynth{A,B}$ (resp. $\isynth{A,B}$) of $A$ over $B$. Intuitively, $\goalGd{A,B}$ (resp. $\intnGd{A,B}$) captures $A$'s expectation of how $B$ updates their mental attitudes. In case $A$ does not possess such information, $\goalGd{A,B}$ (resp. $\intnGd{A,B}$) are trivial. We then define $\gpref{A}{B}^\history$ and $\ipref{A}{B}^\history$ (where $\history$ stands for \emph{history}, representing that updated preference functions are history-dependent) by weighting the state-defined preferences with respect to the evaluation functions as follows:

\begin{itemize}
\item $\displaystyle\gpref{A}{B}^\history(\rho)(x)=\frac{\gpref{A}{B}(last(\rho))(x)\times \gsynth{A,B}(x)(\rho)}{\sum_{x\subseteq \Goal{A}}\gpref{A}{B}(last(\rho))(x)\times \gsynth{A,B}(x)(\rho)}$,
\item   $\displaystyle\ipref{A}{B}^\history(\rho)(x)=\frac{\ipref{A}{B}(\rho)(x)\times \isynth{A,B}(x,y)(\rho)}{\sum_{x\in \Intn{A}}\ipref{A}{B}(\rho)(x)\times  \isynth{A,B}(x,y)(\rho)}$,
\end{itemize}
where $y = gs_B(last(\rho))$.

Intuitively, the preference functions of agent $A$ over $B$ will be adjusted with respect to its current understanding of the system execution and the knowledge it has acquired throughout, represented by a $\Lang_A(\PRTLSf)$ formula in the guarding functions $\goalGd{A,B}$ and $\intnGd{A,B}$.
The probability space construction from Section~\ref{sec:preferencefuncs} can be adapted to work with the updated preference functions $\pref{A}^\history$. Therefore, we have probability spaces for both pro-attitude synthesis and model checking.



\begin{example}
After investing her money with Bob, Alice might be concerned whether she can trust Bob to share his profit with her. Such trust can be formally expressed by a formula $\phi = \DT_{Alice,Bob}^{\geq ?}\psi$, where $\psi = \next (a_{Bob} = share)$ and $\geq$ indicates that Alice wants to find a lower bound on such trust value. We assume that Alice is \emph{active} and let $\rho_1 = s_0s_2s_5s_{12}$, $\rho_2 = s_0s_2s_6s_{14}$, $o_1 = \obs{Alice}(\rho_1) = \obs{Alice}(\rho_2)$ and $k \in \{1,2\}$. We use Alice's preference function defined in Example~\ref{example:trustgamepref}, which we recall to be: 
\[
\gpref{Alice}{Bob}(s_k) = \langle investor \mapsto 1/2, opportunist \mapsto 1/2 \rangle.
\]
From that, we easily compute Alice's belief upon observing $o_1$ to be:
\begin{align*}
\be_{Alice}(o_1, \rho_1) &= 1/2, \\
\be_{Alice}(o_1, \rho_2) &= 1/2.
\end{align*}

We now compute $V_{\DT,\trustgame,Bob,\psi}^{\geq}(\rho_k)$. We recall from Example~\ref{example:trustgameguardfn} that, in state $s_{10}$, Bob's only possible intention change is $Bob.i.keep$, and, in state $s_{14}$, his only possible intention change is $Bob.i.share$. Since $\probm{\trustgame,Bob.i(\rho_k,keep)}(\psi) = 0$ and $\probm{\trustgame,Bob.i(\rho_k,share)}(\psi) = 1$, we obtain:
\begin{align*}
V_{\DT,\trustgame,Bob,\psi}^{\geq}(\rho_1) &= 0, \\
V_{\DT,\trustgame,Bob,\psi}^{\geq}(\rho_2) &= 1.
\end{align*}
Therefore:
\[
\trustgame, \rho_k \models \DT_{Alice,Bob}^{\geq ?}\psi = 1/2.
\]
Hence, Alice can trust Bob with probability at least 50\% in his willingness to share his profit with her. 

Note that an alternative way of formulating the trust we wish to compute is to use competence trust instead, in which case $\phi = \CT_{Alice,Bob}^{\geq ?}\psi$. Most of the computations stay the same in that case, the exception being the $V$-functions:
\begin{align*}
V_{\CT,\trustgame,Bob,\psi}^{\geq}(\rho_1) &= 1, \\
V_{\CT,\trustgame,Bob,\psi}^{\geq}(\rho_2) &= 1.
\end{align*}

Therefore:
\[
\trustgame, \rho_k \models \CT_{Alice,Bob}^{\geq 1}\psi.
\]

Hence, Alice can trust Bob with certainty on his capability to share his profit with her. We remark that the difference in the value of competence trust and disposition trust is as expected -- intuitively, Alice knows for sure that Bob is capable of sharing his profit with her, but she cannot be sure that he is willing to do that. 

Finally, we consider Alice's \emph{belief} that Bob will share money with her. We let $\phi = \B_{Alice}^{=?}\psi$ 
and recall Alice's intention preference function over Bob defined in Example~\ref{example:trustgamepref}:
\begin{alignat*}{2}
\ipref{Alice}{Bob}(s_i) & = \langle share \mapsto 3/4, keep \mapsto 1/4 \rangle && \quad \text{ for } i \in \{8,12\}, \\
\ipref{Alice}{Bob}(s_i) & = \langle share \mapsto 0, keep \mapsto 1 \rangle && \quad \text{ for } i \in \{10,14\}. 
\end{alignat*}

Hence, after Bob settles his intention, Alice's belief is as follows:
\begin{align*}
\be_{Alice}(\rho_1 s_{19}) &= 3/8, \\
\be_{Alice}(\rho_1 s_{20}) &= 1/8, \\
\be_{Alice}(\rho_2 s_{21}) &= 0, \\
\be_{Alice}(\rho_2 s_{22}) &= 1/2.
\end{align*}

Therefore:
\[
\trustgame, s_0s_2 \models \B_{Alice}^{=3/8}\psi.
\]

Intuitively, due to inaccuracy of Alice's prior knowledge about Bob (encoded in preference functions), her \emph{belief} that Bob will share his profit with her differs from her \emph{trust} towards Bob sharing the money. 
\hfill $\Box$
\end{example}

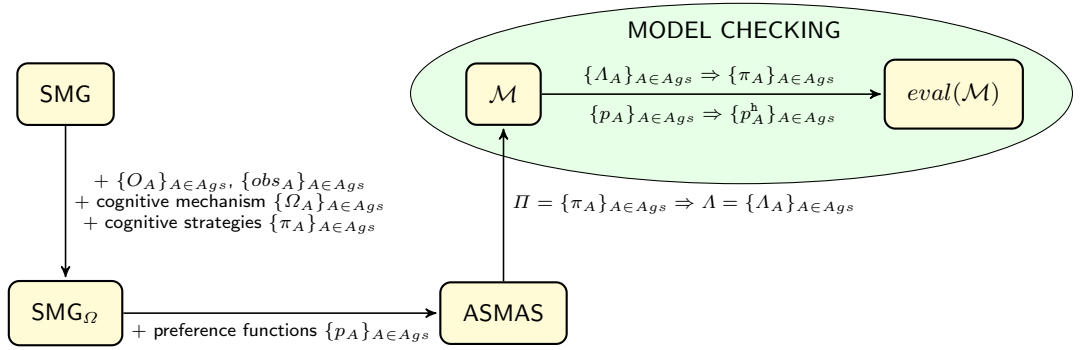
\begin{figure}
\begin{center}
\begin{tikzpicture}[
  font=\sffamily,
  every matrix/.style={ampersand replacement=\&,column sep=4.2cm,row sep=2cm},
  source/.style={draw,thick,rounded corners,fill=yellow!20,inner sep=.3cm},
  process/.style={draw,thick,circle,fill=blue!20},
  sink/.style={source,fill=green!20},
  datastore/.style={draw,very thick,shape=datastore,inner sep=.3cm},
  dots/.style={gray,scale=2},
  to/.style={->,>=stealth',shorten >=1pt,semithick,font=\sffamily\scriptsize},
  every node/.style={align=center}]

  \draw [black,fill=green!10] (3,1.45) ellipse (4.4cm and 1.2cm);
  \node[text width=5cm] at (2.9,2.3) {MODEL CHECKING};
  \matrix{
    \node[source] (smg) {\SMG}; 
    \& \node[source] (fasmas) {$\cM$};
    \& \node[source] (evalasmas) {$eval(\cM)$}; \\
      
    \node[source] (smgc) {\SMGC};
      \& \node[source] (asmas) {\AutoSMAS}; \& \\
  };

  \draw[to] (smg) -- node[midway,right] {+ $\{\Obs{A}\}_{A\in\Ags}$, $\{\obs{A}\}_{A \in \Ags}$ \\ + cognitive mechanism $\{\Cognition_A\}_{A\in\Ags}$ \\ + cognitive strategies $\{\cstrat{A}\}_{A\in\Ags}$} (smgc);
  \draw[to] (smgc) -- node[midway,below] {+ preference functions $\{\pref{A}\}_{A\in\Ags}$} (asmas);
  \draw[to] (asmas) -- node[midway,right] {$\CStrategy = \{\cstrat{A}\}_{A\in\Ags} \Rightarrow \Guard = \{\Guard_A\}_{A\in\Ags}$} (fasmas);
  \draw[to] (fasmas) -- node[midway,below] {$\{\pref{A}\}_{A\in\Ags} \Rightarrow \{\pref{A}^\history\}_{A\in\Ags}$} node[midway,above] { $\{\Guard_A\}_{A\in\Ags} \Rightarrow \{\cstrat{A}\}_{A\in\Ags}$ } (evalasmas);
  
\end{tikzpicture}
\end{center}
\caption{Model progression and model checking diagram}
\label{fig:flowdiagram}
\end{figure}

\section{Model Checking Complexity}\label{sec:overviewcomplexity}
In the next two sections we consider model checking \AutoSMAS\ against \PRTLS\ formulas. We first present the verification procedure and then analyse its complexity. 

We redefine the autonomous stochastic multi-agent system of Definition~\ref{def:autosmas} by replacing an infinite structure $\CStrategy$ with a finite structure $\Guard$, defined in Section~\ref{sec:proupdate}.

Our automated verification framework accepts as inputs  an \AutoSMAS\ $\cM$ satisfying Assumption~\ref{assump:detbeh} and a \PRTLS\ specification formula $\phi$.
The verification procedure proceeds as follows:
\begin{enumerate}
\item Perform pro-attitude synthesis on the system $\cM$ to obtain $\synth{}(\cM)$. In the new system, guarding mechanism $\Guard$ is replaced by cognitive strategies, $\CStrategy$, of agents.  
(In this phase, functions $\be_A$ are based on preference functions $\pref{A}$.)
\item Update the preference functions $\pref{A}$ into $\pref{A}^\history$ according to the results of pro-attitude synthesis, as defined in Section~\ref{sec:prefupdate}.
\item Model check formula $\phi$ on the system $\synth{}(\cM)$. (In this phase, the belief functions $\be_A$ are based on preference functions $\pref{A}^\history$.)
\end{enumerate}

The checking of expressions $\cM,\rho\models \phi$ in Steps 1 and 3 relies on the corresponding preference functions $\pref{A}$ and $\pref{A}^\history$, respectively.
However, for simplicity of the notations, we still write them as $\cM,\rho\models \phi$. Note that cognitive strategies synthesised in Step 1 generally lack a precise finite representation. Remark~\ref{rem:proattitudesynthesis} discusses a possible approximate solution, while Appendix~\ref{sec:bounded} presents a restriction on the system which ensures that a finite representation for $\CStrategy$ exists.
Figure~\ref{fig:flowdiagram} summarises the differences between the models introduced in this paper and outlines the model checking process.
Unfortunately, the problem is undecidable in general for both the synthesis step and the model checking step. First of all,
%
\begin{theorem}
The synthesis of 
pro-attitude functions whose formulas are in the language  $\Lang_A(\PRTLSf)$
is undecidable.
\end{theorem}

A guarding mechanism $\Guard = \{ \langle \goalGd{A},\intnGd{A} \rangle \}_{A\in\Ags}$ is trivial if, for all agents $A$, we have $\goalGd{A}(x)=true$ for all $x\in  \powerset{\Goal{A}} $ and $\intnGd{A}(x,y)=true$ for all $x \in \Intn{A}$ and $y\in \powerset{\Goal{A}}$. 
If 
the guarding mechanism is trivial then the first two steps can be skipped, and the computation proceeds directly to model checking. 
In this case, we have 

\begin{theorem}
Model checking \PRTLS\ is undecidable, even if the guarding mechanism is trivial. 
\end{theorem}

We note that both of the above problems are undecidable even for formulas concerning beliefs of a single agent. Proofs can be found in Appendix~\ref{sec:general}.

%


\section{Decidable Fragments}\label{sec:decidable}

\newcommand{\depth}{{\tt d}}

To counter the undecidability results, we
explore fragments of the general problem. 
In this section, we only present proof ideas for the complexity results.
The details of the proofs can be found in the Appendix.
Table~\ref{tab:fragmentsComplexity} summarises complexity results for the proposed fragments.

\begin{table}
\begin{center}
\caption{Complexity of decidable fragments}
\label{tab:fragmentsComplexity}
\begin{tabular}{c|c}
\textbf{\BPRTLS} & PSPACE-hard\\
\hline
\textbf{\PRTLSS} & PSPACE-complete \\
\hline
\textbf{\PQRTLSS} & PTIME \\
\end{tabular}
\vspace{1em}
\end{center}
\end{table}

\begin{table}
\caption{Assumptions of decidable fragments}
\label{tab:fragmentsAssumptions}
\begin{center}
\begin{tabular}{ccc}
\toprule
\textbf{\BPRTLS} & \textbf{\PRTLSS} & \textbf{\PQRTLSS} \\
\midrule
no $\until$ or $\always$ operators & $\Guard$ is trivial & $\phi = \always(\psi\Rightarrow \prob{\bowtie q}{\eventually\B_{A}^{\geq 1}\psi})$  \\
all $\next$ preceded by $\prob{}{}$, $\forall$ or $\exists$  & $\B_A$, $\T_A$ for single $A$ only & or $\always(\psi\Rightarrow \prob{\bowtie q}{\eventually\T_{A,B}^{\geq 1}\psi})$ \\
constant depth of $\B$, $\T$ nesting  & no nesting of $\B$ and $\T$ & \\
 & no $\B$ or $\T$ in scope of $\prob{}{}$ &  \\
 & constant \# of $\B$ and $\T$ &  \\
 & ``expanded system restriction'' &  \\
\bottomrule
\end{tabular}
\vspace{1em}
\end{center}
\end{table}

\subsection{A Bounded Fragment} The bounded fragment, named \BPRTLS, allows formulas which 1) do not contain temporal operators $\until$ and $\always$, 2) all $\next$ operators are immediately prefixed with a probabilistic operator or a branching operator, i.e., in a combination of $\prob{\bowtie q}{\next\psi}$, $\forall\next\psi$, or $\exists\next\psi$, and 3) the nested depth of belief and trust operators is
constant. The last constraint is needed to ensure that the complexity is not measured over the nested depth of belief and trust operators.
For the upper bound, we have the following result.
\begin{theorem}
The complexity of the bounded fragment \BPRTLS\ is in EXPTIME.
\end{theorem}
In Appendix~\ref{sec:bounded}, we present an EXPTIME algorithm which performs the verification procedure from Section~\ref{sec:overviewcomplexity}. It uses boundedness of the logic fragment to limit the length of paths that need to be considered to evaluate a guard or the specification formula. It then constructs an expanded system whose states capture all past observations of all agents, enabling belief and trust formulas to be evaluated locally. Model checking the original system is then reduced to model checking the expanded system, which admits a standard recursive labelling procedure, similar to the one used for \PCTL\ model checking.

\commentout{
Let $\cM$ be an \AutoSMAS\ with a set of agents $\Ags=\{1,...,n\}$,
$\phi$ be the specification formula and let $\depth(\phi)$ denote the maximal length of path needed to verify $\phi$. Recall that $\Obs{i}$ denotes the set of observations of agent $i$.
For $0\leq k\leq \depth(\phi)-1$, we let $O^k=\{\bot\}\cup (\Obs{1}\times ...\times \Obs{n})$ be the set of possible observations at time $k$, where $\bot$ denotes that agents have not made any observations at time $k$.
Then the size of $O^k$ is $O(S^{|\Ags|})$. The algorithm proceeds by constructing an expanded system $\cM^\#$ whose state space is $S\times O^0\times ...\times O^{\depth(\phi)-1}$. Therefore, the size of the state space is $O(S^{|\Ags|*\depth(\phi)})$, i.e., polynomial with respect to the size of the system but exponential with respect to the size of the formula.
Then we have an algorithm which can complete all three steps of the verification procedure in polynomial time with respect to the size of the system $\cM^\#$.
}
For the lower bound, we have the following result.
\begin{theorem}
The complexity of the bounded fragment \BPRTLS\ is PSPACE-hard.
\end{theorem}
The proof is a reduction from the quantified Boolean formula problem. Given a QBF $\phi$, the reduction constructs an \AutoSMAS\ with two agents, $\forall$ and $\exists$, playing a game of two phases. In the first phase, the agents decide the truth values of the variables (atomic propositions) of $\phi$ by changing their intentions. Then, in the second phase, the agent $\exists$ attempts to show that a randomly chosen clause of $\phi$ evaluates to true by proving that one of its literals is true. If they can do that regardless of which clause gets selected, then they have a winning strategy in the game, which is equivalent to $\phi$ being true. The challenge lies in designing the system so that agents' decisions from phase one are ``remembered'' until phase two, and devising an appropriate formula to express the existence of a winning strategy for $\exists$.
The details of the proof are in Appendix~\ref{sec:bounded}.

\commentout{
The proof is a reduction from the QBF problem. In the reduced system, there are two agents, $\forall$ and $\exists$, playing a game of two phases. In the first phase, the agents decide the truth values of the variables by changing their intentions, and then, in the second phase, the agent $\exists$ justifies its choices by claiming that the game can always move into situations where it has  qualitative beliefs about some propositions. In the reduction, we cannot encode agents' decisions into the states, because that is not a PTIME reduction. Instead, we take advantage of agent $\exists$'s partial observations to work on a set of game plays which it is not able to distinguish. Because the set of game plays encodes the decisions made by both agents, agent $\exists$ can justify its choices by testing its qualitative beliefs.
The details of the proofs are in Appendix~\ref{sec:bounded}.
}

\subsection{A Fragment with $\until$ and $\always$ operators}

The fragment of this section, named \PRTLSS, works with the $\until$ and $\always$ temporal operators in the specification formula $\phi$. With $\until$ and $\always$ operators, the specification formula can express long-run properties about agent's mental attitudes in the system, and therefore this fragment complements the bounded fragment. However, the fragment is subject to other restrictions as follows.

\begin{itemize}
\item 
The guarding mechanism $\Guard $ is trivial. 
That is, the algorithm works with the model checking step without conducting pro-attitude synthesis first.
\commentout{
\footnote{The pro-attitude synthesis problem is strictly harder than model checking for long-run properties because it needs a finite representation for the set $\synth{}(x)$ of paths, and for this reason this fragment and \PRTLSS\ do not handle synthesis. In a companion paper,
we aim to formulate an approach based on order effects~\cite{Asch1946,RNLO2007,HE1992}
%
 to approximate the set $\synth{}(x)$. The focus of this paper is on approaches that can be computed precisely. }
 }

\item The specification formula $\phi$ is restricted so that 1) it works with
a single agent's beliefs and trust,
\commentout{\footnote{We note that, although this fragment considers beliefs and trust of a single agent, the autonomous system may contain multiple agents. }} 2) there are no nested beliefs or trust, 3) beliefs and trust cannot be in the scope of a probabilistic operator $\prob{}{}$, and 4) there is a constant number of belief or trust operators.

\item The expanded system, described below, satisfies 
the restriction
 that, along any path of the expanded system, the evaluation of belief or trust subformulas for agent $A$
are the same on any two expanded states if the support of the probabilistic distributions over the set of states are the same.

\end{itemize}
\begin{remark}\label{rem:proattitudesynthesis}
We remark that the pro-attitude synthesis problem is strictly harder than model checking for long-run properties because it needs a finite representation for the set $\synth{}(x)$ of paths. For this reason, this fragment and \PRTLSS\ do not handle synthesis. Future work may involve formulating an approach based on order effects~\cite{Asch1946,RNLO2007,HE1992} to approximate the set $\synth{}(x)$. The focus of this paper is on approaches that can be computed precisely. 
\end{remark}

The restriction on the expanded system is essential. It results from an investigation into the undecidability, which has established that the undecidable cases arise because of the existence of non-monotonic flows of probability between states in two consecutive visits of the same expanded state. By assuming stability (as we do in this fragment) or monotonicity of the flows, decidability can be achieved.
Moreover, we note that this fragment has full expressiveness of LTL. In the following, we assume that it is the agent $A$'s beliefs that the  formula $\phi$ is to describe. 
We have the following result for the upper bound.
\begin{theorem}
The complexity of the fragment \PRTLSS\ is in PSPACE.
\end{theorem}
Given a system $\cM$ and a formula $\phi$, the algorithm proceeds by first computing those subformulas which do not have belief or trust formulas in their scope. These computations can be done by an adaptation of the usual probabilistic temporal logic model checking. Then, to compute a formula with belief and trust subformulas, we construct an expanded system $\cM_A(\phi)$ with the state space $S^\dagger=S\times \dist{S}\times R \times ... \times R$, where $R$ is the domain of real numbers with fixed precision\footnote{We assume that each number can be encoded with complexity $O(1)$.}. Although $R$ is infinite, we need only a finite number of real numbers. The restriction imposed on the system ensures that the system $\cM_A(\phi)$ is of size exponential with respect to the size of the system $\cM$. However, we do not explicitly construct the system, but instead take an automata-theoretic approach and perform on-the-fly computation  by, e.g., two-phase search. The belief and trust formulas can be directly evaluated on the states of the expanded system. The details of the construction are in Appendix~\ref{sec:mc}.

We have the following results for the lower bound. Recall that a strongly connected component (SCC) of a graph $G$ is a subgraph of $G$ in which every vertex is reachable from every other vertex.
\begin{theorem}
The complexity of the fragment \PRTLSS\ is PSPACE-hard.
\end{theorem}
The proof is a reduction from the language universality problem of nondeterministic finite automata. Given an NFA $A=(Q,q_0,\delta,F)$, we construct a single agent system $\cM(A)$ which, starting from an initial state $s_0$, moves into one of two subsystems $\cM_1(A)$ and $\cM_2(A)$, with a uniform distribution. The system $\cM_1(A)$ simulates the universal language $\Sigma^*$, while $\cM_2(A)$ simulates the language of $A$. In other words, every word in $\Sigma^*$ has a corresponding execution in $\cM_1(A)$ and every word in $\Sigma^*$ that produces a valid run of $A$ has a corresponding execution (one for each unique run) in $\cM_2(A)$ which ends in a designated state (one labelled $\literal{finished}$) if the corresponding run on $A$ ends in a final state. We define the observation function so that the agent cannot distinguish between paths in $\cM_1(A)$ and $\cM_2(A)$ corresponding to the same word. This then allows us to prove that the non-universality of the language of A is equivalent to the satisfaction of a carefully designed formula revolving around agent's qualitative belief that the system reaches a designated state. Finally, the reduction is completed by showing that the system $\cM(A)$ satisfies the restriction outlined above, which follows by observing that agent's belief stays constant in all SCCs of the expanded system.
%
The details of the proof are in Appendix~\ref{sec:mc}.

\subsection{A Polynomial Time Fragment}

One of the restrictions in the previous fragment is that beliefs and trust cannot be in the scope of a probabilistic operator. In that fragment, this restriction ensures that the complexity is in PSPACE, because it can be harder to perform quantitative verification of the probabilistic and belief operators. In this section, we partly complement this with a fragment in which the belief or trust operators can be in the scope of a probabilistic operator but need to be qualitative, i.e., almost sureness. We show that this fragment, named \PQRTLSS, is, very surprisingly, polynomial time. Note that, as before, $\T$ ranges over two trust operators $\CT$ and $\DT$.

We consider the model checking step and restrict formulas to be of the form $\always(\psi\Rightarrow \prob{\bowtie q}{\eventually\B_A^{\geq 1}\psi})$ or $\always(\psi\Rightarrow \prob{\bowtie q}{\eventually \T_{A,B}^{\geq 1}\psi})$ such that, in $\psi$, there are no  belief or trust operators and every temporal operator is immediately prefixed with a branching operator, i.e., in the style of CTL. The system $\cM$ needs to satisfy the formula $\cM\models \always (\psi\Rightarrow \always \psi)$, which means that, once $\psi$ holds, it will hold henceforth.
The following is the result.
\begin{theorem}
The complexity of the fragment \PQRTLSS\ is in PTIME.
\end{theorem}
We give a polynomial time algorithm in Appendix~\ref{sec:ptime}. The algorithm first constructs a system $\cM^\#$ in which two copies of system $\cM$ run by synchronising their observations and the second copy avoids the states where $\psi$ is satisfiable. The computation is then to find strongly connected components (SCCs) of $\cM^\#$  such that a) the formula $\psi$ is satisfiable on some state of the first copy, b) both copies are closed SCCs, and c) the two copies are equivalent in terms of the acceptance probability of words. We show that the checking of formulas $\prob{\bowtie q}{\eventually\B_A^{\geq 1}\psi}$ or $\prob{\bowtie q}{\eventually \T_{A,B}^{\geq 1}\psi}$ on states satisfying $\psi$ is equivalent to comparing the reachability probability of these SCCs with the value $1-q$. The PTIME complexity is due to the fact that the computation of SCCs and checking the three conditions on SCCs can be done in polynomial time; in particular, checking of equivalence of acceptance probability of all words can be done in polynomial time by~\cite{Tzeng1992}.

\begin{table}
\caption{Assumptions of the framework}
\label{tab:assumptions}
\begin{center}
\begin{tabular}{l}
\toprule
1. Deterministic Behaviour Assumption (Assumption~\ref{assump:detbeh})\\
2. Uniformity Assumptions (Assumptions~\ref{assump:uni},~\ref{assump:uniII}) \\
3. Transition Type Distinguishability Assumption (Assumption~\ref{assump:typeObs}) \\
4. Synchronous perfect recall (Remark~\ref{rem:syncPerfRecall}) \\
5. \PCTLS\ formulas are evaluated in induced SMGs (Remark~\ref{rem:pctlsEval}) \\
6. Constraints on the language $\Lang_A(\PRTLSf)$ \\
7. Seriality of \legal\ goal and intention functions \\
\bottomrule
\end{tabular}
\end{center}
\end{table}

\section{Conclusions}\label{sec:concl}

The paper proposes an
automated verification framework for autonomous stochastic multi-agent systems and specifications given in probabilistic rational temporal logic \PRTLS, which includes novel modalities for quantifying and reasoning about agents' cognitive trust. Our model is an extension of stochastic multiplayer games with cognitive reasoning, which specifies how agents' goals and intentions change during system execution, and admits probabilistic beliefs on which the trust concepts are founded. We study computational complexity of the decision problems and show that, although the general problem is undecidable, there are decidable, even tractable, fragments. Their existence is made possible by numerous assumptions and restrictions that we place on our system, which are summarised in Table~\ref{tab:assumptions}.

As can be seen from the illustrative running example in this paper, the framework is applicable to a wide range of scenarios of human-robot interactions. This includes competitive settings such as the trust game, as well as cooperative scenarios, which are more commonly considered in robotics community, such as a table clearing task \cite{Chen18TrustAwarePlanning}. Furthermore, the development of trust sensors (see \cite{Akash18TrustSense}) complements our framework very well; they could serve as a source of agent preferences and for validation purposes.

A natural next step following the development of the framework is implementing its techniques in the form of a model checker. To overcome undecidability, a subset of the problem will be considered, in particular, the bounded fragment \BPRTLS. Since a finite set of finite paths is sufficient to model check a bounded specification formula, pro-attitude synthesis only needs to consider those finite paths and the resulting cognitive strategies have a finite representation. Preference function update also only needs to be performed on the relevant finite paths. Finally, computing the satisfaction of a specification formula is a recursive procedure that keeps track of the execution history to evaluate belief and trust formulas. Other decidable fragments also provide a basis for implementation, but their strong restrictions make them less practical.

Another interesting direction for future works involves investigating how memory decay can be introduced into the framework. Intuitively, humans tend to remember more recent experiences better and this should be reflected in the semantics of trust. As a side effect, decidability will be achieved if appropriately old memory is discarded. With that, a Bellman operator may be defined which will allow more efficient evaluation of trust.

Finally, it would be interesting to study an axiomatisation for our logic, as well as the satisfiability problem, which we believe to be challenging.


\begin{acknowledgments}
\emph{The authors are supported by ERC Advanced Grant VERIWARE and EPSRC Mobile Autonomy Programme Grant EP/M019918/1.}
\end{acknowledgments}

\bibliographystyle{splncs04}
\bibliography{trust}

\clearpage
\onecolumn

\newpage

\appendix

\section{Undecidability of the General Problem}\label{sec:general}

The undecidability proof
is by a reduction from the emptiness problem and strict emptiness
problem of probabilistic automata, both known as undecidable
problems~\cite{Paz71}.
A probabilistic automaton $PA$ is a tuple $(Q, A, (M_a)_{a\in A},
q_0, F)$,~where\begin{itemize}
\item $Q$ is a finite set of states,
\item $q_0$ is the initial state,
\item $F\subseteq Q$ is a set of accepting states,
\item $A$ is the finite input alphabet, and
\item $(M_a)_{a\in A}$ is the set of transition matrix.
\end{itemize}
For each $a\in A$, $M_a\in[0,1]^{Q\times Q}$ defines transition probabilities, such that given $q, q'\in Q$, $M_a(q,q')$ is the probability that $q$ makes a transition to $q'$ when $a$ is the input. For every $q\in Q$ and $a\in A$, we have $\sum_{q'\in Q}M_a(q,q')=1$. Plainly, given a state $q$, an input $a$ makes a transition to a distribution on $Q$, and we further extend $M_a$ to be a transformer from distributions to distributions. Let $\dist{Q}$ be the set of all probabilistic distributions on the set $Q$. Given $\Delta\in\dist{Q}$, we write $M_a(\Delta)$ for the distribution transformed from $\Delta$ by $a$, such that for all $q'\in Q$, $M_a(\Delta)(q')=\sum_{q\in supp(\Delta)}\Delta(q)\cdot M_a(q,q')$. Given $w=a_1\cdot a_2\cdot\ldots \cdot a_n\in A^*$, we write $M_w$ for the function $M_{a_n}\circ M_{a_{n-1}}\circ \dots\circ M_{a_1}$ (we assume function application is right associative).

The emptiness problem of a probabilistic automata is defined as follows: 
Given a probabilistic automaton $PA=(Q, \sep A, \sep (M_a)_{a\in A}, \sep q_0, \sep F)$ and $\epsilon\in[0,1]$, decide whether there exists a word $w$ such that $M_w(\Delta_0)(F)\geq\epsilon$, where $\Delta_0(q_0)=1$ and $\Delta_0(q)=0$ for $q\in Q\setminus\{q_0\}$. 
Replacing `$\geq$' by a strict inequality `$>$' yields the strict emptiness problem. Both problems are known to be undecidable.

\paragraph{Pro-attitude Synthesis is Undecidable}
First of all, we show that the synthesis of pro-attitude functions in which the formulas are in the language  $\Lang_A(\PRTLSf)$  is undecidable.
%
%
Formally, let $PA=(Q, \sep A, \sep (M_a)_{a\in A}, \sep q_0, \sep F)$ be a probabilistic automaton and $AP=\{\literal{final},\literal{found}\}$ be a set of atomic propositions. We can then construct an ASMAS $\cM(PA)=(\Ags, \states, \sinit, \Act{1},
T,  L, \Obs{1}, \obs{1}, \Cognition_1, \cstrat{1}, \pref{1})$  such that
\begin{itemize}
\item $\Ags=\{1\}$, i.e., this is a single-agent system,
\item $\states=A\times Q\times \{1,2\}$,
\item $\sinit =\{(a,q_0,1)\}$ for some $a\in A$,
\item $\Act{1}=\{\tau\}$,
\item the transition relation is as follows for $k\in\{1,2\}$:
$$T((a,q,k),\tau,(a',q',k))=\frac{M_{a'}(q,q')}{\sum_{a'\in A}\sum_{q'\in Q}M_{a'}(q,q')}$$
\item $\literal{found} \in  L((a,q,2))$ for $a\in A$ and $q\in Q$, and $\literal{final} \in  L((a,q,k))$ for $a\in A$, $q\in F$, and $k\in\{1,2\}$.
\end{itemize}
Intuitively, the system consists of two subsystems, indexed with the third component of the states (i.e., $1$ or $2$),  running in parallel without any interaction. The observation is defined as follows: 
\begin{itemize}
\item $\Obs{1} = A$, and 
\item $\obs{1}((a,q,k))=a$ for all $(a,q,k)\in \states$. 
\end{itemize}
The moving from the first subsystem to the second subsystem is done by agent $1$'s intentional changes, which are guarded with a testing on agent $1$'s beliefs.
We only define relevant intentional attitudes in $\Cognition_1$ as follows.
\begin{itemize}
\item $\Intn{1}=\{x_1,x_2\}$, 
\end{itemize}
As defined in Section~\ref{sec:proupdate}, the cognitive strategy $\cstrat{1}$ can be obtained from $\Cognition_1$ and $\Guard$, with the latter defined as follows for the intentional strategy. 
\begin{itemize}
\item $\intnGd{1}(x_2)=\B_1^{\geq (>) \epsilon} final$ for some $\epsilon \in [0,1]$ and $\intnGd{1}(x_1)=true$ 
\end{itemize}
The preference functions $\pref{1}$, the goal attitude $\Goal{1}$ in $\Cognition_1$, and the goal guards $\goalGd{1}$  are not defined, because they are not used in this reduction. The initial distribution $\initdist$ is a Dirac distribution over the initial states $\sinit$.
Therefore, given a $PA$ and a number $\epsilon\in [0,1]$, the (strict) emptiness of the problem $M_w(\Delta_0)(F)\geq (>) \epsilon$ is  equivalent to checking whether $\cM(PA)\models \exists\eventually ~\overline{\I_1} \literal{found}$.

\paragraph{Model Checking \PRTLS\ without Pro-Attitudes is Undecidable}
In the following, we show that model checking \PRTLS\ is also undecidable, for systems where the guarding mechanism $\Guard$ is trivial. 
%
Formally, let $PA$ be as before and $AP=\{\literal{final}\}$ be a set of atomic propositions. We then construct an \AutoSMAS\ $\cM(PA)=(\Ags, \states, \sinit, \Act{1}, 
T,  L, \Obs{1}, \obs{1}, \Cognition_1, \cstrat{1}, \pref{1})$ such that
\begin{itemize}
\item $\Ags=\{1\}$, i.e., this is a single-agent system, 
\item $\states=A\times Q$,
\item $\sinit = \{(a,q_0)\}$ for some $a\in A$, with the initial distribution $\initdist$ being a Dirac distribution over $\sinit$,
\item $\Act{1}=\{\tau\}$,
\item the transition relation $T$ is as follows:
$$T((a,q),\tau,(a',q'))=\frac{M_{a'}(q,q')}{\sum_{a'\in A}\sum_{q'\in Q}M_{a'}(q,q')}$$
\item $\literal{final} \in  L((a,q))$ for $a\in A$ and $q\in F$.
\end{itemize}
Agent $1$'s observation is defined as 
\begin{itemize}
\item $\Obs{1} = A$, and 
\item $\obs{1}((a,q))=a$ for all $(a,q)\in \states$. 
\end{itemize}
We do not need the functions $\Cognition_1, \cstrat{1}, \pref{1}$ in this reduction.
Given a $PA$ and a number $\epsilon\in [0,1]$, the (strict) emptiness of the problem $\cM_w(\Delta_0)(F)\geq (>) \epsilon$ is  equivalent to checking whether  $\cM(PA)\models \exists\eventually \B_1^{\geq (>) \epsilon} \literal{final}$.

\section{A Decidable Fragment of Bounded Length}\label{sec:bounded}

As shown in Section~\ref{sec:general}, the automated  verification problem defined in Section~\ref{sec:overviewcomplexity} is undecidable in general.
In this section, we present a fragment
\commentout{
\footnote{In the following discussions on decidable fragments (i.e., Appendix~\ref{sec:bounded}, \ref{sec:mc}, and \ref{sec:ptime}), we will only handle dispositional trust formula $\DT_{A,B}^{\bowtie q}\psi$, and write it directly as $\YT_{A,B}^{\bowtie q}\psi$. The competence trust formula $\CT_{A,B}^{\bowtie q}\psi$ can be handled in a similar way.} 
}
of the problem whose computational complexity falls between PSPACE and EXPTIME. Note that in the following discussions on decidable fragments (i.e., Appendix~\ref{sec:bounded}, \ref{sec:mc}, and \ref{sec:ptime}), we only consider dispositional trust formula $\DT_{A,B}^{\bowtie q}\psi$. The competence trust formula $\CT_{A,B}^{\bowtie q}\psi$ can be handled in a similar way.

\paragraph{Bounded Fragment}


The bounded fragment works with specification formulas $\phi$ which do not contain temporal operators $\until$ and $\always$,  all $\next$ operators are immediately prefixed with a probabilistic operator or a branching operator, i.e., in a combination of $\prob{\bowtie q}{\next\psi}$, $\forall \next \psi$, or $\exists \next \psi$, and the nested depth of belief and trust operators is constant.
We remark that the specification formula $\phi$ can be extended to include subformulas of the form $\prob{\bowtie q}{(\psi_1\until^{\leq k}\psi_2)}$, $\forall(\psi_1\until^{\leq k}\psi_2)$, or $\exists(\psi_1\until^{\leq k}\psi_2)$.
Moreover, the restriction on nested temporal operators can be relaxed by taking
the bounded semantics for LTL.
We focus on the simpler syntax to ease the notation.

Let $\depth(\phi)$ be the maximal length of the paths that are needed for  the specification formula $\phi$. Specifically:
\begin{itemize}
\item $\depth(p)=0$,
\item $\depth(\phi_1\lor \phi_2)=\max \{\depth(\phi_1), \depth(\phi_2)\}$,
\item $\depth(\neg \psi)=\depth(\forall\psi)=\depth(\prob{\bowtie q}{\psi})=\depth(\B_A^{\bowtie q}\psi)=\depth(\psi)$, and
\item $\depth(\next\psi) =  \depth(\G_A\psi)=\depth(\I_A\psi) = \depth(\C_A\psi)=\depth(\DT_{A,B}^{\bowtie q}\psi)= \depth(\psi) +  1$.
\end{itemize}\

\subsection{Upper  Bound}

The algorithm proceeds in three steps according to the verification procedure in Section~\ref{sec:overviewcomplexity}.  

\subsubsection{Pro-attitude Synthesis}

\newcommand{\clk}{{\tt clk}}
\newcommand{\rP}{{\tt rP}}
\newcommand{\orig}{{\tt M}}

The purpose of the synthesis is to realise the guarding mechanism, i.e., for every $x\in \powerset{\Goal{A}}$ and $y\in \Intn{A}$, compute an equivalent representation for the formulas $\goalGd{A}(x)$ and $\intnGd{A}(x,y)$ for all agents $A$. Here we only consider formulas  of the forms $\B_A^{\bowtie q}\phi$ and $\DT_{A,B}^{\bowtie q}\phi$, and claim that formulas of the other forms $\B_A^{=?}\phi$ and $\DT_{A,B}^{=?}\phi$ can be handled by adapting the technique slightly.  Without loss of generality, we let $x\in \powerset{\Goal{A}}$. By the constraint of the language $\Lang_A(\PRTLSf)$ that no temporal operator can be in the scope of a belief operator, the satisfiability of $M,\rho\models \goalGd{A}(x)$ depends only on those paths of length $|\rho|$. Moreover, by the semantics of the language \PRTLS, model checking the bounded specification formula $\phi$ requires only those paths of length not greater than $\depth(\phi)+1$. Therefore, for this fragment, the synthesis is equivalent to finding a set of paths $\synth{}(x,\depth(\phi))$, where
\[
\synth{}(x,k) = \{\rho\in \fpath{}^\cM~|~|\rho|\leq k, x\in\Goal{A}(last(\rho)), \cM,\rho\models \goalGd{A}(x)\}
\]
for
$1\leq k\leq \depth(\phi)$.

%

Given a state $s\in S$, we write $\obs{}(s)=(\obs{1}(s),...,\obs{n}(s))$ for the tuple of agents' observations at state $s$.
For $0\leq k\leq \depth(\phi)-1$, we let $O^k=\{\bot\}\cup (\Obs{1}\times ...\times \Obs{n})$ be the set of possible observations at time $k$, where $\bot$ denotes that agents have not made any observations at time $k$.  The state space of the following expanded system $\cM^\#$ is $S\times O^0\times ...\times O^{\depth(\phi)-1}$. For a state $s^{\#}=(s,o^0,...,o^{\depth(\phi)-1})$ of $\cM^\#$, we write $ L_{\orig}(s^{\#})$ for original state $s$ and $ L^k(s^\#)$ with $0\leq k\leq \depth(\phi)-1$  for the observations at time $k$, i.e., $ L^k(s^\#)=o^k\in O^k$. Moreover, for a tuple of observations $ L^k(s^\#)$, we write $ L_A^k(s^\#)$ for agent $A$'s observation. These notations can be extended to work with a sequence of states whenever reasonable, e.g., $ L_{\orig}(s^{\#}_0s^{\#}_1...)= L_{\orig}(s^{\#}_0) L_{\orig}(s^{\#}_1)...$, etc.

Given  $\asmasTuple$,
we construct a system $\cM^\#=(\Ags, \states^\#,\sinit^\#,\sep\{\Act{A}\}_{A\in\Ags},
\sep\{T_k^\#\}_{0\leq k\leq \depth(\phi)-1},\sep L^\#,\sep\{\Obs{A}^\#\}_{A\in\Ags}, \sep\{\obs{A}^\#\}_{A\in\Ags},\sep\{\Cognition_A^\#\}_{A\in\Ags}, \sep\{\cstrat{A}\}_{A\in\Ags}, \{\pref{A}^\#\}_{A\in\Ags})$ such that
\begin{itemize}
\item $S^\#=S\times O^0\times ...\times O^{\depth(\phi)-1}$,
\item for $s^\# \in \states^\#$, $s^\#\in \sinit^\#$  whenever $\initdist^\#(s^\#)>0$, where $\initdist^\#(s^\#) = \initdist( L_{\orig}(s^\#))$ if $ L^0(s^\#)=\obs{}( L_{\orig}(s^\#))$ and $ L^k(s^\#)=\bot$ for $1\leq k\leq \depth(\phi)-1$, and $\initdist^\#(s^\#)=0$ otherwise,
\item $T_k^\#(s_1^\#,a,s_2^\#)=T( L_{\orig}(s_1^\#),a, L_{\orig}(s_2^\#))$ if
\begin{itemize}
\item $ L^{j}(s_2^\#)= L^j(s_1^\#)$ for all $0\leq j\leq k$,
\item $ L^{k+1}(s_1^\#)=...= L^{\depth(\phi)-1}(s_1^\#)= L^{k+2}(s_2^\#)=...= L^{\depth(\phi)-1}(s_2^\#)=\bot$, and
\item $ L^{k+1}(s_2^\#)=\obs{}( L_{\orig}(s_2^\#))$.
\end{itemize}
and $T_k^\#(s_1^\#,a,s_2^\#)=0$ otherwise,
\item $ L^\#(s^\#)= L( L_{\orig}(s^\#))$,
\item $\Obs{A}^\#$ and  $\obs{A}^\#$ are to be defined later, 
\item for $\Cognition_A^\#$ with  $A\in\Ags$, we have $\Goal{A}^\#(s^\#)=\Goal{A}(L_{\orig}(s^\#))$ and $\Intn{A}^\#(s^\#)=\Intn{A}(L_{\orig}(s^\#))$ for all $s^\# \in \states^\#$, and 
\item for $\pref{A}^\#$, we have $\gpref{A}{B}^\#(s^\#)=\gpref{A}{B}(L_{\orig}(s^\#))$ and $\gpref{A}{B}^\#(s^\#)=\gpref{A}{B}(L_{\orig}(s^\#))$ for all $A,B\in\Ags$.  
\end{itemize}
Intuitively, in the new system $\cM^\#$, agents' observation history  are remembered in the state, and for every time $0\leq k\leq \depth(\phi)-1$, a separate transition relation $T_k^\#$ is constructed. The transition relation  $T_k^\#$ maintains the previous observation history up to time $k$ and adds a new observation $ L^{k+1}(s_2^\#)$ to the next state $s_2^\#$.

Before evaluating $\Lang_A(\PRTLSf)$ formulas, we define a new belief function $\be_A^\#$ on the constructed system $\cM^\#$. With this function, we can evaluate belief formulas \emph{locally} instead of resorting to the observation history as we did for the original system $\cM$. 
A sequence $\rho^\#=s_0^\#...s_{\depth(\phi)-1}^\#$ of states in the expanded system $\cM^\#$ is a path if, for all $0\leq k\leq \depth(\phi)-1$, one of the following conditions holds:
\begin{itemize}
\item $T_k^\#(s_k^\#,a,s_{k+1}^\#)>0$ for some joint action $a\in\Act{}$.
\item there exist  $B\in\Ags$ and  $x\in \Goal{B}( L_{\orig}(s_k^\#))$ such that
\begin{itemize}
\item $L_{\orig}(s_k^\#)\ctrans{B.g.x}  L_{\orig}(s_{k+1}^\#)$
\item $ L^{j}(s_{k+1}^\#)= L^j(s_k^\#)$ for all $0\leq j\leq k$,
\item $ L^{k+1}(s_k^\#)=...= L^{\depth(\phi)-1}(s_k^\#)= L^{k+2}(s_{k+1}^\#)=...= L^{\depth(\phi)-1}(s_{k+1}^\#)=\bot$, and
\item $ L^{k+1}(s_{k+1}^\#)=\obs{}( L_{\orig}(s_{k+1}^\#))$.
\end{itemize}
\item Similar for a transition on which some agent $B\in\Ags$ changes its intention.
\end{itemize}
The following proposition states that there is a 1-1 correspondence relation between paths in $\cM$ and paths in $\cM^\#$.
\begin{proposition}
For every path $\rho^\#=s_0^\#...s_{\depth(\phi)-1}^\#$ in $\cM^\#$ such that $\initdist^\#(s_0^\#)>0$, we have that $ L_{\orig}(\rho^\#)$ is an initialised path of $\cM$. On the other hand, for every initialised path $\rho$ in $\cM$ whose length is no more than $\depth(\phi)+1$, there exists 
exactly one
path $\rho^\#$ in $\cM^\#$ such that $ L_{\orig}(\rho^\#)=\rho$.
\end{proposition}
%
Note that the transition relation of the constructed system $\cM^\#$ is acyclic, that is, for every state $s^{\#}$, there does not exist a cyclic path $\rho^\#\in \fpath{}^\cM(s^\#)$ such that $s^\#=\rho^\#(m)$ for some $m\geq 1$. Therefore, we can define $\clk(s^\#)$ as the time in which the state $s^\#$ appears, or formally, $\clk(s^\#)=k-1$ for $k$ the greatest number such that $ L^k(s^\#)\neq \bot$.

We can now define a probabilistic function $T_{k,A}^\#$ for the system $\cM^\#$ and agent $A\in\Ags$ as follows, by considering the preference functions of agent $A$, where $\clk(s_1^\#)=k$ and we write $\type(s_1^\#,s_2^\#)$ for $\type( L_{\orig}(s_1^\#), L_{\orig}(s_2^\#))$. We note that $s_1^\#s_2^\#$ needs to be on some path of $\cM^\#$.
$$
\displaystyle
T_{k,A}^\#(s_1^\#,s_2^\#) = \left \{
\begin{array}{cl}
T_k^\#( s_1^\#,a, s_2^\#) & \text{ if } \type(s_1^\#,s_2^\#)=a  \\
\gpref{A}{B}( L_{\orig}(s_1^\#))(x) & \text{ if } \type(s_1^\#,s_2^\#)=B.g \text{ and } L_{\orig}(s_1^\#) \ctrans{B.g.x} L_{\orig}(s_2^\#)\\

\ipref{A}{B}( L_{\orig}(s_1^\#))(x) & \text{ if } \type(s_1^\#,s_2^\#)=B.i \text{ and } L_{\orig}(s_1^\#) \ctrans{B.i.x} L_{\orig}(s_2^\#)\\ 

1 & \text{ if }  \type(s_1^\#,s_2^\#)=A.g.x \text{ for some } x \in \Goal{A}(s_1^\#) \\  
& \text{ or }  \type(s_1^\#,s_2^\#)=A.i.x \text{ for some } x \in \Intn{A}(s_1^\#) \\  

\end{array}
\right.
$$

%
With the non-acyclic property, we can define for every state $s^\#$ a reachability probability $\rP_A(s^\#)$ as follows with respect to a type $t$ of paths.
\begin{itemize}
\item $\rP_A(s^\#)=\initdist^\#(s^\#)$ if $\clk(s^\#)=0$
\item $\rP_A(s^\#) = \sum_{s_1^\#\in S^\#, \clk(s^\#)=\clk(s_1^\#)+1} \rP_A(s_1^\#)\times T_{k,A}^\#(s_1^\#,s^\#)\times (\type(s_1^\#,s^\#)=t(\clk(s_1^\#)))$, if $\clk(s^\#)>0$.
\end{itemize}

The observation function $\obs{A}^\#$ for agent $A$ is defined as follows: 
\begin{itemize}
\item $\obs{A}^\#(s_1^\#)=\obs{A}^\#(s_2^\#)$ if $\obs{A}( L_{\orig}(s_1^\#))=\obs{A}( L_{\orig}(s_2^\#))$, $\clk(s_1^\#)=\clk(s_2^\#)$, and for all $0\leq k\leq \clk(s_1^\#)$, $ L^k_A(s_1^\#)= L^k_A(s_2^\#)$. 
\end{itemize}
Based on this, $\Obs{A}^\#$ contains all possible observations of $A$. 
Moreover, we let $\obs{A}^\#(s^\#)$ to be the set of states that are indistinguishable to agent $A$ at state $s^\#$.
We define a belief function $\be_A^\#:S^\#\rightarrow [0,1]$ as follows for agent $A\in\Ags$:
$$
\be_A^\#(s^\#) = \frac{\rP_A(s^\#)}{\sum_{s_1^\#\in \obs{A}^\#(s^\#)}\rP_A(s_1^\#)}
$$
For every path $\rho\in \fpath{}^\cM$ in the original system $\cM$, we construct the following state $e^\#(\rho)$ in the expanded system $\cM^\#$:
$$
e^\#(\rho)=(last(\rho),\obs{}(\rho(0)),...,\obs{}(\rho(|\rho|-1)),\bot,...,\bot)
$$
We have the following proposition to reduce the probability of $\probm{A}$ in $\cM$ into $\rP_A$ in $\cM^\#$.
\begin{proposition}
For any path $\rho$ of $\cM$ such that $|\rho|\leq \depth(\phi)+1$, we have that $\probm{A}(\cylinder_\rho)=\rP_A(e^\#(\rho))$, where $\cylinder_\rho$ is the basic cylinder with prefix $\rho$.
\end{proposition}
With this proposition, it is not hard to see that $\be_A(\rho)=\be_A^\#(e^\#(\rho))$.
In the following, we inductively define a satisfiability relation $\cM^\#,s^\#\models \phi$ as follows for $\phi$ a boolean combination of atomic propositions and belief formulas of the form $\B_A^{\bowtie q}\varphi$ or $\DT_{A,B}^{\bowtie q}\varphi$ such that $\varphi$ is an atemporal formula.
\begin{itemize}
\item $\cM^\#,s^\#\models p$ for $p\in\ap$ if $p\in L^\#(s^\#)$
\item $\cM^\#,s^\#\models \neg \phi$ if not $\cM^\#,s^\#\models  \phi$
\item $\cM^\#,s^\#\models \phi_1\lor \phi_2$ if  $\cM^\#,s^\#\models  \phi_1$ or $\cM^\#,s^\#\models  \phi_2$
\item $\cM^\#,s^\#\models \B_A^{\bowtie q}\psi$ if
$$
\left (\sum_{\obs{A}^\#(s_1^\#) = \obs{A}^\#(s^\#)} (\cM^\#,s_1^\#\models \psi) \times \be_A^\#(s_1^\#) \right )\bowtie q
$$
\item $\cM^\#,s^\#\models \DT_{A,B}^{\bowtie q}\psi$ if
$$
\left (\sum_{\obs{A}^\#(s_1^\#) = \obs{A}^\#(s^\#)} \be_A^\#(s_1^\#)\times \bigotimes_{s_2^\#\in \intn{B}(s_1^\#)}(T_{k,A}^\#(s_1^\#,s_2^\#)\times \cM^\#,s_2^\#\models \psi) \right )\bowtie q
$$
where $\otimes\equiv \inf$ if $\bowtie\in \{\geq, >\}$, $\otimes\equiv \sup$ if $\bowtie\in \{\leq, <\}$, and the set $ \intn{B}(s_1^\#)$ contains states $s_2^\#$ such that $s_1^\#s_2^\#$ is on some path of $\cM^\#$ and $\type(s_1^\#,s_2^\#)=B.i$.

\item The case of $\CT_{A,B}^{\bowtie q}\psi$ can be done similarly as $\DT_{A,B}^{\bowtie q}\psi$.
\end{itemize}

Now we have the following theorem to reduce the problem of model checking on the original system $\cM$ to the model checking on the expanded system $\cM^\#$. \begin{theorem}\label{thm:expand}
$\cM,\rho\models \phi$ if and only if $\cM^\#,e^\#(\rho)\models \phi$, for $\phi$ a formula in the language $\Lang_A(\PRTLSf)$.
\end{theorem}

Now we can interpret those belief formulas properly with the above theorem, i.e., for each $x\subseteq \Goal{A}$ and $1\leq k\leq \depth(\phi)$, we compute the set
$$
\synth{}^\#(x,k)=\{s^\#\in S^\#~|~\cM^\#,s^\#\models \goalGd{A}(x), \clk(s^\#) \leq k\}
$$
of states in the expanded system $\cM^\#$. It is noted that 
$\synth{}(x,k)=\{\rho~|~|\rho| =k, \exists s^\#\in \synth{}^\#(x,k): s^\#=e^\#(\rho)\}$. But we do not need to compute $\synth{}(x,k)$ because the following procedure will base on $\cM^\#$ instead of $\cM$.


\subsubsection{Preference Function Update} This can be done by following the expressions in Section~\ref{sec:proupdate}.  By the system $\cM^\#$ and the functions $\synth{}^\#(x,k)$, we can have the updated preference functions $\pref{A}$ for all $A\in\Ags$. Then we can update the transition functions $T_{k,A}^\#(s_1^\#,s_2^\#)$ into $T_{k,A}^{\#,*}(s_1^\#,s_2^\#)$ by substituting $\pref{A}$ with the new one. Based on these, we can update the reachability function $\rP_A$ into $\rP_A^*$ and the belief function $\be_A^\#$ into $\be_A^{\#,*}$.

\subsubsection{Model Checking Specification Formula}

The model checking algorithm proceeds by labelling  subformulas of $\phi$ on the states in $S^\#$. We use the notation $sat(s^\#,\psi)$ to represent that a formula $\psi$ is labeled on the state $s^\#$. The labelling can be done recursively as follows.
\begin{itemize}
\item $sat(s^\#,p)$ for $p\in\ap$, if $p\in L^\#(s^\#)$,
\item $sat(s^\#,\psi_1\lor\psi_2)$ if $sat(s^\#,\psi_1)$ or $sat(s^\#,\psi_2)$,
\item $sat(s^\#,\neg \psi)$ if not $sat(s^\#,\psi)$,
\item $sat(s^\#,\forall \next \psi)$ if $sat(t^\#,\psi)$ for all $t^\#$ such that $T_k^\#(s^\#,a,t^\#)>0$,
\item $sat(s^\#,\prob{\bowtie q}{\next\psi})$ if $prob(s^\#,\psi)\bowtie q$,
where the function  $prob(s^\#,\psi)$ is obtained as follows.
\begin{itemize}
\item
$prob(s^\#,\psi)=\sum_{t^\#\in S^\#}T_k^\#(s^\#,a,t^\#)\times sat(t^\#,\psi)$.
\end{itemize}
Intuitively, $prob(s^\#,\psi)$ is the one-step reachability probability of those states satisfying $\psi$ from a given state $s^\#$.
Recall that in the expression $T_k^\#(s^\#,a,t^\#)$, the action $a$ is determined by state $s^\#$.

\item $sat(s^\#,\B_A^{\bowtie q}\psi)$ if $(\sum_{ \obs{A}^\#(t^\#) = \obs{A}^\#(s^\#)}sat(t^\#,\psi)\times \be_A^{\#,*}(t^\#))\bowtie q$

\item $sat(s^\#,\DT_{A,B}^{\bowtie q}\psi)$ if $(\sum_{\obs{A}^\#(t^\#) = \obs{A}^\#(s^\#)}\be_A^{\#,*}(t^\#) \times \bigotimes_{u^\#\in \intn{B}(t^\#)} (T_{k,A}^{\#,*}(t^\#,u^\#)\times sat(u^\#,\psi))\bowtie q$ where $\otimes\equiv \inf$ if $\bowtie\in \{\geq, >\}$, $\otimes\equiv\sup$ if $\bowtie\in \{\leq, <\}$, and the set $ \intn{B}(t^\#)$ contains states $u^\#$ such that
\begin{itemize}
\item
$L_{\orig}(t^\#) \ctrans{B.i,x} L_{\orig}(u^\#)$
for some $x\in \Intn{B}$ such that $u^\#\in \synth{}^\#(x,\clk(t^\#))$,
\item $ L^{j}(u^\#)= L^j(t^\#)$ for all $0\leq j\leq k$,
\item $ L^{k+1}(t^\#)=...= L^{\depth(\phi)-1}(t^\#)= L^{k+2}(u^\#)=...= L^{\depth(\phi)-1}(u^\#)=\bot$, and
\item $ L^{k+1}(u^\#)=\obs{}( L_{\orig}(u^\#))$.
\end{itemize}

\item the case of $sat(s^\#,\CT_{A,B}^{\bowtie q}\psi)$ can be done similarly as that of  $sat(s^\#,\DT_{A,B}^{\bowtie q}\psi)$.

\item $sat(s^\#,\G_A\psi)$ with $\clk(s^\#)=k\leq \depth(\phi)$ if $sat(t^\#,\psi)$ for all $t^\#$ such that
\begin{itemize}
\item
$L_{\orig}(s^\#) \ctrans{B.g.x} L_{\orig}(t^\#)$
for some $x\subseteq \Goal{B}$ such that $t^\#\in \synth{}^\#(x,\clk(s^\#))$,
\item $ L^{j}(t^\#)= L^j(s^\#)$ for all $0\leq j\leq k$,
\item $ L^{k+1}(s^\#)=...= L^{\depth(\phi)-1}(s^\#)= L^{k+2}(t^\#)=...= L^{\depth(\phi)-1}(t^\#)=\bot$, and
\item $ L^{k+1}(t^\#)=\obs{}( L_{\orig}(t^\#))$.
\end{itemize}
\item $sat(s^\#,\I_A\psi)$ with $\clk(s^\#)=k\leq \depth(\phi)$ follows the similar pattern as that of $sat(s^\#,\G_A\psi)$.
\item $sat(s^\#,\C_A\psi)$ with $\clk(s^\#)=k\leq \depth(\phi)$ if $sat(t^\#,\psi)$ for all $t^\#$ such that
\begin{itemize}
\item $ L_{\orig}(s^\#) \ctrans{B.i.x} L_{\orig}(t^\#)$
such that $x\in \Intn{A}$,
\item $ L^{j}(t^\#)= L^j(s^\#)$ for all $0\leq j\leq k$,
\item $ L^{k+1}(s^\#)=...= L^{\depth(\phi)-1}(s^\#)= L^{k+2}(t^\#)=...= L^{\depth(\phi)-1}(t^\#)=\bot$, and
\item $ L^{k+1}(t^\#)=\obs{}( L_{\orig}(t^\#))$.
\end{itemize}
\end{itemize}

It is noted that the above labelling procedure is done locally for all the belief and trust formulas. We have the following theorem to reduce model checking problem to the above labelling algorithm.
\begin{theorem}
$\cM\models \phi$ if and only if $sat(s^\#,\phi)$ for all $s^\#$ such that $\initdist(s^\#)>0$.
\end{theorem}

\subsubsection{Analysis of Complexity}

For the measurement of the complexity, we use the number of states $|S|$ as the size of the system $\cM$, and the depth $\depth(\phi)$ as the size of the specification formula $\phi$. Moreover, we assume that the sets $\Goal{A}$ and $\Intn{A}$, for all $A\in\Ags$, are polynomial with respect to the size of $\cM$, and the size of the formulas in the guarding mechanism, measured with the number of operators, are also polynomial with respect to the size of $\cM$. As usual, the set of agents $\Ags$ is fixed. Also, because the observation function $\obs{}$ is defined on states, the size of $\Obs{1}\times ...\times \Obs{n}$ is no more than the size of $\states$.

For the complexity of the fragment, we notice that the size of the expanded system $\cM^\#$ is polynomial with respect to the size of $\cM$ but  exponential with respect to the size of the specification formula $\phi$.
\begin{itemize}
\item For the pro-attitude synthesis, for every $x\in \bigcup_{A\in\Ags}\powerset{\Goal{A}}\cup \Intn{A}$, the determination of whether $s^\#\in\synth{}^\#(x,\depth(\phi))$ can be done in EXPTIME.
\item The update of preference functions $\{\pref{A}\}_{A\in\Ags}$ and the model $\cM^\#$ can be done in EXPTIME.
\item For the model checking of the specification formula, the labelling procedure can be done in EXPTIME.
\end{itemize}
Putting them together, we have that the computational complexity of the bounded fragment of the problem is in EXPTIME. It is noted that, the exponential is with respect to the size of the specification formula. If the formula is fixed, it is in PTIME.

\subsection{Hardness Problem}

The  lower bound can be obtained by a reduction from the satisfiability problem of the quantified Boolean formula (QBF), which is to determine if a QBF formula is satisfiable. A canonical form of a QBF is $$\phi=Q_1x_1...Q_nx_n:(l_{11}\lor l_{12}\lor l_{13})\land ...\land (l_{m1}\lor l_{m2}\lor l_{m3})$$
where $Q_i=\exists$ if $i$ is an odd number and $Q_i=\forall$ otherwise, and $l_{kj}\in \{x_1,...,x_n,\sep \neg x_1,...,\neg x_n\}$ for all $1\leq k\leq m$ and $1\leq j\leq 3$. Without loss of generality, we assume that $n$ is an even number.  Every QBF formula can be transformed into this format by inserting dummy variables. For a literal $l_{kj}$, we write $\mathit{var}(l_{kj})$ for its variable and $\mathit{sign}(l_{kj})$ for its sign. For example, if $l_{kj}=\neg x_1$ then $\mathit{var}(l_{kj})=x_1$ and $\mathit{sign}(l_{kj})=\mathit{neg}$, and  if $l_{kj}=x_n$ then $\mathit{var}(l_{kj})=x_n$ and $\mathit{sign}(l_{kj})=\mathit{pos}$.

The problem can be simulated by a game between two agents $\forall$ and $\exists$. The game consists of two phases: in the first phase, agents decide the values of variables $\{x_1,...,x_n\}$, and  in the second phase, agents evaluate the boolean formula  $(l_{11}\lor l_{12}\lor l_{13})\land ...\land (l_{m1}\lor l_{m2}\lor l_{m3})$. We write $\mathit{odd}(k)$ ($\mathit{even}(k)$) to denote that the number $k$ is odd (even). Figure~\ref{fig:QBF} gives a diagram for the construction.  The system has a state space $S=\{x_{kb}, x_{kb}^{ha}~|~1\leq h\leq n, 1\leq k\leq n+1, a,b\in \{0,1\}\}\cup \{c_k, c_k^{ha}~|~1\leq k\leq m, 1\leq h\leq n, a\in \{0,1\}\}\cup \{l_{kj}, l_{kj}^{ha}~|~1\leq k\leq m, j\in\{1,2,3\}, 1\leq h\leq n, a\in \{0,1\}\}\cup \{v_{ha}~|~1\leq h\leq n, a\in \{0,1\}\}$. 

Before proceeding to the reduction, we recall notations $\is_\exists(s)$ and $\is_\forall(s)$ which return the intention at state $s$ for agents $\exists$ and $\forall$, respectively. Assuming that agents' intentions are encoded in the states make it easier to understand the reduction. Similar for the notations $\gs_\exists(s)$ and $\gs_\forall(s)$. 

In the first phase, the agents $\exists$ and $\forall$ make their decisions by changing their intentions. The states involved in the first phase are $\{x_{kb}, x_{kb}^{ha}~|~1\leq h\leq n, 1\leq k\leq n+1, a,b\in \{0,1\}\}$. Intuitively, the state $x_{kb}$ represents that the value of variable $x_k$ was set to $b$, and the superscript $ha$ in the state $x_{kb}^{ha}$ represents that the state ``remembers'' the value $a$ of variable $x_h$.  We let $\Intn{\exists} = \{0,1,2,3\}$, where only $0$ and $1$ will be used in the first phase, and let $\Intn{\forall}=\{0,1\}$. On a state $x_{kb}$ or $x_{kb}^{ha}$, we define $\is_\exists(x_{kb})=\is_\exists(x_{kb}^{ha})= b$ for the case of $\mathit{odd}(k)$, and $\is_\forall(x_{kb})=\is_\forall(x_{kb}^{ha})= b$ for the case of $\mathit{even}(k)$. We do not define $\is_\exists(x_{kb})$ and $\is_\exists(x_{kb}^{ha})$ for the cases of $\mathit{even}(k)$ and $\is_\forall(x_{kb})$ and $\is_\forall(x_{kb}^{ha})$ for the case of $\mathit{odd}(k)$, since they will not be used in this proof. Intuitively, the intentional attitude reflects agent's choice. Because the value of variable $x_k$ for $\mathit{odd}(k)$ will be decided by agent $\exists$, we define $\Intn{\exists}(x_{kb})=\Intn{\exists}(x_{kb}^{ha})=\{0,1\}$, so that agent $\exists$ can choose the value by changing its intention. Similar for $\Intn{\forall}(x_{kb})=\Intn{\forall}(x_{kb}^{ha})=\{0,1\}$ when $\mathit{even}(k)$.

Since we will test the beliefs of agent $\exists$ in the specification formula, only agent $\exists$'s observations are relevant. We let $\obs{\exists}(x_{kb})=\obs{\exists}(x_{kb}^{ha})=b$ for any $1\leq k\leq n+1$, i.e., it can see both agents' intentions. However, it is not able to distinguish observations between $x_{kb}$ and $x_{kb}^{ha}$.

The idea of the first phase is that the states of the form $x_{kb}$ contribute to the main computation in which agents use their intentions to decide the values of the variables under their controls. Every time a decision is made, the execution moves to both the next variable $x_{(k+1)b}$ and a helper computation in which all states are labelled with the fact about the last decision. E.g., in Figure~\ref{fig:QBF}, for the first variable $x_1$, the agent $\exists$ can choose between two states $x_{11}$ and $x_{10}$, at which the variable $x_1$ will take different values $1$ and $0$, respectively. Without loss of generality, if agent $\exists$ chooses $x_{11}$, then the execution will move into both $x_{21}$ and $x_{21}^{11}$, where the state $x_{21}^{11}$ remembers that the value of variable $x_1$ is $1$. Therefore, from now on, the execution will be possible on both the main computation from $x_{21}$ and the helper computation from $x_{21}^{11}$. But the agent $\exists$ can not distinguish these two executions by its observations.
Generalising the above, from variable $x_1$ to $x_n$, we will need to explore the main computation and the other $n$ helper computations, such that one helper computation is added for each variable.

\begin{figure}
\centering
\includegraphics[width=14cm,height=9cm]{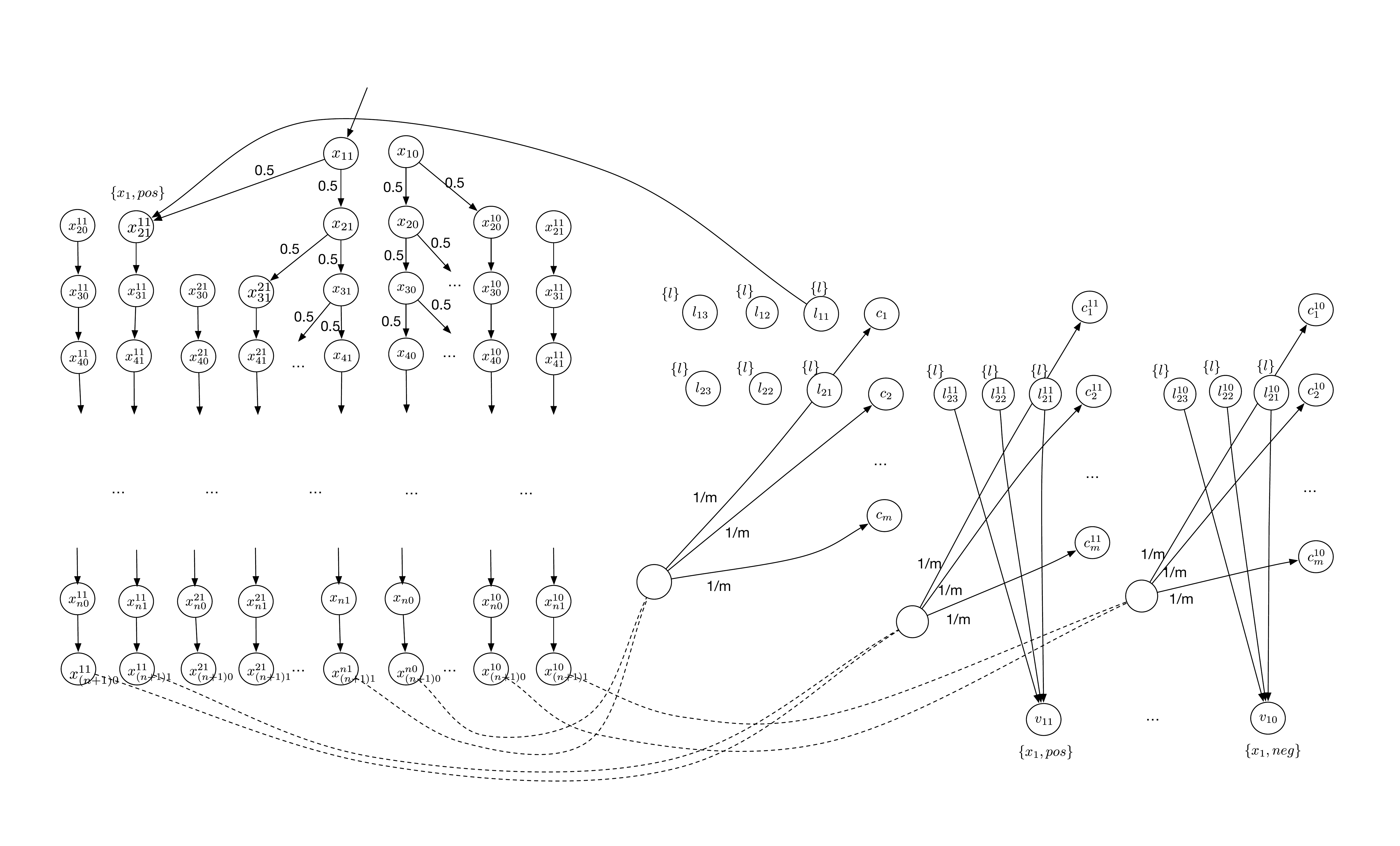}
\caption{Diagram for the reduction from QBF problem}
\label{fig:QBF}
\end{figure}

The execution of the first phase can move into those states $x_{(n+1)b}^{ha}$, where the main computation moves into $x_{(n+1)1}^{n1}$ or $x_{(n+1)0}^{n0}$. By the existence of executions on these states, the system remembers the choices made by the two agents. Then in the second phase, the agent $\exists$ needs to justify that its  choices make the formula $(l_{11}\lor l_{12}\lor l_{13})\land ...\land (l_{m1}\lor l_{m2}\lor l_{m3})$ true whenever the QBF formula is satisfiable. To do this, we let the executions of the first phase move  into those states $c_k$ and $c_k^{ha}$ with a uniform distribution. Each of these states represents intuitively a clause of the QBF formula. Then the agent $\exists$ will need to pick a literal to see if it can make the clause satisfiable.

Similar to the first phase, the choice of literals is done by agent $\exists$ changing its intention. For the cognitive mechanism, we have that $\is_\exists(c_k)=\is_\exists(c_k^{ha})=0$ and $\Intn{\exists}(c_k)=\Intn{\exists}(c_k^{ha})=\{0,1,2,3\}$. That is, agent $\exists$ chooses the literal $l_{kx}$ if its intention is changed into $x$. Note that, we assume that agent $\exists$ can distinguish the clauses, i.e., $\obs{\exists}(c_k)=\obs{\exists}(c_k^{ha})=c_k$. And it is able to know which clause the current literal belongs to, i.e., $\obs{\exists}(l_{kj})=\obs{\exists}(l_{kj}^{ha})=c_k$.

After a literal $l_{kj}$ is chosen, the execution will need to move to obtain its value. This is done by a transition into one of the states in the first phase. Let $\mathit{sign}(l_{kj})$ be the sign of the literal and $\mathit{var}(l_{kj})$ be the value of the literal. Without loss of generality, we let $\mathit{var}(l_{kj})=x_h$ and $\mathit{sign}(l_{kj})=a$, then the transition from state $l_{kj}$ will move into $x_{(h+1)a}^{ha}$, at which the atomic propositions $\{x_{h}, a\}$ are labelled. For the state $l_{kj}^{ha}$, it can only move into the state $v_{ha}$, at which the atomic propositions $\{x_{h}, a\}$ are labelled.

Formally, let $\ap=\{x_i~|~i\in \{1,...,n\}\}\cup \{\literal{l}, \literal{pos}, \literal{neg}\}$ be a set of atomic propositions, we have $\cM(\phi)=(\Ags, S,\sinit,\{\Act{A}\}_{A\in\Ags},
\sep T, \sep  L,\sep  \{\Obs{A}\}_{A\in\Ags}, \sep \{\obs{A}\}_{A\in\Ags}, \sep \{\Cognition_A\}_{A\in\Ags}, \sep \{\cstrat{A}\}_{A\in\Ags}, \sep \{\pref{A}\}_{A\in\Ags})$ such that
\begin{itemize}
\item $\Ags=\{\exists,\forall\}$,
\item $S=\{x_{kb}, x_{kb}^{ha}~|~1\leq h\leq n, 1\leq k\leq n+1, a,b\in \{0,1\}\}\cup \{c_k, c_k^{ha}~|~1\leq k\leq m, 1\leq h\leq n, a\in \{0,1\}\}\cup \{l_{kj}, l_{kj}^{ha}~|~1\leq k\leq m, j\in\{1,2,3\}, 1\leq h\leq n, a\in \{0,1\}\}\cup \{v_{ha}~|~1\leq h\leq n, a\in \{0,1\}\}$,
\item $\sinit=\{x_{11}\}$ with $\initdist(x_{11})=1$,
\item $\Act{A}=\{\epsilon \}$ for $A\in \{\exists,\forall\}$,
\item the transition function is as follows.
\begin{itemize}
\item $T(x_{kb}, x_{(k+1)b})=T(x_{kb}, x_{(k+1)b}^{kb})=0.5$ for $1\leq k\leq n-1$, $T(x_{nb}, x_{(n+1)b})=1$,
\item $T(x_{kb}^{ha}, x_{(k+1)b}^{ha})=1$, for $1\leq k\leq n$,
\item $T(x_{(n+1)b}^{ha},c_k^{ha})=T(x_{(n+1)b},c_k)=1/m$,
\item $T(l_{kj},x_{(h+1)1}^{h1})=1$ for $\mathit{sign}(l_{kj})=\mathit{pos}$ and $\mathit{var}(l_{kj})=v_h$,
\item $T(l_{kj},x_{(h+1)0}^{h0})=1$ for $\mathit{sign}(l_{kj})=\mathit{neg}$ and $\mathit{var}(l_{kj})=v_h$,
\item $T(l_{kj}^{ha},v_{ha})=1$
\end{itemize}
\item the labelling function is as follows.
\begin{itemize}
\item $ L(x_{kj}^{h1})=\{x_h,pos\}$, $ L(x_{kj}^{h0})=\{x_h,neg\}$,
\item $ L(v_{h1})=\{x_h,pos\}$, $ L(v_{h0})=\{x_h,neg\}$,
\item $ L(l_{kj})=l$
\end{itemize}
\item the set of possible observations is as follows. 
\begin{itemize}
\item $\Obs{\exists}=\{0,1\}\cup \{c_k~|~1\leq k\leq m\}\cup \{\emptyset\}$ for player $\exists$, and 
\item we ignore the definition for player $\forall$ as it is not used in the reduction. 
\end{itemize}
\item the partial observation functions are defined as follows for the player $\exists$. 
\begin{itemize}
\item $\obs{\exists}(x_{kb})=\obs{\exists}(x_{kb}^{ha})=b$ for $1\leq k\leq n+1$ and $b\in \{0,1\}$, 
\item $\obs{\exists}(c_k)=\obs{\exists}(c_k^{ha})=c_k$ for $1\leq k\leq m$, 
\item $\obs{\exists}(l_{kj})=\obs{\exists}(l_{kj}^{ha})=c_k$ for $1\leq k\leq m$ and $j \in \{1,2,3\}$, and 
\item $\obs{\exists}(v_{ha})=\emptyset$. 
\end{itemize}
\item the set of intentions is as follows. 
\begin{itemize}
\item $\Intn{\exists}=\{0,1,2,3\}$ and $\Intn{\forall}=\{0,1\}$, and 
\item we ignore the definition for $\Goal{\exists}$ and $\Goal{\forall}$ as they are not used in the reduction. 
\end{itemize}
\item the guarding mechanism is trivial. 
\item  the preference functions are defined as follows.
\begin{itemize}
\item $\ipref{\exists}{\exists}(x_{kb})(0)=\ipref{\exists}{\exists}(x_{kb})(1)=0.5$ for $\mathit{odd}(k)$,
\item $\ipref{\exists}{\forall}(x_{kb})(0)=\ipref{\exists}{\forall}(x_{kb})(1)=0.5$ for $\mathit{even}(k)$,
\item $\ipref{\exists}{\exists}t(c_k)(0)=\ipref{\exists}{\exists}(c_k)(1)=\ipref{\exists}{\exists}(c_k)(2)=\ipref{\exists}{\exists}(c_k)(3)=0.25$ for $1\leq k\leq m$.
\end{itemize}
\end{itemize}

In such a system, we have that $\phi$ is satisfiable if and only if
$$
\cM(\phi)\models \overline{\I_\exists} AX ~\I_\forall AX ... \I_\forall AX ~\overline{\I_\exists} (l\land \bigwedge_{i=1}^n ((\B_\exists^{\geq 1} AX (x_i\rimp pos)) \lor (\B_\exists^{\geq 1} AX (x_i\rimp neg)))).
\commentout{\footnote{In this formula, we write $AX$ for the operator $\forall \next$, to avoid the confusion with the agents $\forall$ and $\exists$. }}
$$
Note that in above formula and below, we use $AX$ instead of $\forall \next$ to avoid confusion with the agents $\forall$ and $\exists$. Intuitively, the prefixes $\overline{\I_\exists} AX$ and $\I_\forall AX$ represent the choosing of variables' values by the agents $\exists$ and $\forall$, respectively. After the last $AX$, the execution moves to the states $c_k$ and $c_k^{ha}$. Then, the last $\overline{\I_\exists}$ is for the agent $\exists$ to choose a specific literal, and then justify that it is satisfiable. The former is represented by the presence of a literal $l$ in the formula. The latter is guaranteed by $\bigwedge_{i=1}^n ((\B_\exists^{\geq 1} (x_i\rimp pos)) \lor (\B_\exists^{\geq 1} (x_i\rimp neg)))$. Without loss of generality, we let $l_{ha}$ be the chosen literal and assume $\mathit{var}(l_{ha})= x_i$. 
If $\mathit{sign}(x_i) = \mathit{neg}$ then the main computation will reach the state $x_{(h+1)0}^{h0}$ which is labelled with $x_h$ and $\mathit{neg}$. By construction, we know that if the variable $x_i$ is assigned the value $\mathit{neg}$ by one of the agents then the helper computation with superscript $i0$ is explored, and the helper computation will lead to the state $v_{i0}$ at which the propositions $x_i$ and $\mathit{neg}$ are labelled. On the other hand, the helper computation with superscript $i1$ will not be explored, which means that the state $v_{i1}$ labelled with $x_i$ and $\mathit{pos}$ is not considered possible by the agent $\exists$. Therefore, we have that $\B_\exists^{\geq 1} (x_i\rimp neg)$. Similar argument can be taken to show that $\B_\exists^{\geq 1} (x_i\rimp pos)$ when $\mathit{sign}(x_i) = \mathit{pos}$.


\section{A Decidable Fragment with $\until$ and $\always$ temporal operators}\label{sec:mc}

We present a fragment of the problem that is PSPACE-complete. Unlike the bounded fragment in Section~\ref{sec:bounded}, this fragment works with the $\until$ and $\always$ temporal operators in specification formula $\phi$. With $\until$ and $\always$ operators, the specification formula can express long-term properties about agent's mental attitudes in the system, and therefore this fragment complements the bounded fragment. However, the fragment is subject to other restrictions, including the following:

\begin{itemize}
\item The guarding mechanism is trivial. That is, the algorithm works with model checking without conducting pro-attitude synthesis first.

\item The specification formula $\phi$ has the restrictions that
\begin{itemize}
\item it works with a single agent's beliefs and trusts,
\item there is no nested beliefs or trusts,
\item beliefs and trusts cannot be in the scope of a probabilistic operator $\prob{}{}$, and
\item there is a constant number of belief or trust operators.
\end{itemize}
We make a remark on these constraints. For the first constraint, that we choose to work with a single agent's beliefs and trusts is because in Section~\ref{sec:general} we have shown that the single-agent fragment is undecidable. The second constraint is imposed to keep the algorithm simple. We conjecture that it can be relaxed. The third and fourth constraints are to keep the complexity in PSPACE. For the probabilistic operator, it requires to compute a set of paths concurrently. We conjecture that the complexity can be in EXPTIME instead of in PSPACE if relaxing this constraint. To complement this, in Section~\ref{sec:ptime}, we study a fragment in which the probabilistic operator works with qualitative beliefs and trusts. The constraint on the number of belief or trust operators is to keep the expanded system, to be given below, in the size of a polynomial with respect to the formula.

\item There is a restriction on the system about agent $A$'s beliefs along the transitions. The restriction, whose details will be given later, is that along any path of an expanded system, the evaluations of belief or trust subformulas
are the same on any two expanded states if the support of a probabilistic distribution over the set of states is the same.
\end{itemize}
Without loss of generality, we assume that it is the agent $A$'s beliefs that the specification formula $\phi$ is to describe. Note that the fragment has the full expressiveness of LTL.

\subsection{Upper Bound}

Because the guarding mechanism is trivial, we need only handle the third step of the verification procedure. 

\paragraph{Preprocessing Sub-Formulas in the Scope of Belief Operators}

The first step of the algorithm is to pre-process subformulas in the scope of belief or trust operators. Because there is no nested beliefs, these subformulas do not contain any belief or trust operator. Therefore, they can be handled by the usual probabilistic temporal model checking techniques, with some slight adaptations for the formulas $\G_A\psi$, $\I_A\psi$ and $\C_A\psi$. For every such subformula $\psi$, we introduce a new atomic proposition $p_\psi$, which intuitively represents the probability of $\psi$ being satisfied. Let $\ap'=\ap \bigcup \ap_1$ such that $\ap_1$ is the set of new atomic propositions.
We upgrade the labelling function $ L$ into $ L': S\times \ap' \rightarrow [0,1]$, which allows us to extract the probability value ''stored`` in $p_\psi$. 
\begin{itemize}
\item If $p\in\ap$ then $ L'(s,p)=1$, if $p\in  L(s)$, and $ L'(s,p)=0$, otherwise.
\item If $p_\psi\in\ap_1$ then $ L'(s,p_\psi)=Prob(\cM,s,\psi)$, where $Prob(\cM,s,\psi)$ is the probability of satisfying formula $\psi$ among all the temporal paths starting from $s$. We note that, because of Assumption~\ref{assump:detbeh}, there exists a unique value for $Prob(\cM,s,\psi)$.  
\end{itemize}

We assume that $L'$ is computed from $L$ as part of the algorithm and therefore no change to our previous definitions (such as Definition~\ref{def:smg}) is necessary.
Note that, after this pre-processing, all belief or trust formulas are of the form $\B_A^{\bowtie q}p$ or $\DT_{A,B}^{\bowtie q}p$ such that $p\in\ap$ is an atomic proposition. This pre-processing can also be applied to subformulas containing no belief or trust formulas. Therefore, in the following, we assume that in $\phi$, all subformulas contain at least one belief or trust formulas.

We consider the negated specification formula $\neg\phi$, and verify the existence of a counterexample. We assume that the formula $\neg\phi$ is in negation normal form, i.e., negation operators only appear  in  front of atomic propositions.

\paragraph{Construction of Expanded System}

\newcommand{\labeling}{{\tt la}}
\newcommand{\trend}{{\tt e}}

Let $\phi$ be the specification formula and $b_\phi=\{\psi_1,...,\psi_k\}$ be the set of belief or trust subformulas of the form $\B_A^{\bowtie q}p$ or $\DT_{A,B}^{\bowtie q}p$ such that $p_1,...,p_k\in\ap$ are their atomic propositions. We need some notations.
Given a set $Q\subseteq S$ of states, a probability space  on $Q$ is a triple $\alpha_Q=(S,\powerset{\powerset{S}},\probm{Q})$ such that $\probm{Q}(\{s\})>0$ if and only if $s\in Q$.  Given a probability space $\alpha_Q$, we can evaluate belief and trust formulas as follows. 
\begin{itemize}
\item $\probm{Q}(\B_A^{\bowtie q}p) =\sum_{s\in Q}\probm{Q}(s)\times  L'(s,p)$.
\item $\probm{Q}(\DT_{A,B}^{\bowtie q}p) =\sum_{s\in Q}\probm{Q}(s)\times \bigotimes_{x\in \Intn{B}(s)} \exists s': L'(s',p)\times (s\ctrans{B.i.x} s')$, where $\otimes\equiv \inf$ if $\bowtie\in \{\geq, >\}$, $\otimes\equiv\sup$ if $\bowtie\in \{\leq, <\}$.
\item the case of $\CT_{A,B}^{\bowtie q}p$ is analogous to $\DT_{A,B}^{\bowtie q}p$.
\end{itemize}
as the probability of those states satisfying formula $\psi$. In the following, we use the probability measure $\probm{Q}$ to represent the probability space $\alpha_Q$.
Given a probability distribution $\probm{P}$ and a set of states $Q$, we define (for a set of states $P$):
\[
\probm{P,\Act{}, Q}(t)=\frac{ \sum_{s\in P}(\probm{P}(s)\times T(s,a,t))}{\sum_{t\in Q}\sum_{s\in P}(\probm{P}(s)\times T(s,a,t))},
\]
as the probability measure over $Q$, by a temporal transition from $P$ in one step. Moreover, on the pro-attitude dimensions, we define:
$$\probm{P,\Goal{B}(s), Q}(t)=\frac{ \sum_{s\in P}\sum_{x\in\Goal{B}(s)(s)}(\probm{P}(s)\times \gpref{A}{B}(s,x)\times (s\ctrans{B.g.x} t))}{\sum_{t\in Q} \sum_{s\in P}\sum_{x\in\Goal{B}(s)(s)}(\probm{P}(s)\times \gpref{A}{B}(s,x)\times (s\ctrans{B.g.x} t))}$$
and
$$\probm{P, \Intn{B}, Q}(t)=\frac{ \sum_{s\in P}\sum_{x\in \Intn{B}(s)}(\probm{P}(s)\times \ipref{A}{B}(s,x)\times (s\ctrans{B.i.x} t))}{\sum_{t\in Q} \sum_{s\in P}\sum_{x\in \Intn{B}(s)}(\probm{P}(s)\times \ipref{A}{B}(s,x)\times (s\ctrans{B.i.x} t))}$$
for $B\in\Ags$. Intuitively, they are the probability measures over $Q$ by a goal or intentional transition of agent $B$ from the states in $P$. Moreover, the transition is parameterised over the agent $A$'s preference functions.
Now we define the initial belief for agent $A$. Given a state $s$ such that $\initdist(s)>0$, we define $\obs{A}(\initdist,s)=\{t~|~\obs{A}(t)=\obs{A}(s), \initdist(t)>0\}$ to be the set of possible initial states whose observation is the same as $s$. Then, we have
$$
\probm{\obs{A}(\initdist,s)}(t)=\frac{\initdist(t)}{\sum_{t\in \obs{A}(\initdist,s)}\initdist(t)}
$$
as the initial belief.

Based on the above notations, we can construct an expanded system which maintains a probability value for each subformula in $b_\phi$. The system is $\cM_i(\phi)=(\Ags, S^\dagger,\sinit^\dagger,\{\Act{A}\}_{A\in\Ags},
T^\dagger, L^\dagger)$ such that
\begin{itemize}
\item $S^\dagger=S\times \dist{S}\times R \times ... \times R$, where $R$ is the domain of real numbers. Note that, although $R$ is infinite, we need only a finite number of real numbers. We assume that all numbers are of fixed precision and can be encoded with complexity $O(1)$.
\item $\sinit^\dagger=\{(s,\probm{\obs{A}(\initdist,s)},\probm{\obs{A}(\initdist,s)}(p_1),...,\probm{\obs{A}(\initdist,s)}(p_k))~|~\initdist(s) > 0, \ap_1 = \{p_1,...,p_k\} \}$,
\item the transition relation is defined as follows.
\begin{itemize}
\item $T^\dagger((s,\probm{Q},v_1,...,v_k),(t,\probm{Q,\Act{}, Q'},\probm{Q,\Act{}, Q'}(\psi_1),...,\probm{Q,\Act{}, Q'}(\psi_k)))=T(s,a,t)$ if \\ $\type(s,t)=\Act{}$ and $Q'=\{t'~|~\probm{Q}(s')>0, s'\trans{a'}t', \obs{A}(t')=\obs{A}(t)\}$.
\item $T^\dagger((s,\probm{Q},v_1,...,v_k),(t,\probm{Q,\lgoal{B}(s),Q'},\probm{Q,\lgoal{B}(s),Q'}(\psi_1),...,\probm{Q,\lgoal{B}(s),Q'}(\psi_k)))=\\ \gpref{A}{B}(s,t)$ if and only if $\type(s,t)=B.g$ and  $Q'=\{t'~|~s'\in Q,~ y \in \lgoal{B}(s)(s), s' \ctrans{B.g.y} t', \\\obs{A}(t')=\obs{A}(t) \}$.
\item $T^\dagger((s,\probm{Q},v_1,...,v_k),(t,\probm{Q, \lintn{B},Q'},\probm{Q, \lintn{B},Q'}(\psi_1),...,\probm{Q, \lintn{B},Q'}(\psi_k)))=\ipref{A}{B}(s,t)$ if and only if $\type(s,t)=B.i$ and  $Q'=\{t'~|~s'\in Q,~ y \in \lintn{B}(s), s'\ctrans{B.i.y} t', \obs{A}(t')=\obs{A}(t) \}$.

\end{itemize}
\item $ L^\dagger((s,\probm{Q},v_1,...,v_k))= L(s)$.
\end{itemize}

\paragraph{Constraints on the Expanded System}

The constraints include,
\begin{itemize}
\item there do not exist two states $(s,\probm{Q},v_1,...,v_k)$ and $(s,\probm{Q'},v_1',...,v_k')$ in the expanded system $\cM^\dagger$ such that they are reachable from one to the other, $supp(\probm{Q})=supp(\probm{Q'})$, and
$v_j\neq v_j'$ for some $1\leq j\leq k$, and
\item  for any two reachable states $(s,\probm{Q},v_1,...,v_k)$ and $(s,\probm{Q'},v_1',...,v_k')$ such that $supp(\probm{Q})=supp(\probm{Q'})$, only one of them needs to be explored with $T^\dagger$.
\end{itemize}
With the above constraints, we can see that the size of the expanded system is $O(2^{|\cM|})$. However, we do not need to explicitly construct it. The algorithm explores it on-the-fly.


\paragraph{Model Checking Algorithm}

\newcommand{\automaton}{{\cal A}}

For the algorithm, we use function $sat(s,\probm{Q},v_1,...,v_k,\phi)$ to denote the satisfiability of the formula $\phi$ on the state $(s,\probm{Q},v_1,...,v_k)$. For the base cases of non-temporal operators, we have the following inductive rules:

\begin{algorithm}\label{alg:alg1}
Given a state $(s,\probm{Q},v_1,...,v_k)\in S^\dagger$ and a formula $\phi$, we define the satisfiability relation $sat(s,\probm{Q},v_1,...,v_k,\phi)$ as follows:
\begin{itemize}
\item $sat(s,\probm{Q},v_1,...,v_k,p)$ for $p\in\ap$  if $p\in L(s)$;
\item $sat(s,\probm{Q},v_1,...,v_k,\neg\phi)$ if not $sat(s,\probm{Q},v_1,...,v_k,\phi)$;
\item $sat(s,\probm{Q},v_1,...,v_k,\phi_1\lor \phi_2)$ if  $sat(s,\probm{Q},\sep v_1,...,v_k,\phi_1)$ or $sat(s,\probm{Q},\sep v_1,...,\sep v_k,\sep \phi_2)$;
\item $sat(s,\probm{Q},v_1,...,v_k,\G_j\phi)$ if $sat(t,\probm{Q'},v_1',...,v_k',\phi)$ for all states $(t,\probm{Q'},\sep v_1',...,\sep v_k')$ such that $T^\dagger((s,\probm{Q},v_1,...,v_k),(t,\probm{Q'},v_1',...,v_k')) > 0 $ and $\type(s,t)=B.g$.

\item $sat(s,\probm{Q},v_1,...,v_k,\I_j\phi)$ if $sat(t,\probm{Q'},v_1',...,v_k',\phi)$ for all states $(t,\probm{Q'},\sep v_1',...,\sep v_k')$ such that $T^\dagger((s,\probm{Q},v_1,...,v_k),(t,\probm{Q'},v_1',...,v_k')) > 0 $ and $\type(s,t)=B.i$.

\item $sat(s,\probm{Q},v_1,...,v_k,\C_j\phi)$ if $sat(t,\probm{Q'},v_1',...,v_k',\phi)$ for all states $(t,\probm{Q'},\sep 0,...,0)$ such that $s\ctrans{B.i.x} t$ for some $x\in \lintn{B}(s)$, and $\probm{Q'}$ is any distribution such that $Q'=\{t'~|~ s'\in Q, s'\ctrans{B.i.x} t', x\in \lintn{B}(s), \obs{A}(t')=\obs{A}(t) \}$. Note that, as we stated earlier, no belief formula will appear in the scope of a \legal\ intention operator. So the probability values are not useful.

\item For the computation of $sat(s,\probm{Q},v_1,...,v_k,\prob{\bowtie q}{\next \phi})$ and $sat(s,\probm{Q},v_1,...,v_k,\sep\prob{\bowtie q}{\phi \until \phi })$, we note that, because the constraint that no belief formula can be in the scope of a probability operator, the components $\probm{Q}$ and $v_1,...,v_k$ in the state are irrelevant. Therefore, we can compute $sat(s,\prob{\bowtie q}{\next \phi})$ and $sat(s,\prob{\bowtie q}{\phi \until \phi })$ as the usual PCTL model checking.

\item $sat(s,\probm{Q},v_1,...,v_k,\B_A^{\bowtie q}p)$ if $v_j \bowtie q$, for $\B_A^{\bowtie q}p$ being the $j$th formula in the set $b_\phi$.

\item $sat(s,\probm{Q},v_1,...,v_k,\DT_{A,B}^{\bowtie q}p)$ if $v_u \bowtie q$, for $\DT_{A,B}^{\bowtie q}p$ being the $u$th formula in the set $b_\phi$.

\end{itemize}
\end{algorithm}
Note that the above cases have the same results for those distributions $\probm{Q}$ whose supports are the same.
Now we discuss the case of $\exists\psi$. With the above Algorithm~\ref{alg:alg1}, we note that other formulas can be interpreted directly on the states. Therefore, we can introduce atomic propositions $p_\psi$ for those $\psi$, and update the labelling function into $p_\psi\in  L^\dagger((s,\probm{Q},v_1,...,v_k))$ if and only if $sat(s,\probm{Q},v_1,...,v_k,\psi)$. With this processing, the formula $\psi$ is turned into an LTL formula $\psi'$.
Then as the usual LTL model checking, we turn the formula $\psi$ into a Buchi automaton $\automaton_\psi$. Let $\automaton_\psi=(Q,\powerset{\ap}, \delta, B, F)$ be the automaton such that $Q$ is a set of states, $\delta: Q\times \powerset{\ap} \rightarrow \powerset{Q}$ is a transition relation, $B\subseteq Q$ is a set of initial states, and $F\subseteq Q$ is a set of sets of acceptance states. Then we construct their product system $\cM^\dagger\times \automaton_{\psi'}=(S^\ddagger, \sinit^\ddagger, T^\ddagger,  L^\ddagger)$ as follows.
\begin{itemize}
\item $S^\ddagger=S\times \dist{S}\times R \times ... \times R \times Q$
\item $(s,\probm{Q}, v_1,...,v_k, q)\in \sinit^\ddagger$ if $(s,\probm{Q}, v_1,...,v_k)\in \sinit^\dagger$ and $(q_0, L^\dagger((s,\probm{Q}, \sep v_1,...,v_k)), \sep q)\in \delta$ such that $q_0\in B$.
\item $((s,\probm{Q}, v_1,...,v_k,q),(t,\probm{Q'}, v_1',...,v_k',q'))\in T^\ddagger$ if we have that \\ $T^\dagger((s,\probm{Q}, v_1,...,v_k),(t,\probm{Q'}, v_1',...,v_k'))>0$, $\type(s,t)=\Act{}$, and \\ $(q, L^\dagger((t,\probm{Q'}, v_1',...,v_k')),q')\in\delta$.
\item $ L^\ddagger(s,\probm{Q}, v_1,...,v_k,q)= L^\dagger((s,\probm{Q}, v_1,...,v_k))$.
\end{itemize}

\begin{algorithm}\label{alg:alg2}
For the case where the formula is $\exists\psi$, we have that
\begin{itemize}
\item $sat(s,\probm{Q},v_1,...,v_k,\exists\psi)$ if $\cM^\dagger[\{s,\probm{Q},v_1,...,v_k\}/\sinit^\dagger] \times \automaton_{\psi}$ is empty;
\end{itemize}
where $\cM^\dagger[\{s,\probm{Q},v_1,...,v_k\}/\sinit^\dagger]$ is the system of $\cM^\dagger$ by substituting the initial distribution $\sinit^\dagger$ into $\sinit^\dagger(\{s,\probm{Q},v_1,...,v_k\})=1$.
\end{algorithm}

Then for a formula with nested branching time operators, we can use the approach of CTL$^*$ model checking.
Finally, we have the following theorem to state the correctness of the above algorithm.
\begin{theorem}
$\cM\models \phi$ iff, in $\cM^\dagger$, $sat(s,\probm{Q}, v_1,...,v_k,\phi)$ holds for all $(s,\probm{Q}, \sep v_1,...,\sep v_k)\in \sinit^\dagger$.
\end{theorem}

\paragraph{Analysis of Complexity}

For the complexity, we need to remember the tuple $(s,\probm{Q}, v_1,...,v_k)$ for Algorithm~\ref{alg:alg1} or the tuple $(s,\probm{Q}, v_1,...,v_k,q)$ for Algorithm~\ref{alg:alg2}. For Algorithm~\ref{alg:alg1}, we may need a polynomial number of alternations to handle the negation operators or the pro-attitudes. Therefore, by Theorem 4.2 of~\cite{CKS1980}, it is in NPSPACE=PSPACE, the same complexity as that of LTL model checking. On the other hand, unlike the LTL model checking, the program complexity, i.e., the complexity measured with respect to the size of the system by assuming that the formula is fixed, is also in PSPACE.

\subsection{Hardness Problem}

The PSPACE-hardness of the problem  can be obtained by a reduction from the problem of deciding if, for a given nondeterministic finite state automaton $A$ over an alphabet $\Sigma$, the language $\Lang(A)$ is  equivalent to the universal language $\Sigma^*$.
Let $A=(Q,q_0,\delta,F)$ be an NFA such that $Q$ is a set of states, $q_0\in Q$ is an initial state, $\delta: Q\times \Sigma\rightarrow \powerset{Q}$ is a transition function, and $F\subseteq Q$ is a set of final states.

We construct a system $\cM(A)$ whose states, except for the initial state $s_0$, can be classified into two subsystems $\cM_1(A)$ and $\cM_2(A)$. Intuitively, the subsystem $\cM_2(A)$ simulates the execution of $A$ and upon reading each input symbol, probabilistically checks whether the current state of the automaton is a final state, while the subsystem $\cM_1(A)$ simulates the universal language $\Sigma^*$, at the same time synchronising with $\cM_2(A)$ on words which produce a valid run of $A$.

In particular, each state of $\cM_2(A)$ is a pair of a symbol of the alphabet and an annotated automaton state, where the latter can be of the form $q$, $q^u$, $q^n$, $q^t$, or $q^{tf}$ for $q \in Q$. The state $(l,q)$ is a regular state corresponding to those of automaton $A$, representing that the state $q$ is reached from another state by reading a symbol $l$. However, from a state $(l_1,q_1)$ to another state $(l_x,q_x)$ such that $q_x\in \delta(q_1,l_x)$, we need to go through two levels of intermediate states.  In the first level, the  state $(l_x,q_x^u)$ represents a decision point from which the execution may or may not conduct a test on the current state. The purpose of the test is to see if the state $q_x$ is a final state.
With a chance of 0.5, the path will proceed to $(l_x,q_x^n)$ which denotes that no test will be conducted.
With another chance of 0.5, the path will conduct the test. If $q_x$ is not a final state (e.g., the state $q_2$ of Figure~\ref{fig:nfa}), then the execution proceeds to the state $(l_x,q_x^t)$, at which the atomic proposition $\literal{test}$ is labelled. If $q_x$ is a final state (e.g., the state $q_3$ of Figure~\ref{fig:nfa}), then the execution proceeds with 0.25 chance to the state $(l_x,q_x^t)$ and another 0.25 chance to the state $(l_x,q_x^{tf})$, at which both the propositions $\literal{test}$ and $\literal{finished}$ are labelled.
From both the states $(l_x,q_x^{t})$ and $(l_x,q_x^{tf})$, the execution can only move to a self-loop.

The subsystem $\cM_1(A)$ consists of states which are pairs of a symbol of the alphabet and one of the elements in $\{1,u,n,t\}$, and its structure mirrors that of $\cM_2(A)$, with the exception of final states (since $\cM_1(A)$ does not know anything about states of automaton $A$). The purpose of $\cM_1(A)$ is to make sure that words which do not produce a valid run of $A$ have a corresponding path in $\cM(A)$.
The intuition of the construction can be seen from the diagram in Figure~\ref{fig:nfa}
The agent $A$ is given distinct observations $t$ on those states $(l,q^t)$ and $(l,q^{tf})$, that is, agent $A$ can distinguish whether a test occurs at a moment. 
Formally,  let $\ap=\{\literal{test},\literal{finished}\}$ be a set of atomic propositions, the system is $\cM(A)=(\Ags, S,\sinit,\{\Act{A}\}_{A\in\Ags},
\sep T, \sep  L, \sep \{\Obs{A}\}_{A\in\Ags}, \sep \{\obs{A}\}_{A\in\Ags}, \sep \{\Cognition_A\}_{A\in\Ags}, \sep \{\cstrat{A}\}_{A\in\Ags}, \sep \{\pref{A}\}_{A\in\Ags})$ such that
\begin{itemize}
\item $\Ags=\{B\}$, i.e., there is only a single agent $B$,
\item $S=\{s_0\}\cup (\Sigma\times (Q\cup Q^t\cup Q^{tf}\cup Q^u\cup Q^n\cup \{1,u,n,t\}))$, where $Q^t=\{q^t~|~q\in Q\}$, $Q^{tf}=\{q^{tf}~|~q\in Q\}$, $Q^u=\{q^u~|~q\in Q\}$ and $Q^n=\{q^n~|~q\in Q\}$,
\item $\sinit = \{s_0\}$ with $\initdist(s_0)=1$,
\item $\Act{B}=\{a_\epsilon\}$, i.e., agent $B$ only has a trivial action $a_\epsilon$
\item the transition function is defined as follows.
\begin{itemize}
\item $T(s_0,a_\epsilon,(l_0,q_0))=T(s_0,a_\epsilon,(l_0,1))=0.5$ for some $l_0\in \Sigma$,
\item $T((l,q),a_\epsilon, (l_2,q_2^u))=1/ k(q)$ if $q_2\in\delta (q,l_2)$, where $k(q)=|\{(l_2,q_2)~|~q_2\in \delta(q,l_2)\}|$,
\item $T((l,q^u),a_\epsilon,(l,q^n))=0.5$,
\item $T((l,q^u),a_\epsilon,(l,q^t))=0.5$, if $q\notin F$,
\item $T((l,q^u),a_\epsilon,(l,q^t))=T((l,q^u),a_\epsilon,(l,q^{tf}))=0.25$, if $q\in F$,
\item $T((l,q^n),a_\epsilon,(l,q))=T((l,q^t),a_\epsilon,(l,q^t))=T((l,q^{tf}),a_\epsilon,(l,q^{tf}))=1$,
\item $T((l,1),a_\epsilon,(l_2,u))=1/|\Sigma|$,
\item $T((l,u),a_\epsilon,(l,n))=T((l,u),a_\epsilon,(l,t))=0.5$, and
\item $T((l,n),a_\epsilon,(l,1))=T((l,t),a_\epsilon,(l,t))=1$.
\end{itemize}
\item $\literal{test} \in  L((l,q^t)) $, $\literal{test} \in  L((l,q^{tf})) $, $\literal{finished} \in  L((l,q^{tf})) $, $\literal{test} \in  L((l,t)) $.
\item 
The observation of the agent $B$ is defined as follows.
\begin{itemize}
\item $\obs{B}((l,q))=\obs{B}((l,q^u))=\obs{B}((l,q^n))=\obs{B}((l,1))=\obs{B}((l,u))=\obs{B}((l,n))=l$, and
\item $\obs{B}((l,q^t))=\obs{B}((l,q^{tf}))=\obs{B}((l,t))=l \land t$.
\end{itemize}
\item other components are omitted as they are not used in this reduction.  
\end{itemize}

\begin{figure}
\centering
\includegraphics[width=10cm,height=7cm]{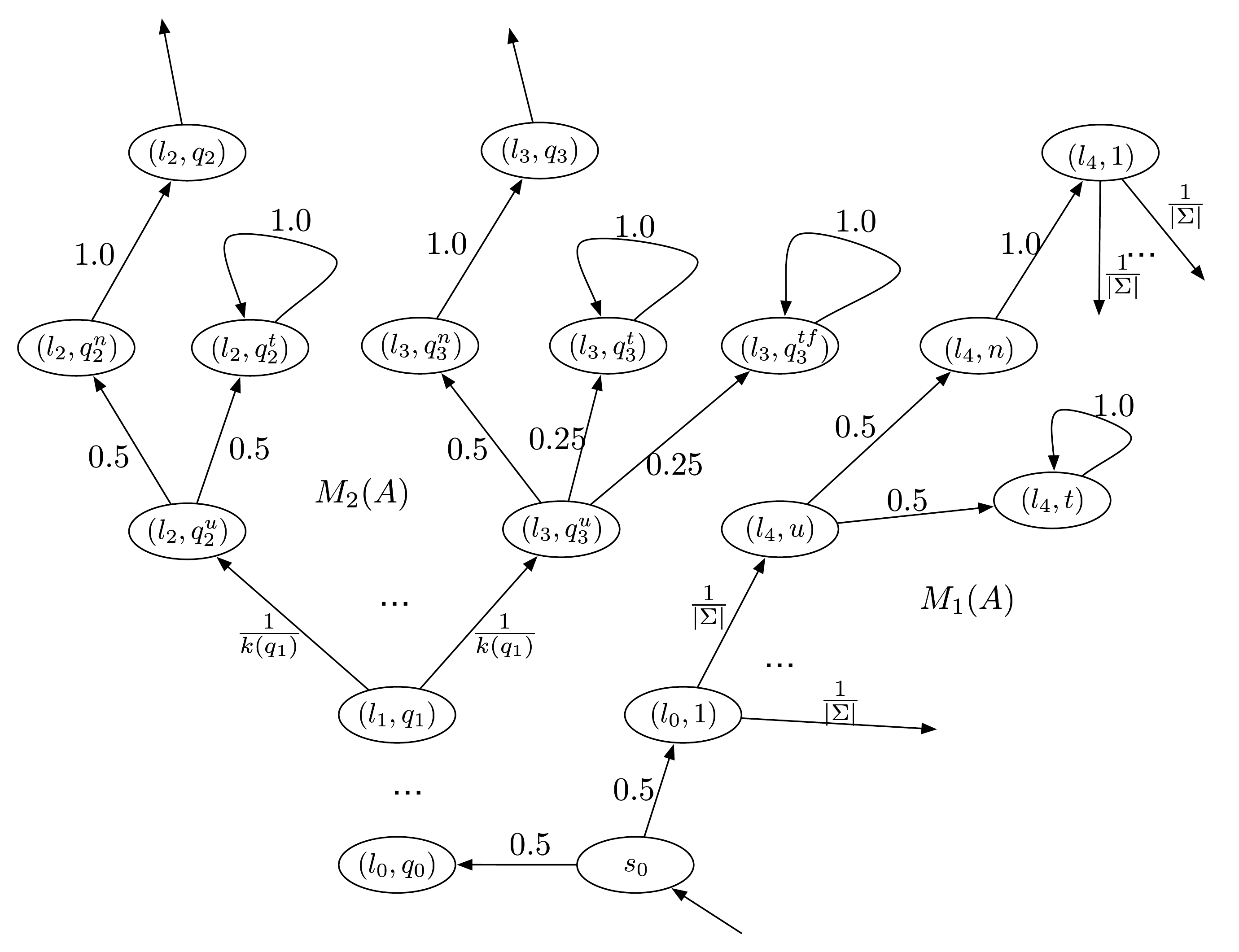}
\caption{Diagram for the reduction from universality problem of NFA}
\label{fig:nfa}
\end{figure}

In such a system, we have that $\Lang(A)=\Sigma^*$ if and only if
$$
\cM(A)\models  \literal{test} \release (\neg \literal{test} \lor \B_A^{<1} (\neg \literal{finished}))
$$
where $\release$ is the release operator of LTL.
 To see this, we show that the language $\Lang(A)$ is not universal if and only if there exists an initialised infinite path $\rho\in \infpath{T}^{\cM(A)(s_0)}$ such that $\cM(A),\rho\models (\neg \literal{test})\until (\literal{test} \land \B_B^{\leq 0} (\literal{finished}))$.

($\Longrightarrow$) If the language is not universal, then there exists a word $w = l_1...l_k$ such that no execution of $A$ on $w$ finishes in a final state. From $l_1...l_k$, we can construct a finite path $\rho_1=s_0(l_0,1)\sep(l_1,u)\sep(l_1,n)\sep(l_1,1)...(l_{k-1},u)\sep(l_{k-1},n)\sep(l_{k-1},1)\sep(l_k,u)(l_k,t)$ in $\cM(A)$, but not able to find a path $\rho_2=(l,q_0)(l_1,q_1^u)\sep(l_1,q_1^n)\sep(l_1,q_1)... (l_{k-1},q_{k-1}^u)\sep(l_{k-1},q_{k-1}^n)\sep(l_{k-1},q_{k-1})\sep(l_k,q_k^u)(l_k,q_k^{tf})$.
To see this, we note that each block $(l_x,q_x^u)(l_x,q_x^n)(l_x,q_x)$ for $1\leq x\leq k-1$ represents a piece of execution that moves from a state $(l_{x-1},q_{x-1})$ to another state $(l_{x},q_{x})$, such that $q_x\in\delta(q_{x-1},l_x)$, without conducting a test. At the state $(l_k,q_k^u)$, a decision is made to conduct  a test. However, because the word $w$ is not a word of $A$, the test will not move into the state $(l_k,q_k^{tf})$. Otherwise, it contradicts with the construction of $\cM(A)$ that only a final state can move into a state of the form  $(l,q^{tf})$.

Because there does not exist such a path $\rho_2$, we have that $\cM(A),\rho_1\models \B_B^{\leq 0}\literal{finished}$. Moreover, for all the states $s$ before $(l_k,q_k^t)$ on the paths $\rho_1$, we have that $\cM(A),s\models \neg \literal{test}$, and for the state $(l_k,q_k^t)$, we have that $\cM,(l_k,q_k^t)\models \literal{test}$. Further, because agent $B$ can observe the test action, the check of its beliefs will only concern those paths $\rho_1'$ on which the test action is only taken in the last state. Therefore, the probability, under the condition of agent $B$'s observations, will not leak into those states of the form $(l,q^f)$, and we have $\cM(A),\rho_1 \models \neg \literal{test}\until (\literal{test} \land \B_B^{\leq 0} \literal{finished})$. Because the until operator assumes the finiteness of the path, the latter means that $\cM(A),\rho_1\delta \models \neg \literal{test}\until (\literal{test} \land \B_B^{\leq 0} \literal{finished})$ for any infinite paths $\delta$ that is consistent with the path $\rho_1$.

($\Longleftarrow$) If $\cM(A),\rho\models (\neg \literal{test})\until (\literal{test} \land \B_B^{\leq 0} (\literal{finished}))$ for some infinite path $\rho$, then there must exist a finite path $\rho_1$ such that $\cM(A),\rho_1\models \B_B^{\leq 0} (\literal{finished})$ and there is a first test occurs on the state $last(\rho_1)$. Now we define function $w(\rho)$ as that $w(\rho_1\rho_2)=w(\rho_1)w(\rho_2)$ and $w((l,q))=l$. Intuitively, it abstracts the alphabets from the paths. Therefore, we have that $w(\rho_1)$ is a word of $A$  that cannot reach a final state.

Finally, we also need to show that the system $\cM(A)$ satisfies the constraints. Note that, the expanded system has two kinds of SCCs. The first kind of SCCs include transitions on which no test action is conducted, and the second kind of SCCs include a self loop on those states with a test action. For both the kinds of SCCs, agent $B$'s belief about the atomic proposition $\literal{finished}$ is kept the same: for the first kind of SCCs, it is always the case that $ \B_B^{\leq 0} (\literal{finished})$ while for the second kind of SCCs, it can be $ \B_B^{\leq 0} (\literal{finished})$ or $ \B_B^{> 0} (\literal{finished})$, depending on whether the current path can be abstracted into a word of $A$, and the belief values will not change.

Not that, because the formula $ \literal{test} \release (\neg \literal{test} \lor \B_B^{<1} (\neg \literal{finished}))$ is constant with respect to different instances of automata $A$, the PSPACE lower bound is also a lower bound of program complexity.
Moreover, the formula complexity, which measures the complexity of problem by the size of the formula by assuming that the system is constant, is also PSPACE-complete: the lower bound comes from the LTL model checking and its upper bound can be  derived directly from the combined complexity.
Putting the above together, the problem is PSPACE-complete for combined complexity, program complexity, and formula complexity.



\section{A Polynomial Time Fragment}\label{sec:ptime}

A stochastic automaton is a tuple $SA=(Q,A,\alpha,PI)$ where $Q$ is a set of states, $A$ is a set of symbols, $\alpha:Q\times A\times Q\rightarrow [0,1]$ is a transition function such that for every $q\in Q$ we have $\sum_{a\in A}\sum_{q'\in Q}\alpha(q,a,q')=1$, and $PI:Q\rightarrow [0,1]$ is an initial distribution over $Q$ such that $\sum_{q\in Q}PI(q)=1$. Note that, the stochastic automata defined here are different with the probabilistic automata in Section~\ref{sec:general}: a probabilistic automaton requires that for every $q\in Q$ and $a\in A$ we have $\sum_{q'\in Q}\cM_a(q,q')=1$.
Given a sequence $\rho$ of symbols, the acceptance probability of $\rho$ in a stochastic automaton $SA$ is
$$
\sum_{q\in Q}(PI(q)\times \sum_{q'\in Q} \alpha(q,\rho,q'))
$$
where $\alpha(q,\rho,q')$ is defined as follows: $\alpha(q,\rho a,q') = \sum_{q''\in Q} \alpha(q,\rho,q'')\times \alpha(q'',a,q')$ and  $\alpha(q,\epsilon,q) = 1$ for $\epsilon$ an empty sequence of symbols.
Two stochastic automata are equivalent if for each string $\rho\in A^*$, they accept $\rho$ with equal probability \cite{Tzeng1992}.

Given a stochastic automaton $SA=(Q,A,\alpha,PI)$, a pair $D=(Q_D,\alpha_D)$ is a component of $SA$ if $Q_D\subseteq Q$ and for all $q,q'\in Q_D$ and $a\in A$, we have $\alpha_D(q,a,q')=\alpha(q,a,q')$. We say that $D$ is a strongly connected component (SCC) if for all $q,q'\in Q_D$, there exists a sequence $\rho$ of symbols such that $\alpha_D(q,\rho,q')>0$, where the computation of $\alpha_D(q,\rho,q')$ can be done recursively as follows: $\alpha_D(q,\rho a,q') = \sum_{q''\in Q_D} \alpha_D(q,\rho,q'')\times \alpha_D(q'',a,q')$ and  $\alpha_D(q,\epsilon,q) = 1$ for $\epsilon$ an empty sequence of symbols. An SCC is closed if for all $q\in Q_D$ we have $\sum_{a\in A}\sum_{q\in Q_D} \alpha_D(q,a,q') = 1$.

\paragraph{Constraints}

Assume that we are given a system $\cM$ and a formula $\always(\psi\Rightarrow \prob{\bowtie q}{\eventually\B_A^{\geq 1}\psi})$ or $\always(\psi\Rightarrow\prob{\bowtie q}{\eventually\DT_{A,B}^{\geq 1}\psi})$ such that
\begin{itemize}
\item $\psi$ does not contain any belief or trust formulas,
\item $\psi$ is of CTL-style, i.e., there is no direct nesting of temporal operators, and
\item the system $\cM$ satisfies that $\cM\models \always (\psi\Rightarrow \always \psi)$, i.e., once formula $\psi$ holds, it will hold since then.
\end{itemize}
We consider their negated specification formulas $\eventually(\psi\land\prob{\widehat{\bowtie} 1-d}{\always\B_A^{> 0}\neg \psi})$ or $\eventually(\psi\land\prob{\widehat{\bowtie} 1-d}{\always\DT_{A,B}^{> 0}\neg \psi})$, and the algorithm, to be given below, is to determine the existence of a witness for them. We make a remark here that, the polynomial time complexity of this fragment relies on the Assumption~\ref{assump:detbeh} that the transition function $T$ of the system $\cM$ is deterministic, i.e., for every state $s\in S$ there is a unique $a\in \Act{}$ such that $\sum_{s'\in S}T(s,a,s')=1$. If this restriction is relaxed by allowing the possibility of more than one actions $a\in\Act{}$ with $\sum_{s'\in S}T(s,a,s')=1$, it is expected that the complexity could be higher.

\paragraph{Algorithm}

Fist of all, we need to compute $\probs{\cM,s}{}(\neg \psi)$ or  $\inf_{x\in \lintn{B}(s)} \probs{\cM,B.i(s,x)}{}(\neg \psi)$ for every state $s$ and formula $\neg \psi$, depending on which operator, $\B_A^{> 0}$ or $\DT_{A,B}^{> 0}$, the formula $\neg \psi$ is in the scope of. Recall that $\probs{\cM,s}{}(\psi)$ is the probability of  satisfying $\psi$ among all the temporal paths starting from state $s$. Because of the restrictions on $\psi$, these can be computed in polynomial time, using a standard \PCTL\ model checking procedure on DTMCs.
We can define $\probs{\cM,\rho}{}(\psi)$ for $\rho$ a path as $\probs{\cM,last(\rho)}{}(\psi)$, and have the following proposition. 
\begin{proposition}
For a formula $\psi$ without any belief or trust operator, we have that $\probs{\cM,\rho}{}(\psi)=\probs{\cM,\rho'}{}(\psi)$ and  $\inf_{x\in \lintn{B}(last(\rho))}\probs{\cM,\rho}{}(\psi)=\sep\inf_{x\in \lintn{B}(last(\rho'))}\probs{\cM,\rho'}{}(\psi)$ whenever $last(\rho)= last(\rho')$.
\end{proposition}

Because belief formulas and trust formulas can be handled in a similar way, in the following we only work with belief formulas. The  technique can be adapted to work with trust formulas.
From the system $\cM$, we can construct a system $\cM^\# =(S\times S,T^\#, \sinit^\#, L^\#)$ such that
\begin{itemize}
\item $\sinit^\#=\{(s_1,s_2)~|~s_1,s_2\in \sinit, \obs{A}(s_1)=\obs{A}(s_2)\}$,
\item $T^\#((s_1,s_2),(a_1,a_2),(s_1',s_2'))=(T(s_1,a_1,s_1'),T(s_2,a_2,s_2'))$ if $\obs{A}(s_1')=\sep\obs{A}(s_2')$, $\probs{\cM,s_2'}{}(\psi)=0$, $T(s_1,a_1,s_1')>0$, and $T(s_2,a_2,s_2')>0$, and
\item $ L^\#((s_1,s_2))=( L(s_1), L(s_2))$.
\end{itemize}
Intuitively, $\cM^\#$ is the production of two copies of the system $\cM$ such that the two copies preserve the same observations along the transitions  (by letting $\obs{A}(s_1)=\obs{A}(s_2)$ in $\sinit^\#$ and $\obs{A}(s_1')=\obs{A}(s_2')$ in $T^\#$) and the second copy always selects those states which do not satisfy the formula $\psi$ (by letting $\probs{\cM,s_2'}{}(\psi)=0$).
On such a system $\cM^\#$, we define operator $ L_k$ for $k\in \{1,2\}$ such that $ L_k(\cM^\#)=(S\times S,T^\#_k,S_{\mathrm{init},k}^\#, L^\#_k)$ where
\begin{itemize}
\item $s_k \in S_{\mathrm{init},k}^\#$ if $(s_1,s_2)\in \sinit^\#$ and $k\in \{1,2\}$. 
\item $T^\#_k((s_1,s_2),(a_1,a_2),(s_1',s_2'))$ is the $k$th component of $T^\#((s_1,s_2),\sep (a_1,a_2),\sep (s_1',s_2'))$, and
\item $ L^\#_k((s_1,s_2))$ is the $k$th component of $ L^\#((s_1,s_2))$.
\end{itemize}
Intuitively, in $ L_k(\cM^\#)$, we only consider those initial and transition probabilities of the $k$th copy of the system $\cM$.

Let $Q^\#_D\subseteq \states\times \states$ be a set of states. A component of $\cM^\#$ is a pair $D^\#=(Q^\#_D,T^\#_D)$. We define $T^\#_{k,D}$ as the $k$th component of $T^\#_D$ for $k\in \{1,2\}$.

\begin{definition}
We say that the component $D^\#$ is an SCC if either $(Q^\#_D,T^\#_{1,D})$ or $(Q^\#_D,T^\#_{2,D})$ is an SCC, and is a double-closed SCC if both $(Q^\#_D,T^\#_{1,D})$ and $(Q^\#_D,T^\#_{2,D})$ are closed SCCs. 
\end{definition}
By the construction of $\cM^\#$, we have
\begin{proposition}
$(Q^\#_D,T^\#_{1,D})$ is an SCC if and only if $(Q^\#_D,T^\#_{2,D})$ is an SCC.
\end{proposition}
Moreover, given $(Q^\#_D,T^\#_{k,D})$, we can compute a stationary distribution for each of $(Q^\#_D,T^\#_{1,D})$ and $(Q^\#_D,T^\#_{2,D})$.  Such a stationary distribution always exists  and uniquely exists for a finite $(Q^\#_D,T^\#_{k,D})$, which is a DTMC. Let  $\mu^\#_{k,D}$ be the stationary distribution of $(Q^\#_D,T^\#_{k,D})$.
Then, from a component $D^\#=(Q^\#_D,T^\#_D)$, we can construct two stochastic automata $SA_1(D^\#)=(Q^\#_D,A,T^\#_{1,D},\mu^\#_{1,D})$ and $SA_2(D^\#)=(Q^\#_D,A,T^\#_{2,D},\mu^\#_{2,D})$ where $\mu^\#_{1,D}$ and $\mu^\#_{2,D}$ are their initial distributions. 
\begin{definition}
We say that $D^\#$ is internal equivalent when two automata $SA_1(D^\#)$ and $SA_2(D^\#)$ are equivalent. A component $D^\#=(Q^\#_D,T^\#_D)$ is formula-specific with respect to the formula $\psi$ if $\probs{\cM,s}{}(\psi)>0$ for some $(s,t)\in Q^\#_D$ and $t\in Q$. 
\end{definition}
Recall that, two stochastic automata are equivalent if they have the same acceptance probability for all strings.
Note that, in the construction of $\cM^\#$, we require that the second copy always selects those states which do not satisfy the formula $\psi$. Therefore,
a formula-specific component suggests that
the component contains inconsistency between two copies of the system $\cM$ with respect to the formula $\psi$.

Let $\apath{}(D^\#)$ be the set of paths in $D^\#$, and $\infpath{}(D^\#)$ and $\fpath{}(D^\#)$ be the set of infinite and finite paths in $D^\#$, respectively. We also extend the transition relations $T^\#_{1,D}$ and $T^\#_{2,D}$ to work with paths in $\apath{}(D^\#)$, similar as that of $\alpha$ and $\alpha_D$ in stochastic automata. We do not define $\obs{A}^\#$ because the two copies of the system have the same observations.


Let $\rho$ be a finite path such that $last(\rho)$ is in an SCC $D^\#$. Given a path $\rho_1$ in the SCC $D^\#$, we let $\obs{A}(\rho_1)=\{\rho\in \apath{}(D^\#)~|~\obs{A}(\rho)=\obs{A}(\rho_1)\}$ be the set of paths that agent $A$ cannot distinguish. Then we write
$$
g(\rho,D^\#,k,\epsilon)\equiv T^\#_{1,D}(\{\rho_1\in \fpath{}(D^\#)~|~\frac{T^\#_{2,D}(\obs{A}(\rho\rho_1))}{T^\#_{1,D}(\obs{A}(\rho\rho_1))}> \epsilon, |\rho_1| = k\})
$$
for the probability of  infinite continuations of $\rho$ whose probability ratio between the two copies of the system is maintained (i.e., greater than $\epsilon$). Moreover, we let $g(\rho,D^\#)=1 $, if for all number $k\in \nat$, there exists $\epsilon>0$  such that $g(\rho,D^\#,k_1,\epsilon)=1$, and $g(\rho,D^\#)=0$, otherwise. We have the following propositions.
\begin{proposition}\label{propn:internal}
If $D^\#$ is not internal equivalent then $g(\rho,D^\#)=0$ for all finite paths $\rho$ whose last state is on $D^\#$.
\end{proposition}
\begin{proof}
Let $\rho_1$ be the shortest path on which the two copies of the system have different acceptance probabilities. By the construction of $\cM^\#$, some transitions are removed due to the restriction of $\probs{\cM,s_2}{}(\psi)=0$. Therefore, the probability of the second copy is less than the one of the first copy, and we have $\displaystyle\frac{T^\#_{2,D}(\obs{A}(\rho\rho_1))}{T^\#_{1,D}(\obs{A}(\rho\rho_1))} < 1$.
Then we can construct a path $\rho_2$ which contains a large number of pieces of observations that are equivalent to $\rho_1$.
With this path, we have that, for all $\epsilon>0$, $\displaystyle\frac{T^\#_{2,D}(\obs{A}(\rho\rho_2))}{T^\#_{1,D}(\obs{A}(\rho\rho_2))} < \epsilon$.
Therefore, we have that $g(\rho,D^\#)=0$.
\end{proof}

\begin{proposition}\label{propn:twoconditions}
$g(\rho,D^\#)=1$ if and only if  $D^\#$ is double-closed and internal equivalent.
\end{proposition}
\begin{proof}
($\Leftarrow$) Assume that $D^\#$ is double-closed and internal equivalent. By the definition of internal equivalence, for all finite paths $\rho_1\in \fpath{}(D^\#)$, we have that $T_{2,D}^\#(\{\rho_2\in \fpath{}(D^\#)~|~\obs{A}(\rho_2)=\obs{A}(\rho_1)\})=T_{1,D}^\#(\{\rho_2\in \fpath{}(D^\#)~|~\obs{A}(\rho_2)=\obs{A}(\rho_1)\})$, which means that $T^\#_{2,D}(\obs{A}(\rho\rho_1))=T^\#_{1,D}(\obs{A}(\rho\rho_1))>0$. Because
$D^\#$ is double-closed, we have that $g(\rho,D^\#)=1$.

($\Rightarrow$) Assume that $g(\rho,D^\#)=1$. We show that both conditions are indispensable. If $D^\#$ is not double-closed and it is the first copy that is not closed, then the probability of $T_{1,D}^\#$ will flow out of the SCC. Therefore, $g(\rho,D^\#,k,\epsilon)=0$ for a sufficiently large number $k$. If $D^\#$ is not double-closed and it is the second copy that is not closed, then $\displaystyle\frac{T^\#_{2,D}(\obs{A}(\rho\rho_1))}{T^\#_{1,D}(\obs{A}(\rho\rho_1))}=0$ for a sufficiently large $|\rho_1|$, which means that $g(\rho,D^\#)=0$. Moreover, the case in which $D^\#$ is double-closed but not internal equivalent is guaranteed by Proposition~\ref{propn:internal}.
\end{proof}

\begin{proposition}
Let $\rho$ be a finite path such that $last(\rho)$ is on an SCC $D^\#$, and $D^\#$ is double-closed, internal equivalent, and formula-specific with respect to $\psi$. Then $g(\rho,D^\#)=1$ if and only if for all $\rho_1\in \apath{}(D^\#)$, $(\sum_{\obs{A}(\rho') = \obs{A}(\rho\rho_1)}\be_A(\rho')\times \probs{\cM,\rho'}{}(\neg\psi)) > 0 $.
\end{proposition}
\begin{proof}
This is equivalent to proving that 
\[
\sum_{\obs{A}(\rho') = \obs{A}(\rho\rho_1)}\be_A(\rho')\times \probs{\cM,\rho'}{}(\neg\psi)>0 \text{ iff } \frac{T^\#_{2,D}(\obs{A}(\rho\rho_1))}{T^\#_{1,D}(\obs{A}(\rho\rho_1))}>0,
\]
for $\rho_1\in \apath{}(D^\#)$.

Now, by the construction of $\cM^\#$ (in which the second copy satisfies $\probs{\cM,s_2}{}(\psi)=0$) and the three conditions of $D^\#$, we have that 
\[
\sum_{\obs{A}(\rho') = \obs{A}(\rho\rho_1)}\be_A(\rho')\times (\probs{\cM,\rho'}{}(\neg\psi)>0)=\frac{T^\#_{2,D}(\obs{A}(\rho\rho_1))}{T^\#_{1,D}(\obs{A}(\rho\rho_1)}.
\]

The equivalence of 
\[
\sum_{\obs{A}(\rho') = \obs{A}(\rho\rho_1)}\be_A(\rho')\times (\probs{\cM,\rho'}{}(\neg\psi)>0)>0
\]
and 
\[
\sum_{\obs{A}(\rho') = \obs{A}(\rho\rho_1)}\be_A(\rho')\times \probs{\cM,\rho'}{}(\neg\psi)>0
\]
can be seen by their structures.
\end{proof}

With these propositions, we have the following algorithm.

\begin{algorithm}
After the computation of  $\probs{\cM,s}{}(\neg \psi)$ or $\mathrm{inf}_{x\in \lintn{B}(s)} \probs{\cM,B.i(s,x)}{}(\neg \psi)$ for every state $s$ and formula $\neg \psi$,
the algorithm proceeds  by the following sequential steps:
\begin{enumerate}
\item compute the set of SCCs of $\cM^\#$ such that $D^\#$ is formula-specific with respect to $\psi$, double-closed, and internal equivalent,
\item  for every state $s$ satisfying $\psi$, we do the following:
\begin{enumerate}
\item compute the reachability probability to those SCCs from the first step; let the probability value be $p$, and
\item check whether $p \widehat{\bowtie} 1- q$.
\end{enumerate}
If there exists a state $s$ that can satisfy the above computation, then the specification formula does not hold, otherwise holds.
\end{enumerate}
\end{algorithm}

\paragraph{Analysis of Complexity} The complexity of the algorithm is in polynomial time. First of all, the computation of  $\probs{\cM,s}{}(\neg \psi)$ or $inf_{x\in \lintn{B}(s)} \probs{\cM, B.i(s, x)}{}(\neg \psi)$ can be done in polynomial time because $\psi$ is of CTL-style and contains no belief or trust formula. Second, for the first step of the algorithm, the computation of all SCCs of $\cM^\#$, whose size is quadratic with respect to $\cM$, can be done in polynomial time by Tarjan's algorithm, and the checking of the three conditions (i.e., formula-specific with respect to $\psi$, double-closed, and internal equivalent) can be done in polynomial time. In particular, the checking of internal-equivalent of an SCC can be done in polynomial time is a result of the existence of a polynomial time algorithm for the equivalence of stochastic automata~\cite{Tzeng1992}. Third, for the second step of the algorithm, the computation of reachability probability on $\cM^\#$ can be done in polynomial time, and the comparison of values can be done in constant time.

\end{document}